%% file: QCLT-arXiv-v2.tex
\documentclass[11pt,notitlepage,tightenlines,nofootinbib]{revtex4-1}

\usepackage{graphicx}
\input{def}

\makeatletter
\def\thmhead@plain#1#2#3{%
  \thmname{#1}\thmnumber{\@ifnotempty{#1}{ }\@upn{#2}}%
  \thmnote{ {\the\thm@notefont#3}}}
\let\thmhead\thmhead@plain
\makeatother

\begin{document}

\title{Convergence rates for the quantum central limit theorem}

\author{Simon Becker}
\email{simon.becker@damtp.cam.ac.uk}
\affiliation{Department of Applied Mathematics and Theoretical Physics, Centre for Mathematical Sciences University of Cambridge, Cambridge CB3 0WA, United Kingdom}

\author{Nilanjana Datta}
\email{n.datta@damtp.cam.ac.uk}
\affiliation{Department of Applied Mathematics and Theoretical Physics, Centre for Mathematical Sciences University of Cambridge, Cambridge CB3 0WA, United Kingdom}

\author{Ludovico Lami}
\email{ludovico.lami@gmail.com}
\affiliation{School of Mathematical Sciences and Centre for the Mathematics and Theoretical Physics of Quantum Non-Equilibrium Systems, University of Nottingham, University Park, Nottingham NG7 2RD, United Kingdom}
\affiliation{Institute of Theoretical Physics and IQST, Universit\"{a}t Ulm, Albert-Einstein-Allee 11D-89069 Ulm, Germany}

\author{Cambyse Rouz\'{e}}
\email{rouzecambyse@gmail.com}
\affiliation{Zentrum Mathematik, Technische Universit\"{a}t M\"{u}nchen, 85748 Garching, Germany}

\date{December 2019}

\begin{abstract}
Various quantum analogues of the central limit theorem, which is one of the cornerstones of probability theory, are known in the literature. One such analogue, due to Cushen and Hudson, is of particular relevance for quantum optics. It implies that the state in any single output arm of an $n$-splitter, which is fed with $n$ copies of a centred state $\rho$ with finite second moments, converges to the Gaussian state with the same first and second moments as $\rho$. Here we exploit the phase space formalism to carry out a refined analysis of the rate of convergence in this quantum central limit theorem. For instance, we prove that the convergence takes place at a rate $\mathcal{O}\left(n^{-1/2}\right)$ in the Hilbert--Schmidt norm whenever the third moments of $\rho$ are finite. Trace norm or relative entropy bounds can be obtained by leveraging the energy boundedness of the state. Via analytical and numerical examples we show that our results are tight in many respects. An extension of our proof techniques to the non-i.i.d.\ setting is used to analyse a new model of a lossy optical fibre, where a given $m$-mode state enters a cascade of $n$ beam splitters of equal transmissivities $\lambda^{1/n}$ fed with an arbitrary (but fixed) environment state. Assuming that the latter has finite third moments, and ignoring unitaries, we show that the effective channel converges in diamond norm to a simple thermal attenuator, with a rate $\mathcal{O}\Big(n^{-\frac{1}{2(m+1)}}\Big)$. This allows us to establish bounds on the classical and quantum capacities of the cascade channel. Along the way, we derive several results that may be of independent interest. For example, we prove that any quantum characteristic function $\chi_\rho$ is uniformly bounded by some $\eta_\rho<1$ outside of any neighbourhood of the origin; also, $\eta_\rho$ can be made to depend only on the energy of the state $\rho$.
\end{abstract}

\maketitle

\tableofcontents

\section{Introduction} \label{sec:introduction}

The Central Limit Theorem (CLT) is one of the cornerstones of probability theory. This theorem and its various extensions have found numerous applications in diverse fields including mathematics, physics, information theory, economics and psychology. Any limit theorem becomes more valuable if it is accompanied by estimates for rates of convergence. The Berry--Esseen theorem (see e.g.~\cite{Feller}), which gives the rate of convergence of the distribution of the scaled sum of independent and identically distributed (i.i.d.)\ random variables to a normal distribution, thus provides an important refinement of the CLT.

\medskip
The first results on quantum analogues of the CLT were obtained in the early 1970s by Cushen and Hudson~\cite{Cushen1971}, and Hepp and Lieb~\cite{HL73-1, HL73-2}. The approach of~\cite{HL73-1} was generalised by Giri and von Waldenfels~\cite{GvW78} a few years later. These papers were followed by numerous other quantum versions of the CLT in the context of quantum statistical mechanics~\cite{estat-1, estat-1bis, estat-2, ne-stat-1, ne-stat-2, ne-stat-3,Arous2013,fern2015equivalence,Brandao-Gour}, quantum field theory~\cite{qft-1, qft-2, qft-3}, von Neumann algebras \cite{Goderis1989,Jak10}, free probability~\cite{free}, noncommutative stochastic processes~\cite{accardi1994quantum} and quantum information theory~\cite{qit-2, qit-1, Campbell2013}. For a more detailed list of papers on noncommutative or quantum central limit theorems (QCLT), see for example \cite{Jak10,Lenczewski1995} and references therein. A partially quantitative central limit theorem for unsharp measurements has been obtained in \cite{DD}.

\medskip

\medskip
An important pair of non-commuting observables is the pair $(x,p)$ of canonically conjugate operators, which obey Heisenberg's canonical commutation relations (CCR) $[x,p] = i I$, where $I$ denotes the identity operator.\footnote{Throughout this paper we set $\hbar = 1$.} These observables could be, for example, the position and momentum operators of a quantum particle, or the so-called position and momentum quadratures of a single-mode bosonic field, described in the quantum mechanical picture by the Hilbert space $\cH_1 \coloneqq L^2(\RR)$ -- the space of square integrable functions on $\RR$. The corresponding annihilation and creation operators are constructed as $a\coloneqq (x+ip)/\sqrt2$ and $a^\dag \coloneqq (x-ip)/\sqrt2$. When expressed in terms of $a,a^\dag$, the CCR take the form $[a,a^\dag]=I$. 

\medskip
Quantum states are represented by density operators, i.e.\ positive semi-definite trace class operators with unit trace. A state $\rho$ of a continuous variable quantum system is uniquely identified by its characteristic function, defined for all $z \in \CC$ by $\chi_\rho(z) \coloneqq \tr\big[\rho\, e^{z a^\dagger - z^* a}\big]$. The special class of Gaussian states comprises all quantum states whose characteristic function is the (classical) characteristic function of a normal random variable on $\CC$.\footnote{The characteristic function of a complex-valued random variable $X$ is defined by $\chi_X(z)\coloneqq \E\big[ e^{zX^* - z^* X}\big]$.} Exactly as in the classical case, a quantum Gaussian state is uniquely defined by its mean and covariance matrix.

\medskip
Cushen and Hudson~\cite{Cushen1971} proved a quantum CLT for a sequence of pairs of such canonically conjugate operators $\{(x_n, p_n): n=1,2,\ldots\}$, with each pair acting on a distinct copy of the Hilbert space $\cH_1$. More precisely, they showed that sequences that are stochastically independent and identically distributed, and have finite covariance matrix and zero mean with respect to a quantum state $\rho$ (given by a density operator on $\cH_1$), are such that their scaled sums converge in distribution to a normal limit distribution~\cite[Theorem~1]{Cushen1971}.

\medskip
Their result admits a physical interpretation in terms of a passive quantum optical element known as the $n$-splitter. This can be thought of as the unitary operator $U_{\text{$n$-split}}$ that acts on $n$ annihilation operators of $n$ independent optical modes as $U_{\text{$n$-split}}\, a_j\, U_{\text{$n$-split}}^\dag = \sum_{k} F_{jk} a_k$, where $F_{jk}\coloneqq e^{\frac{2\pi jk}{n}\, i}$ is the discrete Fourier transform matrix. Passivity here means that $U_{\text{$n$-split}}$ commutes with the canonical Hamiltonian of the field, i.e.\ $\left[ U_{\text{$n$-split}},\, \sumno_j a_j^\dag a_j\right]=0$. When $n$ identical copies of a state $\rho$ are combined by means of an $n$-splitter, and all but the first output modes are traced away, the resulting output state is called the {\em{$n$-fold quantum convolution}} of $\rho$, and denoted by $\rho^{\boxplus n}$. This nomenclature is justified by the fact that the characteristic function $\chi_{\rho\, \boxplus\, \sigma}$ of two states $\rho$ and $\sigma$ is equal to the product of the characteristic functions of $\rho$ and $\sigma$, a relation analogous to that satisfied by characteristic functions of convolutions of classical random variables. Observe state $\rho^{\boxplus n}$ can also be obtained as the output of a cascade of $n-1$ beam splitters with suitably tuned transmissivities $\lambda_j = j/(j+1)$ for $j=1,2, \ldots n-1$ (see Figure~\ref{fig1}).

\medskip
Cushen and Hudson's result is that if $\rho$ is a centred state (i.e.~with zero mean) 
and has finite second moments, its convolutions $\rho^{\boxplus n}$ converge to the Gaussian state $\rho_\G$ with the same first and second moments as $\rho$ in the limit $n\to\infty$ (Theorem~\ref{thm:CushHud}). In~\cite[Theorem~1]{Cushen1971}, the convergence is with respect to the weak topology of the Banach space of trace class operators, which translates to pointwise convergence of the corresponding characteristic functions, by a quantum analogue of Levy's lemma that is also proven in~\cite{Cushen1971}. This in turn implies that the convergence actually is with respect to the strong topology, i.e.\ in trace norm (see~\cite{Davies1969}, or~\cite[Lemma~4]{G-dilatable}).

\medskip
In this paper, we focus on the framework proposed by Cushen and Hudson, and provide a refinement of their result by deriving estimates for the associated rates of convergence. We consider a quantum system composed of $m$ modes of the electromagnectic field, each modelled by an independent quantum harmonic oscillator, so that the corresponding Hilbert space becomes $\cH_1^{\otimes m} = \LL^2\left(\RR^m\right)$. The main contribution of this paper consists of estimates on rate of convergence of $\rho^{\boxplus n}$ to the `Gaussification' $\rho_\G$ of $\rho$, obtained under suitable assumptions on $\rho$ -- typically, the finiteness of higher-order moments. In analogy with the classical case, we refer to our Theorems~\ref{thm:QBE'} and~\ref{thm:QBElow} as \emp{quantum Berry--Esseen theorems}. Our estimates are given in the form of bounds on the Schatten $p$-norms (for $p=1$ and $2$) of the difference $(\rho^{\boxplus n} - \rho_\G)$ in the limit of large $n$, as well as bounds on the relative entropy of $\rho^{\boxplus n}$ with respect to $\rho_\G$ in the same limit. 

\medskip
We also show that the assumption of finiteness of the second moments cannot be removed from the Cushen--Hudson theorem. Namely, we construct a simple example of a single-mode quantum state $\sigma$ such that $\tr\big[\sigma\, (a a^\dag)^{1-\delta}\big]$ is finite for all $\delta>0$ (and infinite for $\delta=0$), yet $\sigma^{\boxplus n}$ does not converge to any quantum state as $n\to \infty$.

\medskip
As an application, we propose and study a new model of optical fibre, represented as a cascade of $n$ beam splitters, each with transmissivity $\lambda^{1/n}$ and fed with a fixed environment state $\rho$, which is assumed to have bounded energy and thermal Gaussification. Such a model may be relevant to the mathematical modelisation of a channel running across an integrated optical circuit~\cite{Carolan2015, Rohde2015}. We are able to show that for $n\to \infty$ the cascade channel converges in diamond norm, up to irrelevant symplectic unitaries, to a thermal attenuator channel with transmissivity $\lambda$ and the same photon number as that of the environment state $\rho$. Furthermore, an extension of our results to the non-i.i.d.~setting allows us to bound the rate of convergence in terms of the diamond norm distance. Finally, combining existing continuity bounds on entropies and energy-constrained channel capacities~\cite{Holevo-energy-constrained, Mark-energy-constrained}, obtained by Winter~\cite{tightuniform, winter2017energy} and Shirokov~\cite{shirokov2017tight, shirokov2018energy}, with the known formulae expressing or estimating energy-constrained classical~\cite{Giovadd, Giovadd-CMP} and quantum~\cite{holwer, PLOB, MMMM, Rosati2018, Sharma2018, Noh2019, Noh2020} capacities of thermal attenuator channels, we derive bounds on the same capacities for the cascade channel.

\medskip
Finally, along the way we derive several novel results concerning quantum characteristic functions, which we believe to be of independent interest. First, we prove the simple yet remarkable fact that \textit{convolving any two quantum states (i.e.\ mixing them in a $50:50$ beam splitter) always results in a state with non-negative Wigner function (Lemma~\ref{positive W lemma}).} This allows us to interpret the quantum central limit theorem as a result on classical random variables, in turn enabling us to transfer techniques from classical probability theory to the quantum setting.
Secondly, we derive new decay bounds on the behaviour of the quantum characteristic function both at the origin and at infinity. For instance, we prove that for any $m$-mode quantum state $\rho$ and for any $\varepsilon>0$ there exists a constant $\eta=\eta(\rho,\varepsilon)<1$ such that $|\chi_\rho(z)|\leq \eta(\rho,\varepsilon)$ for all $z\in \CC^m$ with $|z|\geq \varepsilon$ (Proposition~\ref{prop1}). Moreover, we show that such a constant can be made to depend only on the second moments of the state, assuming they are finite (Proposition~\ref{prop:decay-energy}). As an explicit example, consider a single-mode state $\rho$ with mean energy $E$. We then prove that $|\chi_\rho(z)| \leq 1 - \frac{c}{E^{2}}$ for all $z$ with $|z| \geq \frac{C}{\sqrt{E}}$, where $c,C$ are universal constants. Note that any such bound must depend on the energy, as one can construct a sequence of highly squeezed Gaussian states for which the modulus of the characteristic function approaches one at any designated point in phase space (Example~\ref{ex:squeezed-state}).


\begin{figure}
\centering
\begin{subfigure}{\textwidth}
    \includegraphics[width=0.5\textwidth]{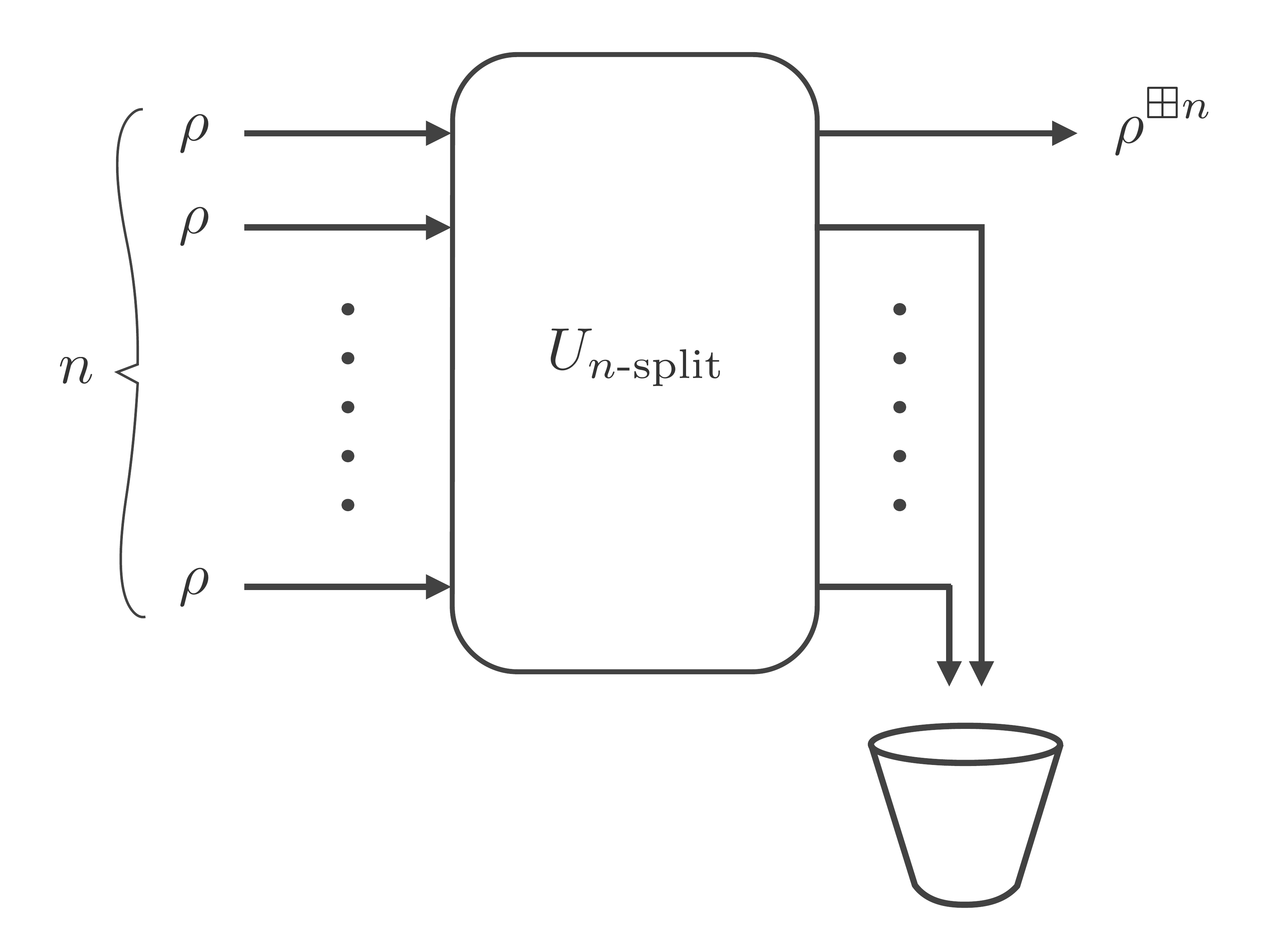}
    \caption{}
    \label{fig1}
\end{subfigure}
\begin{subfigure}{\textwidth}
    \includegraphics[width=0.9\textwidth]{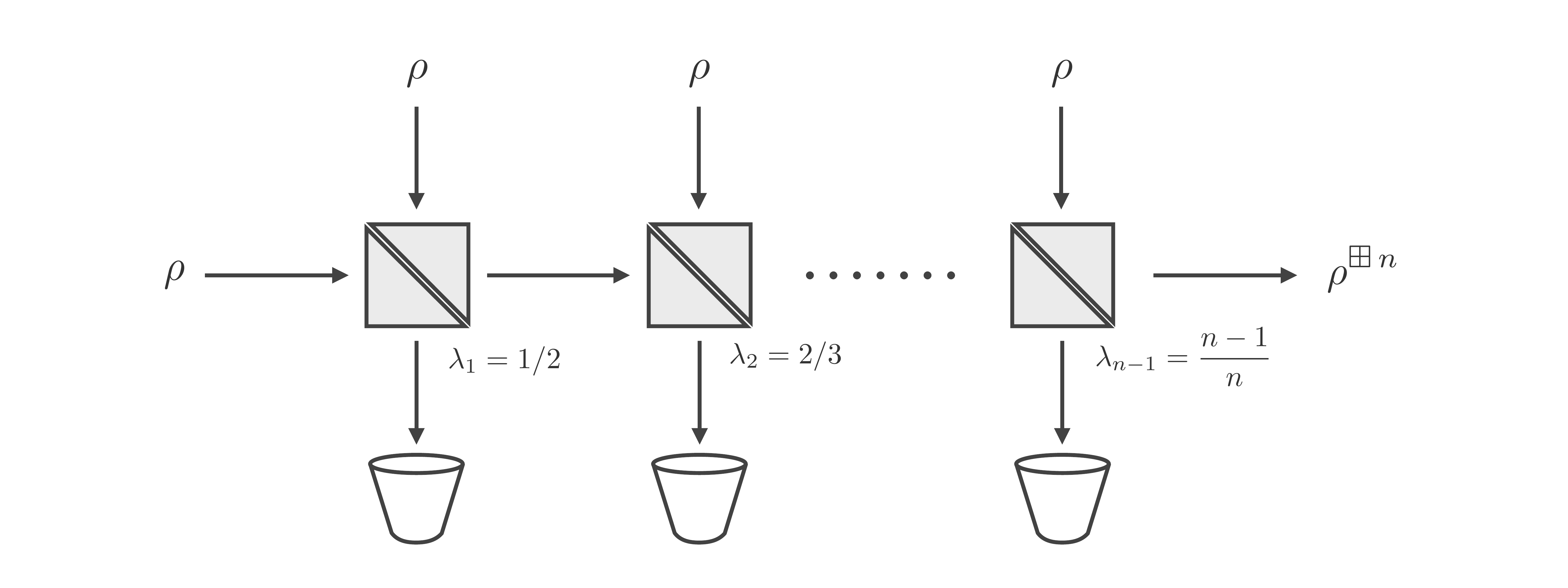}
    \caption{}
    \label{fig2}
\end{subfigure}
    \caption{The $n$-fold convolution $\rho^{\boxplus n}$ of a state $\rho$ can be realised by mixing $n$ copies of it either: (a) in an $n$-splitter; or (b) in a cascade of beam splitters with suitably tuned transmissivities.}
\end{figure}

\smallskip

\noindent
{\textbf{Layout of the paper:}} In Section~\ref{sec:notation} we introduce the notation and definitions used in the paper. In Section~\ref{sec:CH-CLT} we recall the Cushen and Hudson quantum central limit theorem.
Our main results are presented in Section~\ref{sec:main-results}. The rest of the paper is devoted to the proofs of these results. We start with the novel properties of quantum characteristic functions (Section~\ref{sec-proofs:new-results}), which lie at the heart of our approach. Then, in Section \ref{sec-proofs:QCLT} we prove our quantum Berry--Esseen theorems. Section~\ref{sec-proofs:optimality} is devoted to the discussion of the optimality and sharpness of our results. In Section \ref{sec-proofs:cascade} we apply our quantitative non-i.i.d.~extension of the Cushen--Hudson theorem to an optical fibre subject to non-Gaussian environment noise. The paper contains a technical appendix (Appendix \ref{app:moments}) that makes the connection between moments and the regularity of the quantum characteristic function and shows that our definition of moments induces a canonical family of interpolation spaces. 
\medskip

\section{Notation and Definitions} \label{sec:notation}
In this section, we fix the basic notations used in the paper, and introduce the necessary definitions.

\subsection{Mathematical notation} \label{subsec:notation}

Let ${\cH}$ denote a separable Hilbert space, and let $\cB(\cH)$ denote the set of bounded linear operators acting on $\cH$. 
Let $\cD(\cH)$ denote the set of {\it{quantum states}} of a system with Hilbert space $\cH$, that is the set of density operators $\rho$ (positive semi-definite, i.e.~$\rho \ge 0$, trace class operators\footnote{That is, operators $A \in \cB(\cH)$ for which $\|A\|_1 \coloneqq \tr|A| <\infty$.} with unit trace) acting on $\cH$. We denote by $\|\cdot\|_{\mathcal{T}_p(\cH)} \equiv \|\cdot\|_p$ the Schatten $p$-norm, defined as $\|X\|_{p}=\left(\tr|X|^p\right)^{1/p}$. The \emp{Schatten $p$-class} $\cT_p(\cH)$ is the Banach subspace of $\cB(\cH)$ formed by all bounded linear operators whose Schatten $p$-norm is finite. We shall hereafter refer to $\cT_1(\cH)$ as the set of \emp{trace class} operators, to the corresponding norm $\|\cdot\|_1$ as the {trace norm}, and to the induced distance (e.g.\ between quantum states) as the {trace distance}. The case $p=2$ is also special, as the norm $\|\cdot\|_2$ coincides with the \emp{Hilbert--Schmidt norm}.

Let $A,B$ be positive semi-definite operators defined on some domains $\dom(A),\dom(B) \subseteq \cH$. According to \cite[Definition~10.15]{schmuedgen}, we write that $A\geq B$ if and only if $\dom\left(A^{1/2}\right)\subseteq \dom\left(B^{1/2}\right)$ and $\left\|A^{1/2}\ket{\psi}\right\|^2\geq \left\|B^{1/2}\ket{\psi}\right\|^2$ for all $\ket{\psi}\in \dom\left(A^{1/2}\right)$. Now, let $A$ be a positive semi-definite operator, and let $\rho$ be a quantum state with spectral decomposition $\rho=\sum_i p_i \ketbra{e_i}$. We define the \emp{expected value} of $A$ on $\rho$ as
\bb
\tr[\rho A]\coloneqq \sum_{i:\, p_i>0} p_i \left\|A^{1/2}\ket{e_i}\right\|^2 \in \RR_+\cup \{+\infty\}\, ,
\label{expected positive}
\ee
with the convention that $\tr[\rho A]=+\infty$ if the above series diverges or if there exists an index $i$ such that $p_i>0$ and $\ket{e_i}\notin \dom\left(A^{1/2}\right)$. To extend this definition to a generic densely defined self-adjoint operator $X$ on $\cH$, it is useful to consider its decomposition $X=X_+-X_-$ into positive and negative part \cite[Example~7.1]{schmuedgen}. We will say that $X$ \emp{has finite expected value on $\rho$} if $\ket{e_i}\in \dom\big(X_+^{1/2}\big)\cap \dom\big(X_-^{1/2}\big)$ for all $i$ such that $p_i>0$, and moreover the two series $\sum_i p_i \big\|X_\pm^{1/2} \ket{e_i}\big\|^2$ both converge. In this case, we call
\bb
\tr[\rho X]\coloneqq \sum_{i:\, p_i>0} p_i \left\|X_+^{1/2} \ket{e_i}\right\|^2 - \sum_{i:\, p_i>0} p_i \left\|X_-^{1/2} \ket{e_i}\right\|^2
\label{expected}
\ee
the \emp{expected value} of $X$ on $\rho$. Clearly, given two operators $A\geq B\geq 0$, we have that $\tr[\rho A]\geq \tr[\rho B]$.

\medskip
For two real sequences $\left(a_n(\lambda)\right)_n,\, \left(b_n(\lambda)\right)_n$ that depend on some parameter $\lambda$, we write $a_n(\lambda) = \mathcal{O}_\lambda \left(b_n(\lambda)\right)$ if there exists a constant $c_\lambda>0$ that only depends on $\lambda$ such that $|a_n(\lambda)|\leq c_\lambda |b_n(\lambda)|$ holds in the limit $n\to\infty$. We also write $a_n(\lambda) =\mathcal O_{\lambda}\left(b_n(\lambda)^{\infty}\right)$ if for every $N \in \NN$ we have that $a_n(\lambda) =\mathcal O_{\lambda}\left(b_n(\lambda)^{N}\right)$. 

\medskip
For an $n$-linear tensor $A:\times_{i=1}^n \mathbb C^m \rightarrow \mathbb C^k$, we write $A(x^{\times n})\coloneqq A\underbrace{(x,...,x)}_{\text{$n$ times}}$ if the vector we apply the tensor to is the same in every component. For functions $f$, we sometimes abuse the notation by denoting the norm of this function as $\Vert f(z) \Vert$ instead of $\Vert f \Vert.$ We denote with $*$ the entry-wise complex conjugation, with $\intercal$ the standard transposition of vectors, and with $\dag$ the combination of the two.

For partial derivatives with respect to complex variables $z,z^*$ we write $\partial_z$ and $\partial_{z^*}.$ 
Consider an $m$-dimensional multi-index $\alpha = (\alpha_1, \alpha_2,\ldots,\alpha_m)$ with $| \alpha | = \alpha_1 + \alpha_2 + \cdots + \alpha_m$. Then $\partial_{z}^\alpha \coloneqq \partial_{z_1}^{\alpha_1} \partial_{z_2}^{\alpha_2} \ldots \partial_{z_m}^{\alpha_m},$
and analogously for $z^*$. The total derivatives of order $k$ of a function $f:\mathbb C^{m} \rightarrow \mathbb C$ we denote by $D^k f\coloneqq \left( \partial^{\alpha}_z \partial^{\beta}_{z^*}f\right)_{\vert \alpha \vert+\vert \beta \vert =k}.$
We then recall the definition of the Fr\'echet derivative for functions $f:\mathbb{C}^m \rightarrow \mathbb{C}$ such that $D^k f: \mathbb{C}^m \rightarrow \cB(\underbrace{\mathbb{C}^m \times....\times \mathbb{C}^m}_{k \text{ times} },\mathbb{C})$ and therefore
\bb
D^k f(z)\left(v^{(1)},..,v^{(k)}\right) = \sum_{\vert \alpha\vert +\vert \beta \vert = k} \partial^{\alpha}_z\partial^{\beta}_{z^*}f(z) \left(\prod_{\ell=1}^{|\alpha|} v_{j_{\alpha}(\ell)}^{(\ell)}\right) \left( \prod_{\ell=|\alpha|+1}^k v_{j_\beta(\ell-|\alpha|)}^{(\ell)}\right)^*\, ,
\label{Frechet}
\ee
with $j_\alpha(\ell)\coloneqq \min\left\{j\in \{1,\ldots, m\}: \ell\leq \sumno_{j'=1}^j \alpha_{j'}\right\}$. Let $C_0(\CC^m)$ denote the space of continuous functions $f:\CC^m \to \CC$ that tend to zero as $|z| \rightarrow \infty$, where for $z\in \CC^m$ we set
\bb
|z|\coloneqq \sqrt{\sumno_{j=1}^m |z_j|^2}\, .
\label{Euclidean norm}
\ee
We write $C_c^{\infty}(\CC^m)$ to denote the space of smooth and compactly supported functions on $\CC^m$. For some open set $\Omega\subseteq \CC^m$ with closure $\widebar{\Omega}$, a function $f:\widebar{\Omega} \to \CC$, and a non-negative integer $k\in \NN_0$, we denote by $C^k(\widebar{\Omega})$ the space of functions for which the norm 
\bb
\|f\|_{C^k(\widebar{\Omega})} \coloneqq \max_{|\alpha|+|\beta| \leq k} \sup_{z\in\Omega} \left| \partial_z^\alpha \partial_{z^*}^\beta f (z) \right|
\label{Ck norm 1}
\ee
is finite. Here, $\alpha,\beta\in \NN_0^m$ are multi-indices. When $k\geq 0$ is \textit{not} an integer, we define instead
\bb
\|f\|_{C^k(\widebar{\Omega})} \coloneqq \|f\|_{C^{\floor{k}}(\widebar{\Omega})} + \max_{|\alpha|+|\beta| = \floor{k}} \sup_{\substack{z,w \in \Omega, \\ z\neq w }} \frac{\, \left| \partial_z^\alpha \partial_{z^*}^\beta f(z) - \partial_z^\alpha \partial_{z^*}^\beta f (w) \right| \, }{\|z-w\|^{k-\floor{k}}} .
\label{Ck norm 2}
\ee
This extension allows us to consider the normed spaces $C^k\big(\widebar{\Omega}\big)$ for all $k\geq 0$. Typically, we will deal with the case where $\Omega$ is bounded, so that  $C^k\big(\widebar{\Omega}\big)$ is in fact a Banach space. Finally, $L^2(\Omega)$ will denote the space of equivalence classes of measurable functions $f:\Omega \to \CC$ whose $L^2$ norm $ \| f \|^2_{L^2(\Omega)}\coloneqq\int_{\Omega} | f(z) |^2 \ d^{2m}z$ is finite.

\subsection{Definitions}
\subsubsection{Quantum information with continuous variables} \label{subsec:quantum-info-CV}

In this paper, we focus on continuous variable quantum systems. The Hilbert space of a set of $m$ harmonic oscillators, in this context called `modes', is the space $\cH_m\coloneqq \LL^2(\RR^m)$ of square-integrable functions on $\RR^m$. Let $x_j,p_j$ be the canonical position and momentum operators on the $j^{\text{th}}$ mode. The $m$ annihilation and creation operators, denoted by $a_j\coloneqq (x_j+ip_j)/\sqrt2$ and $a_j^\dagger \coloneqq (x_j-ip_j)/\sqrt2$ ($j=1,\ldots, m$), satisfy the commutation relations
\bb
[a_j, a_k]=0\, ,\qquad [a_j,a_k^\dag]= \delta_{jk} I\, ,
\label{CCR}
\ee
where $I$ is the identity on $\cH_m$. An $m$-mode quantum state $\rho$ is said to be \emp{centred} if 
\bb
\tr[\rho\, a_j] \coloneqq \frac{1}{\sqrt2} \left( \tr[\rho\,x_j] + i\tr[\rho\, p_j]\right) \equiv 0\qquad \forall\ j=1,\ldots, m\, ,
\label{centred}
\ee
i.e.\ if all expected values of the canonical operators on $\rho$, defined according to \eqref{expected}, vanish.
For an $m$-tuple of non-negative integers $n = (n_1,\ldots, n_m)\in \NN_0^m$, the corresponding {\em{Fock state}} is defined by $\ket{n}\coloneqq (n_1!\ldots n_m !)^{-1/2} \bigotimes_{j=1}^m \big( a_j^\dag \big)^{n_j} \ket{0}$, where $\ket{0} \in \cH_m$ denotes the (multi-mode) vacuum state. In what follows, we often consider $m=1$.

The (von Neumann) {\em{entropy}} of a quantum state $\rho$ is defined as
\bb
S(\rho) \coloneqq - \tr\left[\rho\log \rho\right] ,
\label{eq:entropy}
\ee
which is well defined although possibly infinite.\footnote{One way to define it is via the infinite sum $S(\rho) = \sum_i (- p_i \log p_i)$, where $\rho=\sum_i p_i \ketbra{e_i}$ is the spectral decomposition of $\rho$. Since all terms of this sum are non-negative, the sum itself can be assigned a well-defined value, possibly $+\infty$.} The \emp{relative entropy} between two states $\rho$ and $\sigma$ is usually written as follows~\cite{Umegaki1962}
\begin{equation} \label{eq:rel-entropy}
    D(\rho\|\sigma) \coloneqq \tr\left[\rho \left(\log \rho - \log \sigma\right)\right] .
\end{equation}
Again, the above expression is well defined and possibly infinite~\cite{Lindblad1973}.\footnote{To define it one considers the infinite sum $D(\rho\|\sigma) \coloneqq \sum_{i,j} \left| \braket{e_i | f_j}\right|^2 \left(p_i \log p_i - p_i \log q_j + q_j - p_i \right)$, where $\rho = \sum_i p_i \ketbra{e_i}$ and $\sigma = \sum_j q_j \ketbra{f_j}$ are the spectral decompositions of $\rho$ and $\sigma$, respectively. As detailed in~\cite{Lindblad1973}, the convexity of $x\mapsto x\log x$ implies that all terms of this sum are non-negative, which makes the expression well defined.} 

\medskip
For two Hilbert spaces $\cH, \cH'$, a quantum channel $\mathcal{N}: \cT_1(\cH) \to \cT_1(\cH')$ is a completely positive, trace-preserving linear map. For a linear map $\mathcal{L}: \cT_1(\cH) \to \cT_1(\cH')$, we define its \emp{diamond norm} as
\begin{equation} \label{eq:diamond}
\left\|\mathcal{L}\right\|_\diamond\coloneqq\sup_{k\in\NN}\sup_{X\in \cT_1(\cH\otimes \mathbb{C}^k)}\,\frac{\left\|(\mathcal{L}\otimes \operatorname{id}_{\mathbb{C}^k})(X)\right\|_1}{\|X\|_1},
\end{equation} 
where the supremum is over all non-zero trace class operators $X$ on $\cH\otimes \CC^k$. 

\medskip
Consider a quantum system with Hilbert space $\cH$, governed by a \emp{Hamiltonian} $H$, which is taken to be a positive (possibly unbounded) operator on $\cH$. The energy of a state $\rho \in \cD(\cH)$ is the quantity $\tr[\rho H]\in \RR_+\cup \{+\infty\}$ defined as in \eqref{expected positive}.

Given two Hilbert spaces $\cH$ and $\cH'$, a Hamiltonian $H$ on $\cH$, and some energy bound $E > \inf_{\ket{\psi}\in \cH} \braket{\psi|H|\psi}\geq 0$, the corresponding \emp{energy-constrained classical capacity} of a channel $\mathcal{N}: \cT_1(\cH) \to \cT_1(\cH')$ is given by~\cite{Holevo1979, proto-HSW, H-Schumacher-Westmoreland, Holevo-S-W, Holevo-energy-constrained}
\bb
\begin{aligned}
\mathcal{C}_{H}\left(\mathcal{N}, E\right) &= \lim_{n\to\infty} \frac1n\, \chi_{H^{(n)}}\left(\mathcal{N}^{\otimes n},\, nE\right) ,\\
\chi_{H} \left(\mathcal{N}, E\right) &\coloneqq \sup_{\substack{\\[-.1ex] \{p_i,\rho_i\}_i \\[.3ex] \sumno_i p_i \tr[\rho_i H]\leq E}} \left\{ S\left( \sumno_i p_i \rho_i \right) - \sumno_i p_i\, S\left(\rho_i\right) \right\} ,
\end{aligned}
\label{classical capacity}
\ee
where it is understood that the Hamiltonian $H^{(n)}$ on $\cH^{\otimes n}$ is given by $H^{(n)}\coloneqq \sum_{j=1}^n H_j$, where $H_j$ acts on the $j^{\text{th}}$ tensor factor, and tensor products with the identity operator are omitted for notational simplicity. With the same notation, one can also define the \emp{energy-constrained quantum capacity} of $\mathcal{N}$, given by~\cite{Lloyd-S-D, L-Shor-D, L-S-Devetak, winter2017energy, Mark-energy-constrained}
\bb
\begin{aligned}
\mathcal{Q}_{H}\left(\mathcal{N}, E\right) &\coloneqq \lim_{n\to\infty} \frac1n\, Q_{H^{(n)}}^{(1)}\left(\mathcal{N}^{\otimes n},\, nE\right) ,\\
Q_{H}^{(1)} \left(\mathcal{N}, E\right) &\coloneqq \sup_{\substack{ \\[-.4ex] \ket{\Psi}\in \cH\otimes \widetilde{\cH} \\[.4ex] \braket{\Psi | H\otimes I | \Psi}\leq E}} \left\{ S \left( \mathcal{N} \left( \tr_{\widetilde{\cH}} \ketbra{\Psi} \right)\right) - S \left( (\mathcal{N}\otimes I) \left( \ketbra{\Psi} \right)\right) \right\} ,
\end{aligned}
\label{quantum capacity}
\ee
where $\tr_{\widetilde{\cH}}$ is the partial trace over the entirely arbitrary ancillary Hilbert space $\widetilde{\cH}$. 
In this paper we are interested in the simple case $\cH=\cH_{m}=\LL^2(\RR^m)$ and $\cH'=\cH_{m'}=\LL^2(\RR^{m'})$, so that there is a natural choice for $H$, namely, the \emp{canonical Hamiltonian} 
\bb
H_m \coloneqq \sum_{j=1}^m a^\dag_j a_j
\label{canonical Hamiltonian}
\ee
of $m$ modes. In this case, we will omit the subscripts and simply write the energy-constrained capacities as $\mathcal{C}\left(\mathcal{N}, E\right)$ and $\mathcal{Q}\left(\mathcal{N}, E\right)$. 



\subsubsection{Phase space formalism} \label{subsec:phase-space}

We define the \emp{displacement operator} $\D(z)$ associated with a complex vector $z\in\CC^m$ as
\bb
\D(z) = \exp \left[ \sumno_j (z_j a^\dag_j - z_j^* a_j) \right].
\label{D}
\ee
Thus, $\D(z)$ is a unitary operator and satisfies $\D(z)^\dag=\D(-z)$ and 
\bb
\D(z) \D(w) = \D(z+w) \,e^{\frac12 (z^\intercal w^* - z^\dag w)} ,
\label{CCR Weyl}
\ee
valid for all $z,w\in \CC^m$.

Let $H_{\operatorname{quad}} = \sum_{j,k} \left(X_{jk} a_j^\dag a_k + Y_{jk} a_j a_k +  Y_{jk}^* a_j^\dag a_k^\dag \right)$, where $X=X^\dag$ is an $m\times m$ Hermitian matrix, and $Y=Y^\intercal$ is an $m\times m$ complex symmetric matrix. The unitaries $e^{-iH_{\operatorname{quad}}}$ generated by such Hamiltonians, and products thereof,\footnote{While not all products of unitaries of the form $e^{-iH_{\operatorname{quad}}}$ can be written as a single exponential, two such factors always suffice. See~\cite[p.37]{GOSSON}, combined with~\cite[Propositions~2.12,~2.18, and~2.19]{GOSSON} and with the observation that the exponential Lie map of the unitary group is surjective.} are called \emp{symplectic unitaries}, because they induce a symplectic linear transformation at the phase space level $(z_R, z_I)\in \RR^{2m}$, where $z_R\coloneqq \Re z$ and $z_I\coloneqq \Im z$ \cite{BUCCO, GOSSON}. A symplectic unitary is called \emp{passive} if it commutes with the number operator $\sum_j a_j^\dag a_j$, which happens whenever the generating Hamiltonian $H_{\text{quad}}$ satisfies $Y=0$. A passive symplectic unitary $V$ acts on annihilation operators as $Va_j V^\dag = \sum_k U_{jk} a_k$, where $U$ is an $m\times m$ unitary matrix.

\medskip
For trace class operators $T\in \cT_1(\cH_m)$, the \emp{quantum characteristic function} $\chi_T:\CC^m\to \CC$ is given by
\bb
\chi_T(z) \coloneqq \tr\left[ T\, \D(z) \right] .
\label{chi}
\ee
Conversely, the operator $T$ can be reconstructed from $\chi_T$ via the weakly defined identity
\bb
T = \int \frac{d^{2m}z}{\pi^{m}}\, \chi_T(z) \D(-z)\,.
\label{chi inverse}
\ee
Observe that the adjoint $T^\dag$ of $T$ satisfies $\chi_{T^\dag}(z)=\chi_T(-z)^*$ for all $z\in \CC^m$, so that $T$ is self-adjoint if and only if $\chi_T(-z)\equiv \chi_T(z)^*$. The characteristic function $\chi_T$ of a trace class operator $T$ is bounded and uniformly continuous \cite[\S~5.4]{HOLEVO}. If $T$ is positive semi-definite (e.g.\ if $T$ is a density operator), then $\max_{\alpha} |\chi_T(\alpha)| = \chi_T(0)=\tr [T]$. 

\medskip

We write $\ket{\psi_f}$ to denote the pure state corresponding to the wave function $f\in \LL^2(\RR^m)$, so that the corresponding rank-one state $\psi_f\coloneqq \ketbra{\psi_f}$ has the following characteristic function:
\bb
\chi_{\psi_f} (z) = e^{-i z_I^\intercal z_R} \int d^m x\, f^*(x)\, f\big( x -\sqrt2 z_R \big)\, e^{\sqrt2\, i\, z_I^\intercal x}\, ,
\label{chi psi f}
\ee
where as usual $z = z_R + i z_I$.

The Fourier transform of the characteristic function is known as the \emp{Wigner function}. For a trace class operator $T$, the Wigner function is given by~\cite[Eq.~(4.5.12) and~(4.5.19)]{BARNETT-RADMORE}
\begin{align}
W_T(z) \coloneqq&\ \int \frac{d^{2m}w}{\pi^{2m}}\, \chi_T(w)\, e^{z^\intercal w^* - z^\dag w} \label{Wigner 1} \\[0.5ex]
=&\ \frac{2^m}{\pi^m} \tr\left[ \D(-z) T \D(z)\, (-1)^{\sumno_j a_j^\dag a_j}\right] . \label{Wigner 2}
\end{align}
Observe that $W_{T^\dag}(z)=W_T(z)^*$, so that $T$ is self-adjoint if and only if $W_T(z)\in \RR$ for all $z\in\CC^m$. From \eqref{Wigner 2} it is not difficult to see that $|W_T(z)|\leq \frac{2^m}{\pi^m} \|T\|_1$, where $\|T\|_1=\tr|T|$ reduces to $1$ when $T$ is a density operator. By taking the Fourier transform of \eqref{chi psi f}, one can show that
\bb
W_{\psi_f} (z) = \frac2\pi \int d^m x\, f^* \big( x+ \sqrt2 z_R \big)\, f\big( -x +\sqrt2 z_R \big)\, e^{2\sqrt2\, i\, z_I^\intercal x}\, .
\label{W psi f}
\ee
Moreover, the energy of any density matrix, $\rho$, can be obtained as a phase space integral 
\bb
\int d^{2m} z\, \|z\|^2\, W_\rho(z) = \tr\left[ \rho \left(H_m +\frac{m}{2}\, I\right) \right]
\label{energy  Wigner}
\ee

The displacement operator $\D(z)$ induces a translation or {\em{displacement}} of the Wigner function as follows, hence the nomenclature:
\bb
\chi_{\D(z)\, \rho\, \D(z)^\dag}(u) = e^{z^\dag u - z^\intercal u^*} \chi_\rho(u)\, ,\qquad W_{\D(z)\, \rho\, \D(z)^\dag}(u) = W_\rho(u-z)\, . \label{chi and W displacement}
\ee

The map $T\mapsto \chi_T$, defined for trace class operators $T$ in \eqref{chi}, extends uniquely to an isomorphism between the space of Hilbert--Schmidt operators and that of square-integrable functions $\LL^2(\CC^m)$. In fact, the \emp{quantum Plancherel theorem} guarantees that this is also an isometry, namely
\bb
\tr[S^\dag T] = \int \frac{d^{2m}z}{\pi^m}\, \chi_S(z)^*\, \chi_T(z) = \pi^m \int d^{2m}z\, W_S(z)^* W_T(z)\, 
\label{Plancherel}
\ee
and therefore
\bb
\|\rho-\sigma\|_2^2 =  \int \frac{d^{2m}z}{\pi^m}\, \left|\chi_\rho(z) - \chi_\sigma(z)\right|^2 = \pi^m \int d^{2m}z\, \left( W_\rho(z) - W_\sigma(z) \right)^2\, .
\label{Plancherel HS norm}
\ee
Henceforth, we refer to \eqref{Plancherel HS norm} as the \emp{quantum Plancherel identity}.

\medskip

{\em{Gaussian states}} on $\cH_m$ are the density operators $\rho\in \cD(\cH_m)$ such that $W_\rho(z)$ is a Gaussian probability distribution on the real space $(z_R, z_I)\in \RR^{2m}$ and are uniquely defined by their first and second moments. A particularly simple example of a single-mode Gaussian state is a \emp{thermal state} with mean photon number $N\in [0,\infty)$, given by
\bb
\tau_N \coloneqq \frac{1}{N+1} \sum_{n=0}^\infty \left( \frac{N}{N+1}\right)^n \ketbra{n}\, .
\label{tau}
\ee
The thermal state is the maximiser of the entropy among all states with a fixed maximum average energy:
\bb
\max\left\{ S(\rho):\ \rho \in\cD(\cH_1),\ \tr\left[ \rho\, a^\dag a\right] \leq N\right\}  = S(\tau_N) = g(N)
\label{tau maximises entropy}
\ee
for all $N\geq 0$, where the function $g$ is defined by
\bb
g(x)\coloneqq (x+1)\log (x+1) - x\log x\, .
\label{g}
\ee
The characteristic function and Wigner function of the thermal state evaluate to \cite[Eq.~(4.4.21) and~(4.5.31)]{BARNETT-RADMORE}
\bb
\chi_{\tau_N}(z) = e^{-(2N+1)|z|^2/2}\, ,\qquad W_{\tau_N}(z) = \frac{2}{\pi(2N+1)}\, e^{-2|z|^2/(2N+1)}\, ,
\label{chi and W tau}
\ee
respectively, so that $\tau_N$ is easily seen to be a centred Gaussian state.

\subsection{Moments}

\begin{Def}[(Standard Moments)] \label{def:moments}
An $m$-mode quantum state $\rho$ is said to have {\em{finite standard moments}} of order up to $k$, for some $k\in [0,\infty)$, if
\bb
M_k(\rho) \coloneqq \tr\left[\rho H_m^{k/2} \right] 
< \infty\, ,
\label{eq:moments}
\ee
where $H_m$ is the canonical Hamiltonian \eqref{canonical Hamiltonian}, and the above trace is defined as in \eqref{expected positive}.
\end{Def}
\smallskip
\begin{rem*}
The above condition is fairly easy to check once the matrix representation of $\rho$ in the Fock basis is given. Namely, resorting to \eqref{expected positive} and exchanging the order of summation for infinite series with non-negative terms, we see that \eqref{eq:moments} is equivalent to
\bb
M_k(\rho) = \sum_{n\in \NN_0^m} (m+|n|)^{k/2}\! \braket{n|\rho|n} <\infty\, ,
\label{eq:moments-Fock}
\ee
where as usual $|n|=\sumno_j n_j$.
\end{rem*}

Given $k>0$ and $m\in\mathbb{N}$, we can also define, by analogy with classical harmonic analysis, the \emp{$m$-mode bosonic Sobolev space} of order $k$ as follows
\begin{align*}
\mathcal{W}^{k,1}(\cH_m)\coloneqq \left\{X\in\mathcal{T}_1\left(\cH_m\right) ;\,\|X\|_{\mathcal{W}^{k,1}(\cH_m)}<\infty \right\},
    \end{align*}
where as usual $\cH_m = L^2(\RR^m)$. Here, we set
\begin{align*}
   \|X\|_{\mathcal{W}^{k,1}(\cH_m)} \coloneqq \left\| \left(H_m + m I \right)^{k/4}\,X\,\left(H_m + mI\right)^{k/4}\right\|_1\, ,
\end{align*}
with the canonical Hamiltonian on $m$ modes being defined by \eqref{canonical Hamiltonian}. For density operators $\rho$ it holds, using monotone convergence and cyclicity of the trace, that
\begin{equation}
    \begin{split}
\left\|\rho\right\|_{\mathcal{W}^{k,1}(\cH_m)} 
&=\sup_E\tr\left[\indic_{[0,E]}(H_m)\left(H_m + mI \right)^{k/4} \rho \,\left(H_m + mI \right)^{k/4}\indic_{[0,E]}(H_m)\right] \\ 
&=\sup_E\tr\left[\rho \,\left(H_m + mI \right)^{k/2}\indic_{[0,E]}(H_m)\right] \\
&= \tr\left[\rho \,\left(H_m + mI \right)^{k/2}\right] 
\end{split}
\end{equation}
where $\indic_{[0,E]}$ is the indicator function of the interval $[0,E].$

It is well known that the characteristic function of any classical random variable with finite moments of order up to $k$ (with $k$ being a positive integer) is continuously differentiable $k$ times everywhere. We can draw inspiration from this fact to devise an alternative way to introduce moments, relying on the regularity of the quantum characteristic function, in the quantum setting as well. We refer to moments defined in this manner as {\em{phase space moments}}.

\begin{Def}[(Phase space moments)] \label{def:ps-moments}
An $m$-mode quantum state $\rho$ is said to have \emp{finite phase space moments} of order up to $k$, for some $k\in [0,\infty)$, if
\bb
M'_{k}(\rho,\varepsilon) \coloneqq \left\| \chi_\rho \right\|_{C^{k}\left(B(0,\varepsilon)\right)} < \infty
\label{eq:ps-moments}
\ee
for some $\varepsilon>0$, where $B(0,\varepsilon) \coloneqq \{z\in \CC^m: |z|\leq \varepsilon\}$ is the Euclidean ball of radius $\varepsilon$ centred in $0$, and the norm on the space $C^{k}\left(B(0,\varepsilon)\right)$ is defined by \eqref{Ck norm 1} and \eqref{Ck norm 2}. 
\end{Def}

In complete analogy with the classical case, finiteness of standard moments implies local differentiability of the characteristic function, and hence finiteness of phase space moments. See Theorem~\ref{thm:all-of-us} of Section~\ref{sec:main-results}.

However, the converse is not true in general. This is not surprising, as the same phenomenon is observed for classical random variables. In fact, a famous example by Zygmund \cite{Zygmund1947} shows the existence of classical random variables with continuously differentiable characteristic function whose first absolute moments do not exist. We can swiftly carry over his example to the quantum realm, e.g.\ by considering a particular displaced vacuum state $\rho\coloneqq c \sum_{n=2}^\infty \frac{1}{n^2\log n} \left(\D(n)\ketbra{0}\D(n)^\dag + \D(n)^\dag\ketbra{0}\D(n)\right)$. One can show that its characteristic function is $\chi_\rho(z)=e^{-|z|^2}\sum_{n=2}^\infty \frac{\cos(2n z_I)}{n^2\log n}$, which turns out to be continuously differentiable everywhere \cite{Zygmund1947, Zygmund1947}. However, 
\[\tr[\rho |x|]\geq c \sum_{n=2}^\infty \frac{1}{n^2\log n}\, 2\sqrt2 n = +\infty,\] which implies that $\rho$ has no finite first-order moments (see Lemma \ref{messy inequality lemma}). 

In spite of the above counterexample, we show in Theorem \ref{thm:converse} that at least if $k$ is an even integer, then the existence of $k^{\text{th}}$ order phase space moment implies the existence of the $k^{\text{th}}$ order standard moment. Again, this is in total analogy with the classical case \cite[Theorem~1.8.16]{Ushakov}.
\smallskip

{\bf Remark:} Due to the above, for even $k$, we simply use the word {\em{moment}} in the statements of our theorems, instead of differentiating between standard moments and phase space moments.

\subsection{Quantum convolution} \label{subsec:BS}

A \emp{beam splitter} with transmissivity $\lambda\in [0,1]$ acting on two sets of $m$ modes is a particular type of a passive symplectic unitary, which we express as\footnote{Tensor products are omitted here.}
\bb
U_\lambda \coloneqq \exp \left[\left( \arccos\sqrt{\lambda}\right)\sumno_j (a_j^\dag b_j - a_j b_j^\dag )\right]\, ,
\label{BS}
\ee
where $a_j$ and $b_j$ ($j=1,\ldots, m$) are the creation operators of the first and second sets of modes, respectively. Its action on annihilation operators can be represented as follows
\bb
U_\lambda a_j U_\lambda^\dag = \sqrt\lambda \, a_j - \sqrt{1-\lambda}\,b_j\, ,\qquad U_\lambda b_j U_\lambda^\dag = \sqrt{1-\lambda} 
\,a_j + \sqrt\lambda\,b_j\, \quad \forall j \in \{1,..,m \}.
\label{BS action creation}
\ee
Accordingly, displacement operators are transformed by
\bb
U_\lambda \left( \D(z)\otimes \D(w)\right) U_\lambda^\dag = \D\left(\sqrt\lambda z + \sqrt{1-\lambda}\, w\right)\otimes \D\left(-\sqrt{1-\lambda}\, z +\, \sqrt\lambda \, w\right)\, .
\label{BS action displacement}
\ee

\medskip
The beam splitter unitary can be used to define the following ($\lambda$-dependent) \emp{quantum convolution}: for two $m$-mode quantum states $\rho,\sigma$ and $\lambda\in [0,1]$, their ($\lambda$-dependent) \emp{quantum convolution} is given by the state $\rho\boxplus_\lambda \sigma$ which is defined according to \cite{Koenig2014} as
\bb
\rho \boxplus_\lambda \sigma \coloneqq \tr_2 \left[ U_\lambda (\rho\otimes \sigma) U_\lambda^\dag\right] .
\label{boxplus}
\ee
In terms of characteristic functions, this definition corresponds to
\bb
\chi_{\rho\, \boxplus_\lambda \sigma} (z) = \chi_\rho \left( \sqrt{\lambda}\, z \right)\, \chi_\sigma\left( \sqrt{1-\lambda}\, z \right) .
\label{boxplus characteristic functions}
\ee
It is not difficult to verify that for all symplectic unitaries $V$ and all $\lambda\in [0,1]$, the beam splitter unitary $U_\lambda$ of \eqref{BS} satisfies $\left[ V \otimes V,\ U_\lambda \right] = 0$. In particular, 
\bb
V(\rho\boxplus_\lambda \sigma) V^\dag = (V\rho V^\dag) \boxplus_\lambda (V\sigma V^\dag)
\label{boxplus covariant symplectic}
\ee
for any state $\sigma$. Also, using \eqref{BS action creation} it can be shown that the mean photon number of a quantum convolution is just the convex combination of those of the input states, i.e.
\bb
\tr\left[ (\rho\boxplus_\lambda \sigma) H_m \right] = \lambda \tr\left[ \rho H_m \right] + (1-\lambda) \tr\left[ \sigma H_m \right] ,
\label{energy convolution}
\ee
where the canonical Hamiltonian is defined by \eqref{canonical Hamiltonian}.

\medskip
For all $m$-mode quantum states $\sigma$ and all $\lambda\in [0,1]$, we can use the corresponding convolution to define a quantum channel $\mathcal{N}_{\sigma, \lambda}:\cT_1(\cH_1)\to \cT_1(\cH_1)$, whose action is given by
\bb
\mathcal{N}_{\sigma, \lambda}(\rho) \coloneqq \rho\boxplus_\lambda \sigma\, .
\label{channel N}
\ee
When $\sigma=\tau_N$ is a thermal state (with mean photon number $N$), the channel $\mathcal{N}_{\tau_N, \lambda} \eqqcolon \att$ is called a \emp{thermal attenuator channel}. Its action, obtained by combining \eqref{boxplus characteristic functions} and \eqref{chi and W tau}, is given by
\bb
\att :\chi_\rho(z) \longmapsto \chi_{\att(\rho)}(z) \coloneqq \chi_\rho\left(\sqrt\lambda\, z\right)\, e^{-\frac{1-\lambda}{2}\, (2N+1)|z|^2}\, .
\label{thermal attenuator}
\ee
For the thermal attenuator channel, the energy-constrained classical capacity (defined in \eqref{classical capacity}) can be shown to reduce to
can be shown to be given by \cite{Giovadd, Giovadd-CMP}
\bb
\mathcal{C}\left(\att, E\right) = g\left(\lambda E+(1-\lambda)N\right) - g\left( (1-\lambda) N\right) ,
\label{classical capacity thermal attenuator}
\ee
where $g$ is given by \eqref{g}.


\medskip
In what follows, we will be interested in the symmetric quantum convolutions $\rho_1\boxplus\ldots \boxplus \rho_n$, iteratively defined for a positive integer $n$ and states $\rho_1, \ldots, \rho_n$,
by the relations $\rho\boxplus \sigma\coloneqq \rho \boxplus_{1/2}\sigma$ and
\bb
\rho_1 \boxplus \ldots \boxplus \rho_n \coloneqq \left( \rho_1\boxplus \ldots \boxplus \rho_{n-1}\right) \boxplus_{1-1/n} \rho_n\, .
\label{boxplus iterated}
\ee

We will also use the shorthand
\bb
\rho^{\boxplus n} \coloneqq \underbrace{\rho\boxplus\ldots \boxplus \rho}_{\text{$n$ times}}\, .
\label{boxplus n}
\ee
In terms of characteristic and Wigner functions, we can also write 
\begin{align}
\chi_{\rho_1\,\boxplus\ldots \boxplus\, \rho_n}(z) = \chi_{\rho_1}\left(z/\sqrt{n}\right) \ldots \chi_{\rho_n}\left( z/\sqrt{n}\right) ,  \label{box chi} \\[0.5ex]
W_{\rho_1\,\boxplus\ldots\boxplus\, \rho_n}(z) = n^m  \left(W_{\rho_1}\star\ldots\star W_{\rho_n}\right) \left(\sqrt{n} z\right) . \label{box W}
\end{align}
Here, $\star$ denotes convolution, which is defined for $n$ functions $f_1,\ldots, f_n:\CC^m\to \RR$ by
\bb
(f_1\star\ldots  \star f_n)(x) \coloneqq \int d^m y_1\ldots d^m y_{n-1}\, f_1(y_1)\ldots f_{n-1}(y_{n-1}) f_n\left(x-y_1-\ldots - y_{n-1}\right) \label{convolution}
\ee
Equation \eqref{box chi} shows that the quantum characteristic function of the symmetric quantum convolution satisfies the same scaling property as a sum of classical i.i.d. (independent and identically distributed) random variables. The important special case $\rho_i\equiv \rho$ of \eqref{box chi} for all $i \in \{i,2,\ldots,n\}$, on which we will focus most of our efforts, reads
\bb
\chi_{\rho^{\boxplus n}}(z) = \chi_\rho\left( z/\sqrt{n}\right)^n\, .
\label{chi CH state}
\ee
Iterating \eqref{boxplus covariant symplectic}, using \eqref{boxplus iterated}, shows that
\bb
V\left( \rho^{\boxplus n}\right)V^{\dag} = (V\rho V^\dag)^{\boxplus n}
\label{symmetric boxplus covariant symplectic}
\ee
holds for all symplectic unitaries $V$.

\section{Cushen and Hudson's quantum central limit theorem} \label{sec:CH-CLT}

In \cite{Cushen1971}, Cushen and Hudson proved the following quantum mechanical analogue of the central limit theorem, which is the starting point of our study.

\begin{thm}[{\cite[Theorem~1]{Cushen1971}}] \label{thm:CushHud}
Let $\rho\in \cD(\cH_m)$ be a centred $m$-mode quantum state with 
finite second moments. Then the sequence $(\rho^{\boxplus n})_{n\in\NN}$ converges weakly to the Gaussian state $\rho_\G$ of same first and second moments as $\rho$:
\begin{align}
    \tr \left[ \rho^{\boxplus n} X\right] \tendsn{} \tr\left[\rho_\G X\right]\, , \qquad \forall\ X\in \cB(\cH)\, ,
\end{align}
where $\cB(\cH_m)$ is the set of bounded operators on $\cH_m$.
\end{thm}

\begin{rem*}
The state $\rho_\G$ is commonly called the \emp{Gaussification} of $\rho$.
\end{rem*}

In fact, the proof of Theorem \ref{thm:CushHud} relies on the equivalence between weak convergence of states and pointwise convergence of their characteristic functions. More precisely, the following holds:

\begin{lemma}[{(\cite[Lemma~4.3]{Davies1969} and~\cite[Lemma~4]{G-dilatable})}] 
Let $(\rho_n)_{n\in\NN}$ be a sequence of density operators on $\cH_m$. The following are equivalent:
\begin{itemize}
    \item $(\rho_n)_{n\in\NN}$ converges to a density operator in the weak operator topology, namely, it holds that $\lim_{n \to \infty} \braket{x|\rho_n|y} = \braket{x|\rho|y}$ for all $\ket{x}, \ket{y}\in \cH_m$;
    \item $(\rho_n)_{n\in\NN}$ converges in trace distance to a trace class operator;
    \item the sequence $(\chi_{\rho_n})_{n\in\NN}$ of characteristic functions converges pointwise to a function that is continuous at $0$.
\end{itemize}
\end{lemma}

Together, the above lemma and Theorem \ref{thm:CushHud} allow us to conclude the following seemingly stronger convergence:

\begin{thm} \label{CH thm}
Under the assumptions of Theorem \ref{thm:CushHud}, we have that
\bb
\lim_{n\to\infty} \left\| \rho^{\boxplus n} - \rho_\G\right\|_1 = 0\, .
\label{CH eq}
\ee
\end{thm}

\section{Main results}\label{sec:main-results}

The main objective of this paper is to refine Theorem~\ref{CH thm} of the previous section in the following directions: 
\begin{itemize}
    \item Firstly, in the case in which the state $\rho$ satisfies the conditions of the Cushen--Husdon theorem, we provide quantitative bounds on the rate at which the sequence of states $(\rho^{\boxplus n})_{n\in\NN}$ converges to $\rho_\G$, under the assumption of finiteness of certain phase space moments of $\rho$. We also show how finiteness of phase space moments is implied by finiteness of the corresponding standard moments, the latter having the advantage of being a more easily verifiable condition. Moreover, we show that finiteness of even integer phase space moments implies finiteness of even integer standard moments  (Section \ref{subsec:quantitative}).
    \item Secondly, we provide an example to show that the assumption that the second moments be finite in the Cushen--Hudson theorem cannot be weakened (Section \ref{subsec:optimality}).
	\item Thirdly, we extend our results to the non-i.i.d.~setting, i.e.\ we consider a scaling in the quantum convolution different from \eqref{boxplus iterated}. This allows us to analyse the propagation of states through cascades of beam splitters with varying transmissivities (Section \ref{subsec:applications}).
    \item Finally, we provide a precise asymptotic analysis of the behaviour of quantum characteristic functions at zero and at infinity (Section \ref{subsec:new-results}).
\end{itemize}

\subsection{Quantitative bounds in the QCLT} \label{subsec:quantitative}

In this section, we state our results on rates of convergence in the Cushen--Hudson quantum central limit theorem. We call them quantum Berry--Esseen theorems, as is customary in the literature. Our first theorem provides convergence rates $\mathcal O\left(n^{-1/2}\right)$ in the quantum central limit theorem under a fourth-order moment condition. The rate of convergence is boosted to $\mathcal O\left(n^{-1}\right)$ if the third derivative of the characteristic function at zero vanishes:

\begin{thm}[(Quantum Berry--Esseen theorem; High regularity)] \label{thm:QBE'}
Let $\rho$ be a centred $m$-mode quantum state with finite fourth-order phase space moments. Then, the convergence in the quantum central limit theorem in Hilbert--Schmidt norm satisfies
\bb
\left\|\rho^{\boxplus n}-\rho_\G\right\|_2 = \mathcal{O}_{M_4'}\left(n^{-1/2}\right).
\ee
Here, $M_4'=M'_4(\rho,\varepsilon)$ is the moment defined in \eqref{eq:ps-moments}, and $\varepsilon>0$ is sufficiently small. Moreover, if $D^3\chi_{\rho}(0)= 0$ then the convergence is at least with rate $\mathcal{O}_{M_4'}\left(n^{-1}\right)$.
\end{thm}

The proof of Theorem \ref{thm:QBE'} is provided in 
Section \ref{sec-proofs:QCLT}. In the next Theorem, we weaken the assumption on the moments of the state $\rho$, which leads to a slower rate of convergence.

\begin{thm}[(Quantum Berry--Esseen theorem; Low regularity)] \label{thm:QBElow}
Let $\rho$ be a centred $m$-mode quantum state with finite $(2+\alpha)$-order phase space moments, where $\alpha \in (0,1]$. The convergence in the quantum central limit theorem in Hilbert--Schmidt norm is given by
\begin{align*}
\|\rho^{\boxplus n}-\rho_\G\|_2 = \mathcal{O}_{M_{2+\alpha}'}\left(n^{-\alpha/2}\right).
\end{align*}
Here, $M_{2+\alpha}'=M'_{2+\alpha}(\rho,\varepsilon)$ is the phase space moment defined in \eqref{eq:ps-moments}, and $\varepsilon>0$ is sufficiently small.
\end{thm}

The proof of Theorem \ref{thm:QBElow} is provided in Section \ref{sec-proofs:QCLT}. The variable $\alpha$ allows us to obtain a convergence rate under the assumption of finiteness of phase space moments of order all the way down to $2$ (excluded), which is the assumption required in the Cushen--Hudson QCLT. The above results can further be used to find convergence rates in other, statistically more relevant, distance measures:

\begin{cor}[(Convergence in trace distance and relative entropy)] \label{cor:tracerel}
Assume that an $m$-mode quantum state $\rho$ has finite third-order phase space moments. Then,
\bb
\left\|\rho^{\boxplus n}-\rho_\G\right\|_1=\mathcal{O}_{M_3'}\left(n^{-\frac{1}{2(m+1)}}\right),\qquad D\left(\rho^{\boxplus n}\big\|\rho_\G\right)= \mathcal{O}_{M_3'}\left(n^{-\frac{1}{2(m+1)}}\right) ,
\ee
where $M_3'=M'_3(\rho,\varepsilon)$ is defined in \eqref{eq:ps-moments}, and $\varepsilon>0$ is sufficiently small. The above rates are replaced by $\mathcal{O}_{M'_{2+\alpha}}\left(n^{-\alpha/(2m+2)}\right)$ when $\rho$ only satisfies the conditions of Theorem~\ref{thm:QBElow}.
\end{cor}
The proof of this Corollary is given in Section~\ref{sec-proofs:QCLT}.

\begin{rem*}[(Condition on the existence of moments)] The error bounds in Theorems \ref{thm:QBE'} and \ref{thm:QBElow} are stated in terms of assumptions on the phase space moments $M'_k$ given by \eqref{eq:ps-moments}, of the state. It is possible to bound the phase space moments $M'_k$ directly in terms of the standard moments $M_k$ defined in \eqref{eq:moments}. This is stated in the following Theorem, whose proof is given in Appendices~\ref{app:moments}--\ref{sec:partconv}
\end{rem*}

\begin{thm} \label{thm:all-of-us}
Let $k\in [0,\infty)$, $m$ a positive integer, and $\varepsilon>0$ be given. Then every $m$-mode quantum state with finite standard moments of order up to $k$ also has finite phase space moments of the same order. More precisely, there is a constant $c_{k,m}(\varepsilon)<\infty$ such that  
\bb
M'_{k}(\rho,\varepsilon) = \left\|\chi_\rho \right\|_{C^{k}\left(B(0,\varepsilon)\right)} \leq c_{k,m}(\varepsilon) M_k(\rho)\, .
\ee

Conversely, if the characteristic function is $2k$ times totally differentiable at $z=0$ for some integer $k$, then the $2k^{\text{th}}$ standard moment is finite as well.
\end{thm}

The importance of Theorem \ref{thm:all-of-us} for us comes from the fact that most of our proofs rest upon local differentiability properties of the characteristic function. While mathematically useful, such properties have no direct physical meaning and may be hard to verify in practice. Instead, the condition of finiteness of higher-order standard moments, as given in Definition \ref{def:moments}, bears a straightforward physical meaning, related to the properties of the photon number distribution of the state, and is often easier to verify.

The key to proving Theorem \ref{thm:all-of-us} for fractional $k$ lies in an interpolation argument. To state it precisely, we briefly recall some basic facts about real interpolation theory (see \cite{bergh2012interpolation} for more details): given two Banach spaces $(\cX_0,\|.\|_{\cX_0})$ and $(\cX_1,\|.\|_{\cX_1})$, and a parameter $0\le \theta\le 1$, define the {\em{$K$-function}} as follows:
\begin{equation}
\label{eq:K}
K(t,X) :=  \inf_{X=X_0+X_1} \left( \left\lVert X_0 \right\rVert_{\cX_0}+t\left\lVert X_1 \right\rVert_{\cX_1} \right) \quad {\forall}\, t>0,
\end{equation}
and derive from this the function $\Phi_{\theta} (K(X)) = \sup_{t>0} t^{-\theta}K(t,X).$ The real interpolation spaces, parametrised by $\theta \in (0,1)$, are then defined as 
\[(\cX_0,\cX_1)_{\theta}:= \left\{ X \in \cX_0;\, \Phi_{\theta}(K(X))< \infty \right\}\text{ with norm }\left\lVert X \right\rVert_{(\cX_0,\cX_1)_\theta} = \Phi_{\theta}(K(X)).\]
Now, given two couples of Banach spaces $\cX_0, \cX_1$ and $\cY_0,\cY_1$, and a map $\Psi:\cX_0+\cX_1\to\cY_0+\cY_1$ such that $\Psi:\cX_0\to\cY_0$ and $\Psi:\cX_1\to\cY_1$ are bounded, the map $\Psi:(\cX_0,\cX_1)_{\theta}\to (\cY_0,\cY_1)_\theta$ is bounded and:
\begin{align*}
    \|\Psi:(\cX_0,\cX_1)_{\theta}\to (\cY_0,\cY_1)_\theta\|\le \|\Psi:\cX_0\to\cY_0\|^{1-\theta}\,\|\Psi:\cX_1\to\cY_1\|^{\theta}\,.
\end{align*}
We want to apply this to the map $\rho\mapsto \chi_\rho$.

The following interpolation result for density operators then holds: 

\begin{prop}\label{theo:moments}
Let $k_1 \ge k_0 \ge 0$ be real numbers. The \emp{$m$-mode bosonic Sobolev spaces} $\mathcal{W}^{k_0,1}(\cH_m)$ and $\mathcal{W}^{k_1,1}(\cH_m)$ form a compatible couple such that for any $m$-mode quantum state $\rho$ and $\theta \in (0,1)$ the real interpolation norm satisfies
\begin{equation}
\Vert \rho \Vert_{(\mathcal{W}^{k_0,1}(\cH_m),\mathcal{W}^{k_1,1}(\cH_m))_{\theta,\infty}} \le \Vert \rho \Vert_{\mathcal{W}^{(1-\theta)k_0+ \theta k_1,1}(\cH_m)}.
\end{equation}
\end{prop}
The proof of Proposition \ref{theo:moments} is stated in Appendix \ref{sec:SMF}.

\subsection{Optimality of convergence rates and necessity of finite second moments in the QCLT}\label{subsec:optimality}

The results stated in the previous section lead naturally to the following questions: 

(i) Can the assumption of finiteness of second moments in the Cushen--Hudson theorem be weakened?

(ii) Are the convergence rates of Theorems \ref{thm:QBE'} and \ref{thm:QBElow} and Corollary \ref{cor:tracerel} optimal? 
\medskip

\noindent
 We start by answering the first question in the negative: there exists a state with finite moments of all orders $2(1-\delta)$ (for $\delta>0$) for which neither 
 Theorem \ref{thm:CushHud} nor Theorem \ref{CH thm} holds.

\begin{prop} \label{crazy state prop}
Consider the one-mode state $\psi_f\coloneqq \ketbra{\psi_f}$ with wave function 
\bb
f(x) \coloneqq \frac{1}{\sqrt2}\, \frac{1}{(1+x^2)^{3/4}}\, .
\label{crazy wave function}
\ee
Then: (a)~$\psi_f$ is centred; (b)~$M_{2(1-\delta)}(\psi_f) = \braket{\psi_f|(a a^\dag)^{1-\delta}|\psi_f}<\infty$ for all $\delta>0$; yet (c)~the sequence $\left( \ketbra{\psi_f}^{\oplus n}\right)_n$ does not converge to any quantum state. Hence, the assumption of finiteness of second moments in the Cushen--Hudson QCLT (Theorems \ref{thm:CushHud} and \ref{CH thm}) cannot be weakened.
\end{prop}

The proof of the above proposition is given in Section \ref{sec-proofs:optimality}.

\medskip
We now come to the second question (ii) regarding  tightness of the estimates in Theorems \ref{thm:QBE'} and \ref{thm:QBElow} and Corollary \ref{cor:tracerel}. In Section \ref{sec-proofs:optimality} below, we study several explicit examples and provide convincing numerical evidence that our estimates are indeed tight, at least as far as the Hilbert--Schmidt convergence rates are concerned. Our findings are summarised as follows.
\begin{itemize}
\item We start by looking at the pure state $\ket{\psi}=(\ket{0}+\ket{3})/\sqrt2$, with density matrix $\psi = \ketbra{\psi}$ and thermal Gaussification $\psi_\G = \tau_{3/2}$. Our findings indicate that $\left\|\psi^{\boxplus n} - \psi_\G\right\|_2 \sim c\, n^{-1/2}$, in the sense that the ratio between the two sides tends to $1$ as $n\to \infty$, for some absolute constant $c$ (Example \ref{ex:Simon} and Figure~\ref{fig:rates}). Hence, the $\mathcal O(n^{-1/2})$ convergence rate of Theorem \ref{thm:QBElow} is attained.
\item Next, we focus on the second estimate of Theorem \ref{thm:QBE'}, and show that it is also tight. Namely, we compute the differences $\left\|\psi^{\boxplus n} - \psi_\G\right\|_\alpha$ for the simple case of a single-photon state $\psi=\ketbra{1}$ and for $\alpha=1,2$, and find numerical evidence that again $\left\|\psi^{\boxplus n} - \psi_\G\right\|_\alpha\sim c\, n^{-1}$ for some absolute constant $c$ (Example \ref{ex:Ludovico} and Figure~\ref{fig:rates}). This shows that the $\mathcal O(n^{-1})$ convergence rate stated in Theorem \ref{thm:QBE'}, under the assumption that $D^3\chi_\rho(0)=0$, is also attained. 
\end{itemize}

\subsection{Applications to capacity of cascades of beam splitters with non-Gaussian environment} \label{subsec:applications}

We now discuss applications of our results to the study of channels that arise naturally in the analysis of lossy optical fibres. We model a physical fibre of overall transmissivity $\lambda$ as a cascade of $n$ beam splitters, in each of which the signal state $\omega$ is mixed via an elementary beam splitter of transmissivity $\lambda^{1/n}$ with a fixed state $\rho$, modelling the environmental noise (Figure~\ref{fig:optical-fibre}). Each step corresponds to the action of the channel $\oneN:\omega\mapsto \omega\boxplus_{\lambda^{1/n}} \rho$ (cf.\ the definition \eqref{channel N}), so that the whole cascade can be represented by the $n$-fold composition $\oneN^n$. Note that this is in general a non-Gaussian channel, albeit it is Gaussian \textit{dilatable}~\cite{KK-VV,G-dilatable}. We are interested in the asymptotic expression of the output state $\oneN^n(\omega)$ as the number $n$ tends to infinity, as a function of the input state $\omega$. In other words, we want to study the asymptotic channel $\lim_{n\to\infty} \oneN^n$.

At this point, it should not come as a surprise that such a channel exists and coincides with $\mathcal{N}_{\rho_\G,\,\lambda}$.

Before we see why, let us justify why the above model may be relevant to applications. The recently flourishing field of integrated quantum photonics sets as its goal that of implementing universal quantum computation on miniaturised optical chips~\cite{OBrien2009, Politi2009, Carolan2015, Rohde2015}. A quantum channel that runs across such a circuit is susceptible to noise generated by other active elements of the same circuit, e.g.\ single-photon sources. While we expect such noise to be far from thermal, it may become so in the limit $n\to \infty$ of many interactions. In a regime where $n$ is finite, albeit large, our setting will thus be the appropriate one. The forthcoming Corollary \ref{cor:capac} allows us to study the classical and quantum capacity of the effective channel in such a regime.

Let us note in passing that the cascade architecture we are investigating now, in spite of some apparent resemblance, is different from that depicted in Figure \ref{fig2}. While we regard the former as more operationally motivated, the latter is mathematically convenient, as the transmissivities are tuned in such a way as to yield the symmetric convolution $\rho^{\boxplus n}$ at the output.

\begin{figure}[h!]
    \centering
    \includegraphics[width=1\textwidth]{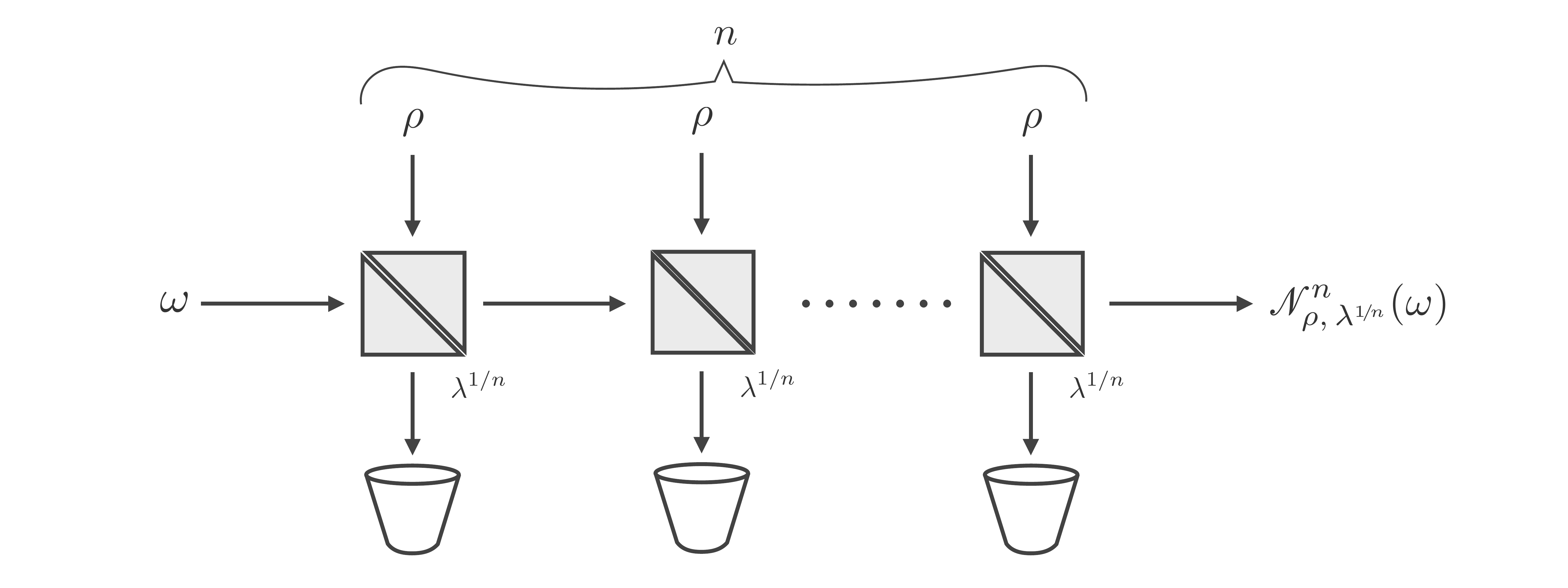}
    \caption{An input state $\omega$ enters an optical fibre modelled by a cascade of $n$ beam splitters with equal transmissivities $\lambda^{1/n}$ and environment states $\rho$.}
    \label{fig:optical-fibre}
\end{figure}

\begin{thm}[(Approximation of thermal attenuators channels by cascades of beam splitters)]\label{thm:approxchann}
Let $\rho$ be a centred $m$-mode quantum state with finite third-order phase space moments $M_3'$, cf. \eqref{eq:ps-moments}, and denote by $\rho_\G$ its Gaussification. Then, 
\begin{align*}
    \left\|\oneN^{n}-\mathcal{N}_{\rho_\G,\,\lambda}\right\|_{\diamond}=\mathcal{O}_{M_3'}\left(n^{-\frac{1}{2(m+1)}}\right)\,,
\end{align*}
where $\|\cdot\|_{\diamond}$ stands for the diamond norm \eqref{eq:diamond}.
\end{thm}

One can further make use of the recently derived continuity bounds under input energy constraints \cite{tightuniform, winter2017energy, shirokov2017tight, shirokov2018energy} in order to find bounds on capacities of the cascade channel $\oneN^n$ in the physically relevant case where the Gaussification $\rho_\G$ of $\rho$ is a thermal state.\footnote{This amounts to assuming that $\rho$ can be brought to its so-called Williamson form (see~\eqref{Williamson-rho} of Section \ref{sec-proofs:QCLT}) by a passive symplectic unitary only.}

\begin{cor} \label{cor:capac}
Consider a single-mode quantum state $\rho$ with finite third-order phase space moments $M_3'$ (cf.~\eqref{eq:ps-moments}) and thermal Gaussification $\rho_\G=\tau_N$ as in~\eqref{tau}. Then, for $\lambda\in [0,1]$, mean photon number $N \coloneqq \tr\big[\rho\, a^\dag a\big]<\infty$, and some input energy $E>0$, the energy-constrained classical and quantum capacity of the cascade channel $\oneN^n$ relative to the canonical Hamiltonian $a^\dag a$ satisfy
\bb
\left|\mathcal{C}\big(\oneN^n,E\big) - g(\lambda E+(1-\lambda)N) + g((1-\lambda)N)\right| \leq \Delta_c (n; N, M_3', \lambda, E)
\label{classical capacity cascade}
\ee
and
\bb
\left|\mathcal{Q}\big(\oneN^n, E\big) - \mathcal{Q}\big(\att, E\big) \right| \leq \Delta_q(n; N, M_3', \lambda, E)\, ,
\label{quantum capacity cascade}
\ee
where $g(x)\coloneqq (x+1) \log(x+1)-x\log x$ (as in \eqref{g}), and $\mathcal{Q}\big(\att, E\big)$ is the quantum capacity of the thermal attenuator.\footnote{An analytical formula for this quantity is currently not known. We report the best lower~\cite{holwer, Noh2020} and upper~\cite{PLOB, MMMM, Rosati2018, Sharma2018, Noh2019} bounds known to date in \eqref{lower bound Q th att 1}--\eqref{lower bound Q th att 2} and \eqref{upper bound Q th att 1}--\eqref{upper bound Q th att 3}, respectively. These results can be used together with \eqref{quantum capacity cascade} to find bounds on $\mathcal{Q}\big(\oneN^n, E\big)$.}

The remainder terms are such that
\bb
\begin{aligned}
\Delta_c(n; N, M_3', \lambda, E) &\leq C(M_3')\, n^{-1/4} \log n\, ,\\
\Delta_q(n; N, M_3', \lambda, E) &\leq C(M_3')\, n^{-1/8} \log n\, .
\end{aligned}
\label{remainder term cascade}
\ee
for some constant $C=C(M_3')$ and all sufficiently large $n\geq n_0\left(\lambda E +(1-\lambda) N, M_3'\right)$.
\end{cor}

The proofs of Theorem \ref{thm:approxchann} and Corollary \ref{cor:capac} are postponed to Section \ref{sec-proofs:cascade}.

\subsection{New results on quantum characteristic functions} \label{subsec:new-results}

In this subsection we state our refined asymptotic analysis of the decay of quantum characteristic functions that we employ in the proofs of our main theorems. For arbitrary quantum states, we have the following asymptotic result on the quantum characteristic function at infinity. It states that the quantum characteristic function can, in absolute value, only attain the value one at zero and decays to zero at infinity. Both these properties do not hold for general classical random variables, see Section \ref{subsec-proofs:decay}.

\begin{prop} \label{prop1}
The quantum characteristic function of an $m$-mode quantum state $\rho$ is a continuous function that is arbitrarily small in absolute value outside of a sufficiently large compact set, i.e.\ $\chi_{\rho}$ belongs to the Banach space $C_0(\mathbb \CC^{m})$ of asymptotically vanishing functions. Moreover, for any $\varepsilon>0$ we have
\bb
\max_{z \in \CC^{m} \backslash B(0,\varepsilon)} \left| \chi_{\rho}(z) \right| <1,
\label{strictly non-lattice chi}
\ee
where $B(0,\varepsilon)\coloneqq \{z\in \CC^m: |z|\leq \varepsilon\}$ denotes a Euclidean ball of radius $\varepsilon$ centred at the origin.
\end{prop}

The proof of Proposition \ref{prop1} is given in Section \ref{subsec-proofs:decay}. Interestingly, we can obtain a much more refined asymptotic on the decay of quantum characteristic functions if we assume that the state has finite second order moments.

\begin{prop} \label{prop:decay-energy}
Let $\rho$ be an $m$-mode state with finite average energy $E\coloneqq \tr\big[\rho\, \left( H_m+\frac{m}{2}\, I \right)\big]$, where we have explicitly accounted for the non-zero energy of the vacuum state. Then, for all $z\in \CC^m$ and all $\delta\in [0,1]$ it holds that
\bb
\left|\chi_\rho(z)\right| \leq 1 - \frac{(1-\delta)^3\delta^{2m-1} \left((2m+1)!!\right)^2}{6\cdot 2^{4m}\cdot E^{2m-1}} \min\left\{|z|^2, \frac{\pi^2 \delta}{4E}\right\} .
\ee
\end{prop}

The proof of Proposition \ref{prop:decay-energy} is given in Section \ref{subsec-proofs:decay}.

\section{New results on quantum characteristic functions: Proofs} \label{sec-proofs:new-results}

Quantum characteristic functions constitute a central tool in our approach. Therefore, the first step in our path towards the quantum Berry--Esseen theorems is to prove the results stated in Section \ref{subsec:new-results}. The structure of this section is as follows:
\begin{itemize}
    \item \textit{Quantum--classical correspondence:} We derive a quantum--classical correspondence of the central limit theorems by showing that the quantum convolution of two arbitrary density operators naturally induces a classical random variable (Section \ref{subsec-proofs:classical-trick}).
    \item \textit{Decay bounds:} We derive new decay estimates and asymptotic properties of the quantum characteristic function at infinity (Section \ref{subsec-proofs:decay}).
\end{itemize}

\subsection{Quantum--classical correspondence} \label{subsec-proofs:classical-trick}

In this section we show that the quantum convolution $\rho \boxplus \sigma$ of any two states $\rho$ and $\sigma$ has a non-negative Wigner function. While the mathematics behind this is known (see e.g.~\cite[Proposition~(1.99)]{Folland},~\cite[Proposition~5]{Cushen1971}, and~\cite[Eq.~(8)]{Jagannathan1987}), we believe that its physical implications have not been appreciated to the extent they deserve.

\begin{lemma} \label{positive W lemma}
Let $\rho$ and $\sigma$ be arbitrary $m$-mode quantum states. Then the Wigner function of their convolution $\rho\boxplus\sigma$ defined by \eqref{boxplus}, with $\lambda=1/2,$ is given by
\bb
W_{\rho\,\boxplus\, \sigma}(z) = \frac{2^m}{\pi^m} \tr\left[ \rho\, \D\big(\sqrt2 z\big) J \sigma J \D\big(\sqrt2 z\big)^\dag\right],
\label{Wigner convolution}
\ee
where $J\coloneqq (-1)^{\sumno_j a_j^\dag a_j}$ is the unitary and self-adjoint operator that implements a phase space inversion (in the sense of eq.~(\ref{W inversion}) below). In particular,
\bb
\label{eq:positivity}
W_{\rho\, \boxplus\, \sigma}(z) \geq 0\qquad \forall\ z\in \CC^m\, .
\ee
\end{lemma}

\begin{proof}
We start by verifying that $J$ actually corresponds to a phase space inversion, in the sense that
\bb
W_{J \rho J}(z) = W_\rho(-z)
\label{W inversion}
\ee
for all $m$-mode quantum states $\rho$ and all $z\in \CC^m$. This follows from the easily verified fact that $J a_j J=-a_j$ for all $j$, which also implies that $J\D(z)J = \D(-z)$. In fact, using \eqref{Wigner 2} we find that
\bbb
W_{J\rho J}(z) = \frac{2^m}{\pi^m} \tr\left[ \D(-z) J\rho J \D(z)\, J\right] = \frac{2^m}{\pi^m} \tr\left[ J\D(z) \rho \D(-z)\right] = W_\rho(-z)\, .
\eee
We now compute
\begin{align*}
    W_{\rho\,\boxplus\, \sigma}(z) &\texteq{\eqref{box W}} 2^m \left(W_\rho\star W_\sigma\right)\left(\sqrt2 z\right) \\
    &\texteq{\eqref{convolution}} 2^m \int d^{2m}u\, W_\rho(u)\, W_\sigma\left(\sqrt2 z- u\right) \\
    &\texteq{\eqref{W inversion}} 2^m \int d^{2m}u\, W_\rho(u)\, W_{J\sigma J}\left(u- \sqrt2 z\right) \\
    &\texteq{\eqref{chi and W displacement}} 2^m \int d^{2m}u\, W_\rho(u)\, W_{\D\left(\sqrt2 z\right) J \sigma J \D\left(\sqrt2 z\right)^\dag}\left(u\right) \\
    &\texteq{\eqref{Plancherel}} \frac{2^m}{\pi^m} \tr\left[ \rho\, \D\big(\sqrt2 z\big) J \sigma J \D\big(\sqrt2 z\big)^\dag \right] .
\end{align*}
The above equalities are labelled by the equation numbers corresponding to the identities that justify them.
\end{proof}


\begin{rem*}
It is not difficult to see that $\lambda=1/2$ is the only special value for which Lemma~\ref{positive W lemma} can hold, i.e.\ such that $W_{\rho\, \boxplus_\lambda \sigma} (z)\geq 0$ for all $m$-mode states $\rho,\sigma$ and for all $z\in \CC^m$. To see why, consider the case where $m=1$ and $\rho,\sigma$ are the first two Fock states. The action of the beam splitter unitary on the annihilation operators, as expressed by \eqref{BS action creation}, leads to the identity $\ketbra{0} \boxplus_\lambda \ketbra{1} = \lambda \ketbra{0} + (1-\lambda) \ketbra{1}$. Using the expression for the Wigner function of Fock states~\cite[Eq.~(4.5.31)]{BARNETT-RADMORE}, we see that
\begin{equation*}
W_{\ket{0}\bra{0}\, \boxplus_\lambda \ket{1}\bra{1}}(z) = W_{\lambda \ket{0}\bra{0} + (1-\lambda) \ket{1}\bra{1}}(z) = \frac{2}{\pi}\, e^{-2|z|^2} \left( \lambda - (1-\lambda)\left(1-4|z|^2\right)\right) .
\end{equation*}
Hence, $W_{\ket{0}\bra{0}\, \boxplus_\lambda \ket{1}\bra{1}}(0)<0$ as soon as $0\leq \lambda < 1/2$. For $1/2<\lambda\leq 1$, we arrive at the same conclusion by looking at the state $\ketbra{1}\boxplus_\lambda \ketbra{0} = \ketbra{0} \boxplus_{1-\lambda} \ketbra{1}$, obtained by sending $\lambda\mapsto 1-\lambda$.
\end{rem*}

We proceed by showing how the above result bridges the gap between classical and quantum central limit theorems. We now fix an $m$-mode quantum state $\rho$, and notice that $\rho^{\boxplus 2n} = (\rho\boxplus \rho)^{\boxplus n}$. Consider the probability density function $f_X\coloneqq W_{\rho\,\boxplus\, \rho}\geq 0$, where positivity holds by \eqref{eq:positivity}. Let $X$ be a random variable with density $f_X$. The mean and covariance matrix of $X$ coincide with those of $\rho\boxplus\rho$, which are in turn the same as those of $\rho$. Hence, at the level of Gaussifications, $f_\G = W_{\rho_\G}$. We write for an i.i.d.~family of random variables $X_i$ with law $f_X$
\begin{align*}
W_{\rho^{\boxplus 2n}}(u) &= W_{\left(\rho\,\boxplus\, \rho\right)^{\boxplus n}}(u) \\
&\texteq{1} n^m W_{\rho\, \boxplus\, \rho}^{\star n}\left( \sqrt{n} u\right) \\
&= n^m f_X^{\star n}\left(\sqrt{n}u\right) \\
&\texteq{2} f_{(X_1+\ldots+X_n)/\sqrt{n}}(u)\,
\end{align*}
where 1 follows from \eqref{box W} and 2 follows from the change of variables $u\mapsto \sqrt{n}u.$
This implies by applying the classical and quantum Plancherel identities (\ref{Plancherel HS norm}) that
\bb
    \begin{split}
    \label{eq:Plancherel}
\left\|\rho^{\boxplus 2n} - \rho_\G\right\|_{2}^2 &= \pi^{-m} \left\|\chi_{\rho^{\boxplus 2n}} - \chi_{\rho_\G}\right\|^2_{\LL^2(\mathbb R^{2m})} \\
&=\pi^{m} \left\|W_{\rho^{\boxplus 2n}} - W_{\rho_\G}\right\|^2_{\LL^2(\mathbb R^{2m})} \\
&=\pi^{m} \left\|f_{(X_1+\ldots +X_n)/\sqrt{n}} - f_\G\right\|^2_{\LL^2(\mathbb R^{2m})}
\end{split}
\ee
which shows that the QCLT is equivalent to a certain CLT for classical i.i.d. random variables. The problem with this approach is that the right classical tool to use here would be an estimate on the rate of convergence of $(X_1+\ldots +X_n)/\sqrt{n}$ to the normal variable $X_\G$ with respect to the $\LL^2$ norm. However, it is known that convergence fails to hold in general, and even under some finiteness of moments assumption there does not seem to be a readily available result in the literature, that is powerful enough to be successfully employed here. Therefore, we do not pursue this route further here.

\subsection{Decay estimates on the quantum characteristic function} \label{subsec-proofs:decay}

Before studying the rate of convergence in the quantum central limit theorem, we show that quantum characteristic functions have the so-called \emp{strict non-lattice property}. To motivate this property, we start by recalling some basic properties of characteristic functions from classical probability theory.

The characteristic function $\chi_X^{\text{cl}}$ of a classical random variable $X$ always attains the value one at zero. However, it can also attain the value one, in absolute value, at any other point. The random variables that exhibit this latter behaviour are precisely those that are \emp{lattice-distributed};\footnote{These are discrete random variables with probability distributions supported on a lattice.} see also \cite[Section~3.5]{durrett_2019}. Examples include the Dirac, Bernoulli, geometric and Poisson distributions.

\medskip
Knowing that $ \left| \chi^{\text{cl}}_X(t) \right| <1$ for all values $t\neq 0$ however does not imply that $\limsup_{t \rightarrow \infty}  \left| \chi^{\text{cl}}_X(t) \right| <1$. This latter condition is known as the \emp{strict non-lattice} property of a random variable. An example of a non-lattice distributed random variable which does not satisfy the strict non-lattice property is as follows.

\begin{ex}[{\cite[Section~3.5]{durrett_2019}}]
Consider an enumeration of the positive rationals $q_1,q_2,... \in \mathbb{Q}_{+}$ with $q_i \le i$ and a non-lattice random variable $X$ defined by
\bbb
\mathbb{P}(X = q_n) = \mathbb{P}(X = -q_n)  = 2^{-(n+1)} .
\eee
One easily shows that the characteristic function of the random variable $X$ is given by $\chi^{\text{cl}}_X(t) = \sum_{i=1}^{\infty} \frac{\cos\left(tq_i\right)}{2^{i}}$. One has $\limsup_{t \rightarrow \infty}  \left| \chi^{\text{cl}}_X(t) \right| =1.$
\end{ex}

We now show the surprising fact that quantum characteristic functions do not exhibit this somewhat pathological behaviour. Instead, for any quantum state $\rho$ it holds that $\limsup_{\vert z \vert \rightarrow \infty}  \left| \chi_{\rho}(z) \right| =0$, as the proof of Proposition~\ref{prop1} below shows.

\begin{proof}[Proof of Proposition~\ref{prop1}]
Thanks to the spectral theorem and by the dominated convergence theorem, it suffices to prove that $\lim_{|z|\to\infty} \chi_{\psi_f}(z)=0$ for all wave function $f\in \LL^2(\RR^m)$, where $\psi_f\coloneqq \ketbra{\psi_f}$, and $\ket{\psi_f}$ is the pure state with wave function $f$. We rephrase this as the requirement that $\chi_{\psi_f}$ belongs to the Banach space $C_0\left(\CC^m\right)$, where the norm on $C_0\left(\CC^m\right)$ is the supremum norm.

We consider smooth compactly supported functions $f$ first. For such functions, the claim follows by combining (i)~Eq.~\eqref{chi psi f}; (ii) the fact that $f$ is normalised, i.e.\ $\int d^mx |f(x)|^2=1$; and (iii)~the Riemann--Lebesgue lemma. For general $f \in \LL^2(\RR^m)$, the result then follows by a density argument: for an arbitrary $f \in \LL^2(\RR^m)$ there is a sequence of smooth and compactly supported functions $f_n \in C_c^{\infty}(\RR^m)$ converging to $f \in \LL^2(\RR^m)$, so that 
\bbb
\begin{aligned}
&\sup_{z \in \CC^m} \left| \left(\chi_{\psi_{f_n}} - \chi_{\psi_f}\right)(z) \right| \\
&\qquad = \sup_{z \in \CC^m} \left| \braket{\psi_{f_n}|\D(z)|\psi_{f_n}} - \braket{\psi_f|\D(z)|\psi_f} \right| \\
&\qquad \le \sup_{z \in \CC^m} \Big\{ \left| \braket{\psi_{f_n}|\D(z)|\psi_{f_n}} - \braket{\psi_{f_n}|\D(z)|\psi_f} \right| + \left| \braket{\psi_{f_n}|\D(z)|\psi_{f}} - \braket{\psi_f|\D(z)|\psi_f} \right| \Big\} \\
&\qquad \le 2 \left\| \ket{\psi_{f_n}} - \ket{\psi_f}\right\| = 2 \left\|f_n - f \right\|_{\LL^2(\RR^m)}\tends{}{n \to \infty} 0 .
\end{aligned}
\eee
Since $C_0(\CC^m)$ is a Banach space and $\chi_{\psi_{f_n}} \in C_0(\CC^m)$, this implies that also the limit $\chi_{\psi_f} \in C_0(\CC^m)$. Thus, to complete the proof of \eqref{strictly non-lattice chi} it suffices to show that for every $\eps>0$ and any $z \in \CC^m \backslash B(0,\varepsilon)$ one has that $\left| \chi_{\psi_f}(z) \right| <1.$ If this were not the case, then $\ket{\psi_f}$ would be an eigenvector of the displacement operator $\D(z)$. This is well known to be impossible, see e.g.~\cite[Lemma~10]{G-dilatable}.
\end{proof}

For a given state $\rho$ and some fixed $\varepsilon>0$, Proposition~\ref{prop1} tells us that there exists a constant $\eta(\rho,\varepsilon)<1$ such that $\max_{z\in \CC^m\setminus B(0,\varepsilon)} \left|\chi_\rho(z)\right|\leq \eta(\rho,\varepsilon)$ (cf.~\eqref{strictly non-lattice chi}). However, the problem of characterising the quantity $\eta(\rho,\varepsilon)$ in terms of some physically meaningful property of the state $\rho$ remains. To this end, a natural candidate turns out to be the energy of the state. To see why this is the case, consider the following simple example.

\begin{ex}[(Squeezed states)] \label{ex:squeezed-state}
For every $z\in \CC^m$ and every $\delta\in (0,1)$ there is a (Gaussian) state $\rho_\G$ of mean photon number $\tr\left[\rho_\G H_m \right] \leq \frac{t^2}{8 \ln \frac{1}{1-\delta}} - \frac14$ such that $\left|\chi_{\rho_\G}(z)\right|\geq 1-\delta$.

To see that this is the case, up to the application of passive symplectic unitaries, it suffices to consider the case $z=(t,0,\ldots, 0)$, where $t>0$. Consider the `squeezed' Gaussian state defined by the characteristic function
\bb
\chi_{\rho_\G}(z) \coloneqq e^{-\frac{\eta}{2} (\Re z_1)^2 - \frac{1}{2\eta} (\Im z_1)^2 - \frac12 \sumno_{j>1} |z_j|^2}\, ,
\ee
where we set $\eta\coloneqq \min\left\{\frac{2}{t^2} \ln \frac{1}{1-\delta},\ 1\right\}>0$. 
The mean photon number of $\rho_\G$ is well known to be given by $\tr\left[\rho_\G H_m \right] = \frac14 \left( \eta + \frac1\eta\right) - \frac12 \leq \frac{1}{4\eta}-\frac14$, where we used the fact that $\eta\leq 1$.
\end{ex}


The above example shows that any estimate on $\eta(\rho,\varepsilon)$ can be reasonably expected to depend on the energy.
We now show that our preliminary work on the quantum--classical correspondence allows us to derive a general upper estimate for $|\chi_\rho(z)|$ at any designated point $z\in \CC^m$ in terms of the energy of the state $\rho$. For this purpose, we draw upon some important mathematical results from the well-developed theory of \textit{classical} characteristic functions. Proposition \ref{prop:decay-energy}, whose proof we present now, implies e.g.\ that for a one-mode state $\rho$, we can take $\eta(\rho,\varepsilon) = 1 - \frac{c}{E}\, \min\left\{\eps^2, \frac{C}{E}\right\}$, where $E$ is the energy of $\rho$, and $c,C$ are universal constants.

\begin{proof}[Proof of Proposition \ref{prop:decay-energy}]
Denoting as usual with $|z|$ the Euclidean norm~\eqref{Euclidean norm} of $z\in \CC^m$, we write the following chain of inequalities.
\begin{align*}
\left|\chi_\rho(z)\right| &\texteq{1} \left| \chi_{\rho\,\boxplus\, \rho}\left(\sqrt2 z\right) \right|^{1/2} \\
&\texteq{2} \left| \chi_{X(\rho\,\boxplus\, \rho)}^{\text{cl}}\left(\sqrt2 z\right) \right|^{1/2} \\
&\textleq{3} \left( 1 - \frac{2(1-\delta)^3\delta^{2m-1} \left((2m+1)!!\right)^2}{3\cdot 2^{4m}\cdot E^{2m-1}} \min\left\{|z|^2, \frac{\pi^2 \delta}{4E}\right\}\right)^{1/2} \\
&\textleq{4} 1 - \frac{2(1-\delta)^3\delta^{2m-1} \left((2m+1)!!\right)^2}{6\cdot 2^{4m}\cdot E^{2m-1}} \min\left\{|z|^2, \frac{\pi^2 \delta}{4E}\right\} .
\end{align*}
Here, 1 is an application of the quantum convolution rule (cf.\ the $n=2$ case of \eqref{chi CH state}). In 2 we introduced the classical random vector $X(\rho\boxplus \rho)$ taking values in $\CC^m$,  with probability distribution given by the Wigner function $W_{\rho\, \boxplus\, \rho}$, which is everywhere non-negative by Lemma \ref{positive W lemma}. The inequality in 3, which is the non-trivial one, follows from \cite[Corollary~2.7.2]{Ushakov}: we set $a\coloneqq \sup_{z\in \CC^m} W_{\rho\, \boxplus\, \rho}(z) \leq \frac{2^m}{\pi^m}$, with the latter estimate coming from~\eqref{Wigner 2}, and $\alpha=2$, so that 
\bbb
\gamma_\alpha=\gamma_2=\int d^m z\, |z|^2 \,W_{\rho\, \boxplus\, \rho}(z)\, \texteq{\eqref{energy Wigner}}\, \tr\left[ (\rho\boxplus \rho) \left( H_m + \frac{m}{2}\, I \right)\right]\, \texteq{\eqref{energy convolution}}\, \tr\left[ \rho \left( H_m +\frac{m}{2}\, I\right) \right]  \eqqcolon E\, ;
\eee
also, we substituted $m\mapsto 2m$, because our phase space $\CC^m$ has real dimension $2m$; finally, we used the well-known formula $\Gamma(m+1/2) = \sqrt\pi\, 2^{-m} (2m-1)!!$, where $(\cdot)!!$ is the bi-factorial. Lastly, the inequality in 4 is just an application of the elementary estimate $\sqrt{1-x}\leq 1-\frac{x}{2}$ for $0\leq x <1$.
\end{proof}

\begin{rem*}
In \cite[Section~2.7]{Ushakov}, several other estimates for $\left| \chi_X^{\text{cl}}(t)\right|$ are derived. While we decided to stick to the simplest one, as it is already very instructive, it is possible to substantially improve over it, e.g.\ by resorting to non-isotropic estimates (cf.\ for instance \cite[Theorem~2.7.14]{Ushakov}). Notably, our quantum--classical correspondence allows us to translate \textit{all} of these inequalities to the quantum setting, up to an irrelevant factor of $1/2$ in the associated constants (see step 4 in the above proof). We do not pursue this approach further, though we want to stress that it immediately leads to a plethora of further results.
\end{rem*}


\section{Quantitative bounds in the QCLT: Proofs} \label{sec-proofs:QCLT}

In this section, we provide proofs of the convergence rates in our quantum Berry--Esseen theorems. We also provide proofs of some of the statements in Section \ref{subsec:applications} on the convergence rate for cascades of beam splitters converging to thermal attenuator channels.

\medskip
\smallsection{Outline of this section:} To fix ideas, we give a high-level outline of our proofs:
\begin{itemize}
    \item \textit{Williamson form}: We apply a suitable symplectic unitary to the state, so as to make the Hessian of its characteristic function diagonal and larger than the identity. Subsequently, we use the quantum Plancherel identity to express the difference of the convolved state and its Gaussification in Hilbert--Schmidt norm as a difference of quantum characteristic functions in $\LL^2$ norm (Section \ref{subsec-proofs:preliminaries}).
    
    \item \textit{Local-tail decomposition}: We then split the integral of the $\LL^2$ norm of the difference of the quantum characteristic functions of the convolved state and the Gaussification of the original state into a regime around zero (Lemma \ref{lem:tech}), in which we can control the behaviour of the quantum characteristic function by its Taylor expansion, and a tail-regime in which we estimate the difference using Proposition \ref{prop1}. The error in the Taylor expansion is controlled by the phase space moments of the state, cf.~Lemma \ref{lem:tech2}.
    
    \item \textit{Hilbert--Schmidt convergence}: We implement the above ideas to prove Theorems \ref{thm:QBE'} and \ref{thm:QBElow}, and Proposition \ref{propnon-i.i.d.} (Section \ref{subsec-proofs:convrates}).
    
    \item \textit{Trace norm and entropic convergence:} We then use the preservation of the boundedness of the second moment under quantum convolutions to obtain a quantitative estimate of convergence in trace distance, employing Markov's inequality and the Gentle Measurement Lemma~\cite{VV1999}, and in relative entropy, using entropic continuity bounds \cite{tightuniform} (Section \ref{sec-proofs:traceentr}).
    
    \item \textit{Convergence rates for cascades of beam splitters:} In the final subsection, we prove the results claimed in Section \ref{subsec:applications}, namely convergence rates for cascades of beam splitters converging to thermal attenuator channels (Section \ref{sec-proofs:cascade}).
\end{itemize}

\subsection{Preliminary steps} \label{subsec-proofs:preliminaries}

\subsubsection{Williamson form} \label{subsubsec:Williamson}
Let $\rho$ be a centred $m$-mode quantum state with finite second moments, as in the Cushen--Hudson theorem. It is known that one can find a symplectic unitary $V$ and numbers $\nu_1,\ldots, \nu_m\geq 1$ such that 
\bb\label{Williamson-rho}
\rho' \coloneqq V^\dag \rho V
\ee
satisfies
\bb
\chi_{\rho'}(z) = 1 - \frac12 \sum_j  \nu_j |z_j|^2 + o\left(|z|^2\right)\qquad (z\to 0) .
\label{Williamson chi}
\ee
With a slight abuse of terminology, we will call $\rho'$ the \emp{Williamson form} of $\rho$ \cite{willy}. Bringing a state to its Williamson form allows us to assume that (i) the smallest eigenvalue of its covariance matrix is at least one. Also, (ii) the transformation in \eqref{Williamson-rho} does not change the first moments of the state, so that if $\rho$ is centred then $\rho'$ remains centred. Finally, (iii) the same unitary $V$ brings not only $\rho$ but also its Gaussification $\rho_\G$ to their Williamson forms simultaneously, so that
\bb
\chi_{\rho'_\G}(z) = \exp \left[-\frac12 \sumno_j \nu_j |z_j|^2 \right] ,\qquad W_{\rho'_\G}(z) = \left( \frac2\pi\right)^m \exp\left[ -2\sumno_j \nu_j |z_j|^2 \right] .
\label{Williamson chi Gaussian}
\ee
holds as well. Thanks to the covariance of the quantum convolution with respect to symplectic unitaries \eqref{symmetric boxplus covariant symplectic}, we see that
\begin{align*}
    \left\|\rho^{\boxplus n} - \rho_\G \right\|_2 &= \left\|V^\dag \left(\rho^{\boxplus n} - \rho_\G \right) V \right\|_2 \\
    &= \left\|(\rho')^{\boxplus n} - \rho'_\G \right\|_2 .
\end{align*}
Combining this with the quantum Plancherel identity~\eqref{Plancherel HS norm} yields
\begin{align}
    \left\|\rho^{\boxplus n} - \rho_\G \right\|^2_2 &= \int \frac{d^{2m}z}{\pi^m}\, \left|\chi_{(\rho')^{\boxplus n}}(z) - \chi_{\rho'_\G}(z)\right|^2 \label{HS distance QCLT chi} \\
    &= \pi^m \int d^{2m}z\, \left( W_{(\rho')^{\boxplus n}}(z) - W_{\rho'_\G}(z) \right)^2 . \label{HS distance QCLT W}
\end{align}
In short, when estimating any unitarily invariant distance of $\rho^{\boxplus n}$ from its limit $\rho_\G$, we can assume without loss of generality that all states are in their Williamson forms. When the Hilbert--Schmidt norm is employed, we can compute the distance as an $\LL^2$ norm at the level of characteristic functions, or equivalently at that of Wigner functions.

\subsubsection{Local-Tail decomposition} \label{subsubsec:local-tail}
We continue with an important technical lemma that reduces the convergence in the quantum central limit theorem to the behaviour of the quantum characteristic function around zero.

\begin{lemma} \label{lem:tech}
Let $\rho$ be an $m$-mode quantum state with finite second-order phase space moment. Without loss of generality, we assume that $\rho$ is centred and in Williamson form, and that its Gaussification $\rho_\G$ has characteristic function as in \eqref{Williamson chi Gaussian}. Then for every $\varepsilon>0$ we have that
\begin{equation} \label{eq:firstid}
\begin{aligned}
\left\|\rho^{\boxplus n}-\rho_\G\right\|_{2}^2 &=\frac{1}{\pi^{m}}\int_{|z|\leq \sqrt{n}\, \varepsilon} d^{2m}z \left|  \chi_{\rho}\left(\tfrac{z}{\sqrt{n}}\right)^n - e^{-\frac12 \sumno_j \nu_j |z_j|^2} \right|^2 + \mathcal O\left(n^{-\infty}\right)
\end{aligned}
\end{equation}
as $n\to\infty$. If $\rho$ has also finite third-order phase space moments, then
\begin{equation} \label{eq:secondid}
\begin{aligned}
\|\rho^{\boxplus n}-\rho_\G\|_2 &\le \frac{\sqrt{m(m+1)(m+2)}}{6}\,\frac{\Vert D^3 \chi_{\rho}(0) \Vert}{\sqrt{n}}+\mathcal O(n^{-\infty})\\
&\quad +\frac{1}{\pi^{m/2}}\left(\int_{|z|\leq \sqrt{n}\, \varepsilon} d^{2m}z \left| \chi_\rho \left(\tfrac{z}{\sqrt{n}}\right)^n - e^{-\frac12 \sumno_j \nu_j |z_j|^2} \left(1+\tfrac{1}{6\sqrt{n}} D^3\chi_\rho (0)\left(z^{\times 3}\right)\right) \right|^2  \right)^{1/2} ,
\end{aligned}
\end{equation}
where the Fr\'echet derivative of $\chi_\rho$ is defined by \eqref{Frechet}.
\end{lemma}

\begin{proof}
The first identity \eqref{eq:firstid} follows along the lines of the second one \eqref{eq:secondid} and so we focus on verifying the latter. Using the quantum Plancherel identity \eqref{Plancherel HS norm} and the relation \eqref{box chi}, we apply the triangle inequality and split the integration domain into two disjoint sets such that
\begin{equation} \label{eq:estm} \begin{aligned}
\pi^{\frac{m}{2}}\|\rho^{\boxplus n}-\rho_\G\|_{2} & =\left(\int d^{2m}z   \left\lvert \chi_{\rho}\left(\tfrac{z}{\sqrt{n}}\right)^n - e^{-\frac12 \sumno_j \nu_j |z_j|^2} \right\rvert^2 \ \right)^{1/2} \\
& \le \left(\int_{|z|\leq \sqrt{n}\, \varepsilon} d^{2m}z \left| \chi_{\rho}\left(\tfrac{z}{\sqrt{n}}\right)^n - e^{-\frac12 \sumno_j \nu_j |z_j|^2} \left(1+\tfrac{1}{6\sqrt{n}}D^3\chi_{\rho}(0)\left(z^{\times 3}\right)\right) \right|^2 \right)^{1/2}\\
&\quad +\left(\int_{|z|>\sqrt{n}\, \varepsilon} d^{2m}z \left| \chi_{\rho}\left(\tfrac{z}{\sqrt{n}}\right)^n - e^{-\frac12 \sumno_j \nu_j |z_j|^2} \left(1+\tfrac{1}{6\sqrt{n}}D^3\chi_{\rho}(0)\left(z^{\times 3}\right)\right)  \right|^2 \right)^{1/2}  \\
&\quad +\frac{1}{\sqrt{n}}\left(\int d^{2m}z \left| e^{-\frac12 \sumno_j \nu_j |z_j|^2} \left(\frac{1}{6}D^3\chi_{\rho}(0)\left(z^{\times 3}\right)\right)\right|^2 \right)^{1/2}\, .
\end{aligned} \end{equation}
The last term on the rightmost side of \eqref{eq:estm} can be estimated explicitly using spherical coordinates. Namely, combining the fact that the coefficients appearing in the Williamson form satisfy $\nu_j\geq 1$ with the bound $\left| D^3\chi_{\rho}(0)\left(z^{\times 3}\right)\right| \leq \left\| D^3\chi_{\rho}(0) \right\| |z|^3$, we obtain that
\begin{align*}
\int d^{2m}z\left\lvert e^{-\frac12 \sumno_j \nu_j |z_j|^2} \left(\frac{1}{6}D^3\chi_{\rho}(0)(z^{\times 3})\right)\right\rvert^2 &\le \frac{ \operatorname{vol}(\mathbb S^{2m-1})  \Vert D^3 \chi_{\rho}(0) \Vert^2}{36} \int_{0}^{\infty}dr\, e^{-r^2} r^{2m+5}  \\
&= \frac{\Gamma\left(m+3\right) \operatorname{vol}(\mathbb S^{2m-1})}{72}\, \Vert D^3 \chi_{\rho}(0) \Vert^2 \\
&= \frac{\pi^m}{36}\, m(m+1)(m+2)\, \Vert D^3 \chi_{\rho}(0) \Vert^2\, ,
\end{align*}
where we used that $\int_0^{\infty}dr\, e^{-r^2} r^{2m+5}= \frac{\Gamma(m+3)}{2}$, and recalled the expression $\operatorname{vol}\left(\mathbb S^{N-1}\right)=\frac{2\pi^{N/2}}{\Gamma(N/2)}$ for the volume of the $(N-1)$-sphere. Furthermore, the second-to-last term in \eqref{eq:estm} can be shown to be exponentially small. In fact,
\begingroup
\allowdisplaybreaks
\begin{align*}
&\int_{|z|> \sqrt{n}\, \varepsilon} d^{2m}z \left| \chi_{\rho} \left(\tfrac{z}{\sqrt{n}}\right)^n - e^{-\frac12 \sumno_j \nu_j |z_j|^2} \left(1+\tfrac{1}{6\sqrt{n}}D^3\chi_{\rho}(0)\left(z^{\times 3}\right)\right) \right|^2 \\
&\qquad \le 2 \int_{|z|> \sqrt{n}\, \varepsilon} d^{2m}z \left(\left| \chi_{\rho}\left(\tfrac{z}{\sqrt{n}}\right)^n \right|^2 + \left| e^{-\frac12 \sumno_j \nu_j |z_j|^2} \left(1+\tfrac{|z|^3}{6\sqrt{n}}\left\|D^3\chi_{\rho}(0)\right\|\right) \right|^2 \right)  \\
&\qquad\le 2 \left( \sup_{|z|>\varepsilon} \left| \chi_\rho(z) \right| \right)^{2(n-1)} \int_{|z|> \sqrt{n}\, \varepsilon} d^{2m}z \left| \chi_{\rho}\left(\tfrac{z}{\sqrt{n}}\right) \right|^2 \\
&\qquad \quad + 2 \int_{|z|> \sqrt{n}\, \varepsilon} d^{2m}z \left| e^{-\frac12 \sumno_j \nu_j |z_j|^2} \left(1+\tfrac{\vert z \vert^3}{6\sqrt{n}}\left\| D^3\chi_{\rho}(0)\right\| \right)  \right|^2 \\
&\qquad= 2 n^m \left( \sup_{|z|>\varepsilon} \left| \chi_\rho(z) \right| \right)^{2(n-1)} \int_{|u|>\varepsilon} d^{2m}u \left| \chi_{\rho}(u) \right|^2 \\
&\qquad \quad + 2 \int_{|z|> \sqrt{n}\, \varepsilon} d^{2m}z \left| e^{-\frac12 \sumno_j \nu_j |z_j|^2} \left(1+\tfrac{\vert z \vert^3}{6\sqrt{n}}\left\| D^3\chi_{\rho}(0)\right\| \right)  \right|^2 \\
&\qquad\textleq{1} 2 n^m \left( \sup_{|z|>\varepsilon} \left| \chi_\rho(z) \right| \right)^{2(n-1)} + 2\operatorname{vol}\left(\mathbb S^{2m-1}\right)\int_{\sqrt{n}\, \varepsilon}^{\infty} dr\ e^{-r^2} r^{2m-1} \left(1+\frac{r^3\Vert D^3\chi_{\rho}(0)\Vert}{6\sqrt{n}}\right)^2 \\
&\qquad\textleq{2} 2 n^m \left( \sup_{|z|>\varepsilon} \left| \chi_\rho(z) \right| \right)^{2(n-1)} + \frac{4 \pi^{m}}{(m-1)!}\, e^{-\frac{\varepsilon^2}{2}\, n} \int_{\sqrt{n}\,\varepsilon}^\infty dr \, e^{-\frac{r^2}{2}} r^{2m-1}\left(1+\frac{r^3\left\| D^3\chi_{\rho}(0)\right\|}{6\sqrt{n}}\right)^2\, ,
\end{align*}
\endgroup
where in~1 we used the fact that the $L^2$ norm of the characteristic function is at most one and switched to spherical coordinates to compute the second integral. In~2, instead, we estimated $e^{-r^2}< e^{-\frac{\varepsilon^2}{2}\, n} e^{-\frac{r^2}{2}}$ for $r>\sqrt{n}\,\varepsilon$. Note that the first addend goes to zero faster than any inverse power of $n$ for $n\to\infty$ by Proposition \ref{prop1}. The second decays exponentially, essentially because the integral is bounded in $n$ (in fact, it tends to $0$ as $n\to\infty$).
This concludes the proof.
\end{proof}

The first term on the right-hand side of \eqref{eq:secondid} features an explicit dependence on $n$, while the second decays faster than any inverse power of $n$. Therefore, all that is left to do is to estimate the third term, which can be done by looking at the behaviour of the characteristic function in a neighbourhood of the origin. The first step in this direction, rather unsurprisingly, involves a Taylor expansion of $\chi_\rho$ around $0$. In the subsequent lemma we record various important estimates of this sort, which will play a key role in the proofs of our quantum Berry--Esseen theorems.

\begin{lemma} \label{lem:tech2}
For $\varepsilon>0$ and $k\in [0,\infty)$, let $\rho$ be an $m$-mode state with finite phase space moments of order up to $k$ (namely, with the notation of Definition \ref{def:ps-moments}, assume that $M'_k(\rho,\varepsilon)<\infty$). Then for all $z\in \CC^m$ with $|z|\leq \sqrt{n}\, \varepsilon$ it holds that
\bb
\left| \chi_\rho\left(\tfrac{z}{\sqrt{n}}\right) - \sum_{h=0}^{\ceil{k}-1} \frac{1}{h!\, n^{h/2}}\, D^h\chi_\rho(0)\left( z^{\times h} \right) \right| \leq \frac{1}{\floor{k}!}\binom{2m+\floor{k}-1}{\floor{k}}\, M'_k(\rho,\varepsilon)\, \frac{|z|^k}{n^{k/2}}\, . 
\label{Taylor}
\ee
In particular, if $\rho$ is centred and in Williamson form,
\begin{align}
\left| \chi_{\rho}\left(\tfrac{z}{\sqrt{n}}\right) - 1 \right|&\leq \frac{m(2m\!+\!1)}{2}\,M'_2(\rho,\varepsilon)\,\frac{|z|^2}{n}\, , \label{Taylor 1} \\
\left| \chi_{\rho}\left(\tfrac{z}{\sqrt{n}}\right) - 1 + \frac{1}{2n} \sumno_j \nu_j |z_j|^2\right| &\leq \frac{m(2m\!+\!1)}{2}\, M'_{2+\alpha}(\rho,\varepsilon)\, \frac{|z|^{2+\alpha}}{n^{1+\frac{\alpha}{2}}} \, , \label{Taylor alpha} \\
\left| \chi_{\rho}\left(\tfrac{z}{\sqrt{n}}\right) - 1 + \frac{1}{2n} \sumno_j \nu_j |z_j|^2\right| &\leq \frac{m(m\!+\!1)(2m\!+\!1)}{9}\, M'_3(\rho,\varepsilon)\, \frac{|z|^3}{n^{3/2}} \, , \label{Taylor 2} \\
\left| \chi_{\rho}\left(\tfrac{z}{\sqrt{n}}\right) - 1 + \frac{1}{2n} \sumno_j \nu_j |z_j|^2 - \frac{1}{6n^{3/2}} D^3 \chi_\rho(0)\left( z^{\times 3}\right) \right| &\leq \frac{m(m\!+\!1)(2m\!+\!1)(2m\!+\!3)}{144}\, M'_4(\rho,\varepsilon)\, \frac{|z|^4}{n^2}\, , \label{Taylor 3}
\end{align}
depending on what phase space moments are finite. In \eqref{Taylor alpha}, we assumed that $\alpha\in(0,1)$.
\end{lemma}

The estimate in \eqref{Taylor} follows immediately from using H\"older continuity of the derivative.

\subsection{Proofs of convergence rates in Hilbert--Schmidt distance} \label{subsec-proofs:convrates}

We start with the proof of Theorem \ref{thm:QBE'} assuming fourth-order moments.

\begin{proof}[Proof of Theorem~\ref{thm:QBE'}]
By the discussion in Section \ref{subsubsec:Williamson}, we can assume that $\rho$ is in Williamson form, namely, that its characteristic function satisfies \eqref{Williamson chi}, with $\nu_1,\ldots, \nu_m\geq 1$. Since $M'_2(\rho,\varepsilon)$ is monotonically non-decreasing in $\varepsilon$, for any fixed $\mu\in (0,2)$ we can chose $\varepsilon>0$ small enough so that for any $z\in B \coloneqq B\left(0,\sqrt{n}\, \eps\right)$ it holds that
\bb
\frac{m(2m+1)}{2}\, M'_2(\rho,\varepsilon)\, \frac{|z|^2}{n} \leq \frac{m(2m+1)}{2}\, \varepsilon^2\, M'_2(\rho,\varepsilon)\leq \frac{\mu}{2}\, .
\label{eq:mu}
\ee
Looking at \eqref{Taylor 1}, this implies that $2\left|1-\chi_\rho \left( \tfrac{z}{\sqrt{n}} \right) \right|\leq \mu$. Now, for $x\in \CC$ with $|x|<2$ define the function
\bb
a(x) \coloneqq -\frac{4}{x^2}\left(\log\left(1-\frac{x}{2}\right) + \frac{x}{2} \right) = \sum_{\ell=0}^{\infty} \frac{x^\ell}{2^\ell\, \ell}\, .
\label{a}
\ee
Substituting $x=2\left( 1-\chi_\rho \left( \tfrac{z}{\sqrt{n}} \right)\right)$, we then have that
\bb
\begin{aligned}
\left| \log \left( \chi_\rho \left( \tfrac{z}{\sqrt{n}} \right) \right) + \left(1-\chi_\rho\left( \tfrac{z}{\sqrt{n}}\right) \right) \right| &= \left| -\left(1-\chi_{\rho}\left(\tfrac{z}{\sqrt{n}}\right)\right)^2 a\left( 2\left(1-\chi_{\rho}\left(\tfrac{z}{\sqrt{n}}\right)\right)\right) \right| \\
&\leq \frac{m^2(2m+1)^2}{4}\, M'_2(\rho,\varepsilon)^2\, a(\mu)\, \frac{|z|^4}{n^2}\, ,
\end{aligned}
\label{eq:estim0}
\ee
where to deduce the last inequality we observed that $|x|\leq \mu$ implies that $|a(x)|\leq a(\mu)$. Then, thanks to \eqref{eq:estim0} and \eqref{Taylor 2}, an application of the triangle inequality yields
\bb
\begin{aligned}
&\left| \log \left( \chi_{\rho}\left(\tfrac{z}{\sqrt{n}}\right)^n \right) + \frac12 \sumno_j \nu_j |z_j|^2 \right| \\
&\qquad \leq n \left| \log \left( \chi_{\rho}\left(\tfrac{z}{\sqrt{n}}\right) \right) + 1 - \chi_\rho \left( \tfrac{z}{\sqrt{n}}\right) \right| + n \left| \chi_\rho \left( \tfrac{z}{\sqrt{n}}\right) -1 + \frac12 \sumno_j \nu_j |z_j|^2 \right| \\
&\qquad \leq \frac{m^2(2m+1)^2}{4}\, M'_2(\rho,\varepsilon)^2\, a(\mu)\, \frac{|z|^4}{n} + \frac{m(m+1)(2m+1)}{9}\, M'_3(\rho,\varepsilon)\, \frac{|z|^3}{\sqrt{n}} \\
&\qquad \leq \frac{C_1 |z|^3}{\sqrt{n}}\, ,
\end{aligned}
\label{eq:estm2}
\ee
where for fixed $m$ the constant $C_1$ depends only on $M'_3$ (remember that $M'_2\leq M'_3$ by construction). Using again \eqref{eq:estim0} but now in conjunction with \eqref{Taylor 3}, by a swift application of the triangle inequality we see that
\bb
\begin{aligned}
&\left| \log \left( \chi_{\rho}\left(\tfrac{z}{\sqrt{n}}\right)^n \right) + \frac12 \sumno_j \nu_j |z_j|^2 -\frac{1}{6\sqrt{n}}\, D^3\chi_\rho (0)\left( z^{\times 3} \right) \right| \\
&\qquad \leq n \left| \log \left( \chi_{\rho}\left(\tfrac{z}{\sqrt{n}}\right) \right) + 1 - \chi_\rho \left( \tfrac{z}{\sqrt{n}}\right) \right| + n \left| \chi_\rho \left( \tfrac{z}{\sqrt{n}}\right) - 1 + \frac12 \sumno_j \nu_j |z_j|^2 -\frac{1}{6\sqrt{n}}\, D^3\chi_\rho (0)\left( z^{\times 3} \right) \right| \\
&\qquad \textleq{2} \frac{m^2(2m+1)^2}{4}\, M'_2(\rho,\varepsilon)^2\, a(\mu)\, \frac{|z|^4}{n} + \frac{m(m+1)(2m+1)(2m+3)}{144}\, M'_4(\rho, \varepsilon)\, \frac{|z|^4}{n} \\
&\qquad \leq \frac{C_2 |z|^4}{n}\, ,
\end{aligned}
\label{eq:estm1}
\ee
where for fixed $m$ the constant $C_2$ depends only on $M'_4$ (remember that $M'_2\leq M'_4$ by construction). We now estimate
\bbb
\begin{aligned}
&e^{\frac12 \sum_j \nu_j |z_j|^2} \left|\chi_\rho \left(\tfrac{z}{\sqrt{n}}\right)^n - e^{-\frac12 \sumno_j \nu_j |z_j|^2} \left( 1+\frac{1}{6\sqrt{n}}\, D^3\chi_\rho(0)\left( z^{\times 3}\right) \right) \right| \\
&\qquad = \left| \exp\left( \log \left( \chi_\rho \left(\tfrac{z}{\sqrt{n}}\right)^n \right) + \frac12 \sumno_j \nu_j |z_j|^2 \right) - \left( 1+\frac{1}{6\sqrt{n}}\, D^3\chi_\rho(0)\left( z^{\times 3}\right) \right) \right| \\
&\qquad \textleq{1} \left| \exp\left( \log \left( \chi_\rho \left(\tfrac{z}{\sqrt{n}}\right)^n \right) + \frac12 \sumno_j \nu_j |z_j|^2 \right) - \left( 1 + \log \left( \chi_\rho \left(\tfrac{z}{\sqrt{n}}\right)^n\right) + \frac12 \sum_j \nu_j |z_j|^2 \right) \right| \\
&\qquad \quad + \left| \log \left( \chi_\rho \left(\tfrac{z}{\sqrt{n}}\right)^n\right) + \frac12 \sum_j \nu_j |z_j|^2 - \frac{1}{6\sqrt{n}}\, D^3\chi_\rho(0)\left( z^{\times 3}\right) \right| \\
&\qquad \textleq{2} \frac{C_1^2 |z|^6}{n}\, e^{C_1 |z|^3/\sqrt{n}} + \frac{C_2 |z|^4}{n} \\
&\qquad \textleq{3} \frac1n \, e^{\frac14 |z|^2} \left( C_1^2 |z|^6 + C_2 |z|^4 \right) .
\end{aligned}
\eee
Here, 1~follows simply by the triangle inequality. In~2, we (i)~observed that $\left| e^u - (1+u)\right|\leq |u|^2 e^{|u|}$; (ii)~operated the substitution $u=\log \left( \chi_\rho\left( \tfrac{z}{\sqrt{n}}\right)^n\right) + \frac12 \sum_j \nu_j |z_j|^2$; (iii)~noted that $\RR\ni x\mapsto x^2 e^x$ is a monotonically increasing function; and (iv)~used the fact -- proved in \eqref{eq:estm2} -- that $|u|\leq \frac{C_1 |z|^3}{\sqrt{n}}$. Finally, in~3 we remembered that $|z|\leq \sqrt{n}\, \varepsilon$ and assumed that $\varepsilon>0$ is small enough so that $\varepsilon C_1\leq \frac14$. Now, since $\nu_1,\ldots, \nu_m\geq 1$, we can rephrase the above estimate as
\bb
\left|\chi_\rho \left(\tfrac{z}{\sqrt{n}}\right)^n - e^{-\frac12 \sum_j \nu_j |z_j|^2} \left( 1+\frac{1}{6\sqrt{n}}\, D^3\chi_\rho(0)\left( z^{\times 3}\right) \right) \right| \leq \frac1n\, e^{-\frac14 |z|^2} \left( C_1^2 |z|^6 + C_2 |z|^4 \right) .
\label{eq:estim3}
\ee
Upon integration, \eqref{eq:estim3} naturally yields an upper bound for the second term on the right-hand side of \eqref{eq:secondid}. We obtain that
\bb
\begin{aligned}
&\int_{|z|\leq \sqrt{n}\, \varepsilon} d^{2m} z \left|\chi_\rho \left(\tfrac{z}{\sqrt{n}}\right)^n - e^{-\frac12 \sum_j \nu_j |z_j|^2} \left( 1+\frac{1}{6\sqrt{n}}\, D^3\chi_\rho(0)\left( z^{\times 3}\right) \right) \right|^2 \\
&\qquad \leq \frac{1}{n^2} \int_{|z|\leq \sqrt{n}\, \varepsilon} d^{2m} z\ e^{-\frac12 |z|^2} \left( C_1^2 |z|^6 + C_2 |z|^4 \right)^2 \\
&\qquad \leq \frac{1}{n^2} \int d^{2m} z\ e^{-\frac12 |z|^2} \left( C_1^2 |z|^6 + C_2 |z|^4 \right)^2 \\
&\qquad \texteq{4} \frac{1}{n^2} \operatorname{vol}\left(\mathbb{S}^{2m-1}\right) \int_0^\infty dr\ r^{2m-1} e^{-\frac12 r^2} \left( C_1^2 r^6 + C_2 r^4 \right)^2 \\
&\qquad \texteq{5} \frac{1}{n^2}\, 2^{m+3} \operatorname{vol}\left(\mathbb{S}^{2m-1}\right) \int_0^\infty ds\ e^{-s} \left( 4 C_1^4 s^{m+5} + 4C_1^2 C_2 s^{m+4} + C_2^2 s^{m+3}\right) \\
&\qquad \texteq{6} \frac{2^{m+4}\pi^m m(m+1)(m+2)(m+3)}{n^2} \left( 4C_1^4 (m+4)(m+5) + 4C_1^2 C_2 (m+4) + C_2^2 \right) \\
&\qquad \textleq{7} \frac{C_3^2}{n^2}\, .
\end{aligned}
\label{eq:estim4}
\ee
The justification of the above steps goes as follows: in~4 we switched to spherical coordinates; in~5 we performed the change of variables $s\coloneqq \frac12 r^2$; in 6 we computed the gamma integrals, also remembering that $\operatorname{vol}\left( \mathbb{S}^{2m-1}\right) = \frac{2\pi^m}{(m-1)!}$; finally, the constant $C_3$ introduced in~7 depends -- for fixed $m$ -- only on $M'_4$ (note that $M'_3\leq M'_4$ by construction). The proof of the first claim is completed once one inserts \eqref{eq:estim4} into \eqref{eq:secondid}. In particular, if $D^3\chi_{\rho}(0) = 0$ we see that the convergence rate is $\mathcal O_{M'_4}\left(n^{-1}\right)$. This proves also the second claim.
\end{proof}

We continue with the proof of the low-regularity QCLT that assumes finiteness of phase space moments of order up to $2+\alpha$, for some $\alpha\in (0,1]$.

\begin{proof}[Proof of Theorem~\ref{thm:QBElow}]
We just deal with the case where $\alpha\in (0,1)$. As above, we start by fixing $\mu\in (0,2)$ and choosing a sufficiently small $\varepsilon>0$ so that for any $z\in B \coloneqq B\left(0,\sqrt{n}\, \eps\right)$ the inequality \eqref{eq:mu} holds. By a similar estimate as in \eqref{eq:estm2}, but now leveraging \eqref{Taylor alpha} instead of \eqref{Taylor 2}, we have that for any $z\in B\left(0,\sqrt{n}\, \varepsilon\right)$
\bb \label{eq:estm4}
\begin{aligned}
&\left| \log \left( \chi_{\rho}\left(\tfrac{z}{\sqrt{n}}\right)^n \right) + \frac12 \sumno_j \nu_j |z_j|^2 \right| \\
&\qquad \leq n \left| \log \left( \chi_{\rho}\left(\tfrac{z}{\sqrt{n}}\right) \right) + 1 - \chi_\rho \left( \tfrac{z}{\sqrt{n}}\right) \right| + n \left| \chi_\rho \left( \tfrac{z}{\sqrt{n}}\right) -1 + \frac12 \sumno_j \nu_j |z_j|^2 \right| \\
&\qquad \leq \frac{m^2(2m+1)^2}{4}\, M'_2(\rho,\varepsilon)^2\, a(\mu)\, \frac{|z|^4}{n} + \frac{m(2m+1)}{2}\, M'_{2+\alpha}(\rho,\varepsilon)\, \frac{|z|^{2+\alpha}}{n^{\alpha/2}} \\
&\qquad \leq \frac{C_4 |z|^{2+\alpha}}{n^{\alpha/2}}\, ,
\end{aligned} 
\ee
where the constant $C_4$ introduced in the last line depends only on $M'_{2+\alpha}$ (note that $M'_2\leq M'_{2+\alpha}$).
\begin{align*}
e^{\frac12 \sum_j \nu_j |z_j|^2} \left|\chi_\rho \left(\tfrac{z}{\sqrt{n}}\right)^n - e^{-\frac12 \sumno_j \nu_j |z_j|^2} \right| &= \left| \exp\left( \log\left( \chi_\rho\left( \tfrac{z}{\sqrt{n}} \right)^n\right) + \frac12 \sumno_j \nu_j |z_j|^2\right) - 1\right| \\
&\textleq{1} \frac{C_4 |z|^{2+\alpha}}{n^{\alpha/2}}\, e^{C_4 |z|^{2+\alpha} \big/ n^{\alpha/2}} \\
&\textleq{2} \frac{C_4 |z|^{2+\alpha}}{n^{\alpha/2}}\, e^{\frac14 |z|^{2}}\, .
\end{align*}
Here, in 1 we used the elementary estimate $\left|e^u - 1\right| \leq |u| e^{|u|}$, together with the observation that the function $\RR \ni x\mapsto x e^x$ is monotonically increasing. In 2 we used the fact that $|z|\leq \sqrt{n}\, \varepsilon$, and chose $\varepsilon>0$ sufficiently small so that $\varepsilon^{\alpha} C_4\leq \frac14$. Combining the above estimate with the fact that $\nu_1,\ldots, \nu_m\geq 1$ yields
\bb
\left|\chi_\rho \left(\tfrac{z}{\sqrt{n}}\right)^n - e^{-\frac12 \sumno_j \nu_j |z_j|^2} \right| \leq \frac{C_4 |z|^{2+\alpha}}{n^{\alpha/2}}\, e^{-\frac14 |z|^{2}}\, ,
\label{eq:estm5}
\ee
which upon integration in turn leads to
\bb
\begin{aligned}
\int_{|z|\leq \sqrt{n}\, \varepsilon} d^{2m} z \left|\chi_\rho \left(\tfrac{z}{\sqrt{n}}\right)^n - e^{-\frac12 \sumno_j \nu_j |z_j|^2} \right|^2 &\leq \frac{C_4^2}{n^\alpha} \int_{|z|\leq \sqrt{n}\, \varepsilon} d^{2m} z\, e^{-\frac12 |z|^{2}} |z|^{4+2\alpha} \\
&\leq \frac{C_4^2}{n^\alpha} \int d^{2m} z\, e^{-\frac12 |z|^{2}} |z|^{4+2\alpha} \\
&\texteq{3} \frac{C_4^2}{n^\alpha}\, \operatorname{vol}\left( \mathbb{S}^{2m-1}\right) \int_0^\infty dr\, e^{-\frac12 r^{2}}\, r^{2m+3+2\alpha} \\
&\texteq{4} \frac{2^{m+\alpha+2}\, C_4^2\, \pi^m\, \Gamma\left( m+\alpha+1\right)}{n^\alpha\, (m-1)!} \\
&\textleq{5} \frac{C_5^2}{n^\alpha}\, .
\end{aligned}
\label{eq:estm6}
\ee
Here, in~3 we switched to spherical coordinates; in~4 we operated the change of variables $s\coloneqq \frac12 r^2$ and computed the gamma integrals; the constant introduced in~5 depends, for fixed $\alpha$, only on $M'_{2+\alpha}$. Inserting \eqref{eq:estm6} into the right-hand side of \eqref{eq:secondid} completes the proof.
\end{proof}

\subsection{Convergence in trace distance and relative entropy} \label{sec-proofs:traceentr}

In this section, we further use the assumption of finiteness of the second moments of the state in order to find convergence rates in trace distance and in relative entropy.

\begin{proof}[Proof of Corollary~\ref{cor:tracerel}]
The hypothesis implies in particular that $\rho$ has finite phase space moments of the second order. By Theorem \ref{thm:converse}, this amounts to saying that $\rho$ has also finite standard moments of the second order, that is, that $\tr\left[\rho H_m \right]\leq E<\infty$. Iterating \eqref{energy convolution} and passing to the limit, we see that in fact
\bbb
\tr\left[\rho^{\boxplus n} H_m \right] = \tr\left[\rho_\G H_m \right] = \tr\left[\rho H_m \right] \leq E\, .
\eee
Now, for any $E'>0$, denote by $P_{E'}$ the projection onto the finite dimensional subspace generated by the eigenvectors of the canonical Hamiltonian $H_m$ of eigenvalue less than $E'$. Then, by Markov's inequality, for any $\eps>0$,
\begin{align*}
\tr\left[\rho^{\boxplus n} P_{E/\eps}\right],\ \tr\left[\rho_\G P_{E/\eps}\right] \ge 1-\eps\,.
\end{align*}
From the so-called `gentle measurement lemma' \cite[Lemma~9]{VV1999}, we have that
\begin{align*}
\left\|\rho^{\boxplus n}- P_{E/\eps}\,\rho^{\boxplus n} P_{E/\eps}\right\|_1,\ \left\|\rho_\G - P_{E/\eps}\, \rho_\G\, P_{E/\eps}
\right\|_1 \le 2\sqrt{\eps}\,.
\end{align*}
Then, 
\begin{align*}
\left\|\rho^{\boxplus n}-\rho_\G\right\|_1 &\le \left\|\rho^{\boxplus n}-P_{E/\eps}\, \rho^{\boxplus n} P_{E/\eps}\right\|_1+\left\|P_{E/\eps}\left(\rho^{\boxplus n}-\rho_\G\right) P_{E/\eps}\right\|_1 + \left\|P_{E/\eps}\,\rho_\G\, P_{E/\eps}-\rho_\G\right\|_1 \\
& \le 4\sqrt{\eps} + \left\|P_{E/\eps}\left(\rho^{\boxplus n}-\rho_\G\right)P_{E/\eps}\right\|_1 \\
&\le 4\sqrt{\eps} + \left\|P_{E/\eps}\right\|_2\, \left\|P_{E/\eps}\left(\rho^{\boxplus n}-\rho_\G\right)P_{E/\eps}\right\|_2 \\
&\le 4\sqrt{\eps} + (E/\eps)^{m/2}\left\|\rho^{\boxplus n}-\rho_\G\right\|_2
\end{align*}
The result follows after optimising over $\eps>0$. In particular, if $\left\|\rho^{\boxplus n}-\rho_\G\right\|_2=\mathcal{O}\left(n^{-\alpha}\right)$, we find that $\left\|\rho^{\boxplus n}-\rho_\G\right\|_1=\mathcal{O}\left(n^{-\frac{\alpha}{m+1}}\right)$.

We now turn to the proof of the convergence in relative entropy. Observe that, since $\rho^{\boxplus n}$ and $\rho_\G$ share the same first and second moments, $\tr\left[\rho^{\boxplus n}\log \rho_\G\right]=\tr\left[\rho_\G\log\rho_\G\right]$ and thus 
$D\left(\rho^{\boxplus n}\big\|\rho_\G\right) = S\left(\rho_\G\right) - S\left(\rho^{\boxplus n}\right)$. The result follows directly from \cite[Lemma~18]{tightuniform}.
\end{proof}

\section{Optimality of convergence rates and necessity of finite second moments in the QCLT: Proofs} \label{sec-proofs:optimality}

In this section we discuss the optimality of our results in two different directions: 

\begin{itemize}
    \item First, we provide examples of states $\rho$ that do not have finite second moments and for which $\rho^{\boxplus n}$ does not converge to any quantum state. This shows the necessity of the assumptions on finite second moments in the Cushen--Hudson Theorem (Section \ref{subsec-proofs:unbounded}).
    \item Secondly, we provide examples of explicit states which saturate our convergence rates in Theorems \ref{thm:QBE'} and \ref{thm:QBElow} (Section \ref{subsec-proofs:optimality-rates}).
\end{itemize}

\subsection{Failure of convergence for states with unbounded energy} \label{subsec-proofs:unbounded}

We now show that the assumption of finiteness of second moments in Theorems \ref{thm:CushHud} and \ref{CH thm} cannot be weakened, e.g.\ by replacing it with finiteness of some lower-order moments. Some examples of states with undefined moments that do not satisfy Theorems \ref{thm:CushHud} and \ref{CH thm} can be obtained by drawing inspiration from probability theory. For instance, remembering that a classical Cauchy-distributed random variable does not satisfy the central limit theorem, we construct the following example.

\begin{ex}[(Cauchy-based wave function)] \label{Cauchy wave function ex}
Consider the pure state $\ket{\psi_f}$ with wave function $f(x)\coloneqq \frac{1}{\sqrt{\pi}}\frac{1}{x+i}$. The characteristic function of this state can be computed thanks to \eqref{chi psi f}, which in this case evaluates to
\bb
\chi_{\ket{\psi_f}\bra{\psi_f}}(z) = \frac{\sqrt2\, e^{- |z_I| \left( \sqrt2 + i z_R\right)}}{\sqrt2 + i z_R}\, .
\ee
The absolute value of this characteristic function is illustrated in Figure \ref{fig:qcf}.
\begin{figure}
    \centering
    \includegraphics[width=0.5\textwidth]{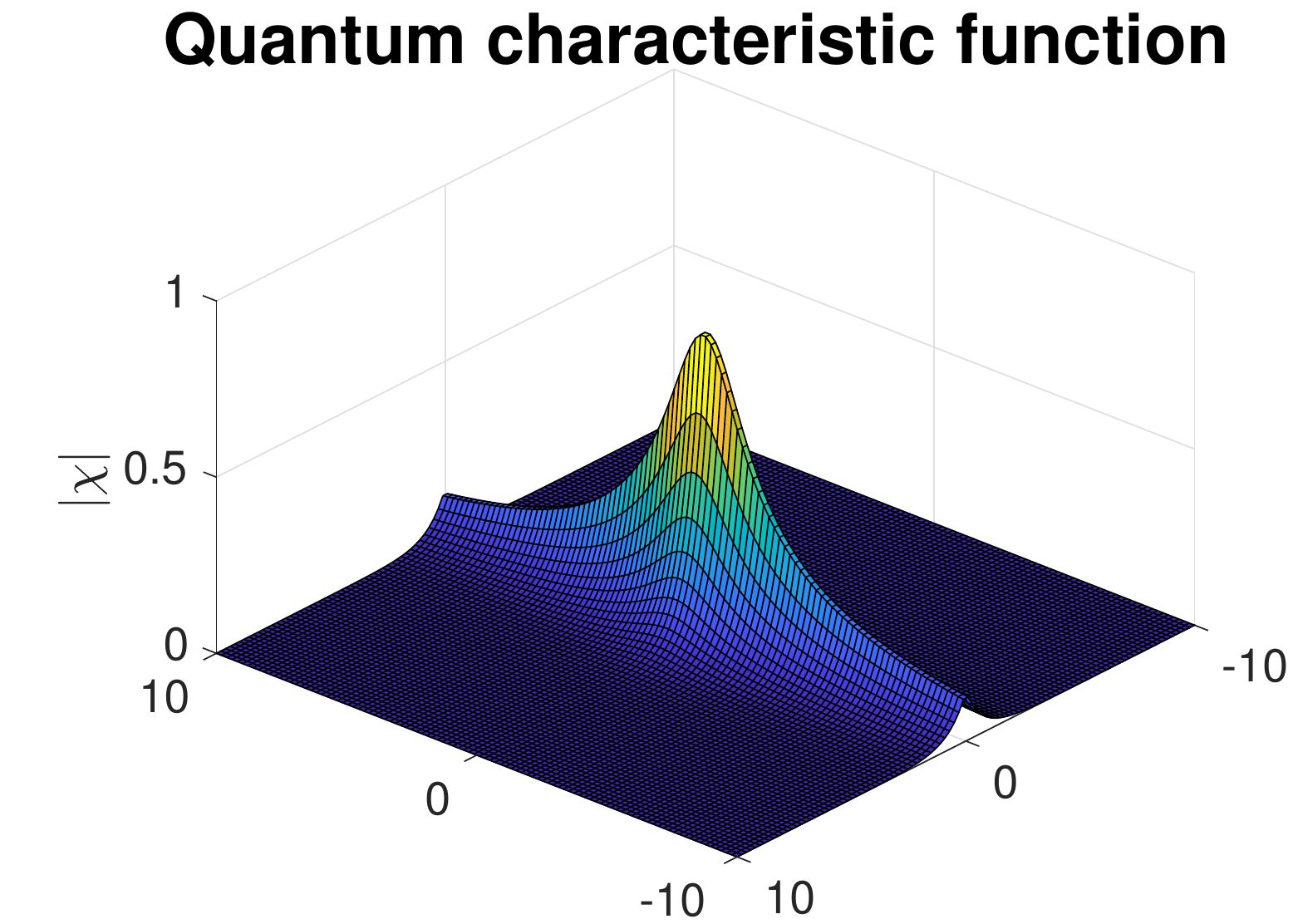}
    \caption{Example of the modulus of a quantum characteristic function, taken from Example \ref{Cauchy wave function ex}, with heavy tails in a single direction.}
    \label{fig:qcf}
\end{figure}

We then find the pointwise limit $\lim_{n \to \infty} \chi_{\ket{\psi_f}\bra{\psi_f}}\left( z/\sqrt{n}\right)^n = \delta_{z,0}$ which again is not continuous at $0$ and hence is not the characteristic function of any quantum state.
\end{ex}

The main drawback of the above state is that it does not have even first order moments. We can fix this by considering a slightly more sophisticated example. To proceed further, we first need to recall a well-known integral representation of fractional matrix powers.

\begin{lemma}[{\cite[Proposition~5.16]{schmuedgen}}] \label{integral representation lemma}
For all $r\in (0,1)$, all positive (possibly unbounded) operators $A$, and all $\ket{\psi}\in \dom\left(A^{1/2}\right)$, we have that
\bb
\left\|A^{r/2}\ket{\psi}\right\|^2 = \frac{\sin(\pi r)}{\pi} \int_0^\infty t^{r-1} \braket{\psi|\frac{A}{t I+A}|\psi} \, dt\, ,
\label{integral representation}
\ee
where all functions of $A$ are defined by means of its spectral decomposition.
\end{lemma}

\begin{proof}[Proof of Proposition \ref{crazy state prop}]
The state is clearly centred, for instance because the wave function is symmetric under inversion $x\mapsto -x$. We proceed to prove claim (b). Note that, since $x^2+p^2=I+ 2a^\dagger a\ge I$, $2 a a^\dag = x^2+p^2 + I \leq 2(x^2+p^2)$, where $p\coloneqq -i \frac{d}{dx}$ is the momentum operator. We now apply the operator inequality $(A+B)^{r} \leq A^{r} + B^{r}$, which can be shown to hold for all $r \in [0,1]$ and all positive (possibly unbounded) self-adjoint operators $A,B$. To prove this explicitly in the non-trivial case where $r\in (0,1)$, we apply \eqref{integral representation} to $A+B$. For a generic $\ket{\psi}\in \dom\left(A^{r/2}\right)\cap \dom\left( B^{r/2}\right)$, we obtain that
\begin{align*}
\left\|(A+B)^{r/2}\ket{\psi}\right\|^2 &= \frac{\sin(\pi r)}{\pi} \int_0^\infty t^{r-1} \braket{\psi|\frac{A+B}{t I+A+B}|\psi} \, dt \\
&= \frac{\sin(\pi r)}{\pi} \int_0^\infty t^{r-1} \braket{\psi|\left(\frac{A}{t I+A+B} + \frac{B}{t I+A+B} \right)|\psi} \, dt \\
&\leq \frac{\sin(\pi r)}{\pi} \int_0^\infty t^{r-1} \left(\braket{\psi|\frac{A}{t I+A}|\psi} + \braket{\psi|\frac{B}{t I+B}|\psi} \right) \, dt \\
&= \left\|A^{r/2}\ket{\psi}\right\|^2 + \left\|B^{r/2}\ket{\psi}\right\|^2 ,
\end{align*}
where the inequality in the above derivation follows e.g.\ from \cite[Corollary~10.13]{schmuedgen}. Now, setting $A=x^2$, $B=p^2$ and $r=1-\delta$, we obtain that
\bbb
(a a^\dag)^{1-\delta} \leq |x|^{2(1-\delta)} + |p|^{2(1-\delta)} \leq |x|^{2(1-\delta)} + 1+ p^{2}\, .
\eee
Computing the expectation value on $\ket{\psi_f}$ yields
\begin{align*}
\braket{\psi_f | (a a^\dag)^{1-\delta} | \psi_f} &\leq \braket{\psi_f | \left(|x|^{2(1-\delta)} + 1+ p^2\right) | \psi_f} = \frac{1}{\sqrt\pi}\, \Gamma \left(\frac32-\delta\right) \Gamma (\delta) + 1 + \frac{7}{10}\, ,
\end{align*}
where the last step is by explicit computation. This proves~(b). We now move on to~(c). For this we evaluate the characteristic function of the convolution $\ketbra{\psi_f}^{\boxplus n}$ on the purely imaginary line. For $t\in \RR$, using \eqref{chi psi f} we obtain that
\bbb
\chi_{\ket{\psi_f}\bra{\psi_f}}(i t) = \int_{-\infty}^{+\infty} dx\, |f(x)|^2 e^{\sqrt2\, itx} = \sqrt2\,|t|\, K_1\left(\sqrt2\, |t|\right) ,
\eee
were $K_1$ is a modified Bessel function of the second kind, and the last equality follows from \eqref{crazy wave function} and \cite[Eq.~(9.6.25)]{ABRAMOWITZ}. Therefore, for any fixed $t > 0$ it holds that
\bbb
\begin{aligned}
\lim_{n\to\infty}\chi_{\ket{\psi_f}\bra{\psi_f}^{\boxplus\, n}}(i t) &= \lim_{n\to\infty} \chi_{\ket{\psi_f}\bra{\psi_f}}\left(\frac{i t}{\sqrt{n}}\right)^n \\
&= \lim_{n\to\infty} \left( 1 + \left( c + \log t - \frac12 \log n\right) \frac{t^2}{n} + O\left(n^{-3/2}\right) \right)^n = 0\, ,
\end{aligned}
\eee
where we have used the expansion in \cite[Eq.~(9.6.53)]{ABRAMOWITZ} (see also \cite[Eq.~(6.3.2) and~(9.6.7)]{ABRAMOWITZ}). Since $\chi_{\ket{\psi_f}\bra{\psi_f}^{\boxplus\, n}}(0)=1$ for all $n$ because $\ket{\psi_f}\!\bra{\psi_f}^{\boxplus\, n}$ is a valid quantum state, the sequence of functions $\chi_{\ket{\psi_f}\bra{\psi_f}^{\boxplus\, n}}$ does not possess a continuous limit. Hence, it cannot converge to the characteristic function of any quantum state. This proves~(c).
\end{proof}

\subsection{Optimality of the convergence rates}  \label{subsec-proofs:optimality-rates}

The following two examples show that the bounds stated in Theorems \ref{thm:QBE'} and \ref{thm:QBElow} are indeed saturated. Both examples consist of states constructed using the Fock basis. The construction of examples saturating the bounds in Theorems~\ref{thm:QBE'} and \ref{thm:QBElow} is motivated by the following Proposition. 

\begin{prop} \label{prop:Fockdiag}
Let $\rho$ be a one-mode density operator satisfying the assumptions of Theorem \ref{thm:QBE'} and also $\braket{i| \rho|j} =0$ for $|i-j| \in \left\{1,3 \right\}$. Then the state $\rho^{\boxplus n}$ converges at least with rate $\mathcal{O}\left(n^{-1}\right)$ to its Gaussification
\bb
\left\|\rho^{\boxplus n}-\rho_\G\right\|_2 = \mathcal{O}\left(n^{-1}\right) .
\ee
In particular, every density operator satisfying the assumptions of Theorem \ref{thm:QBE'} that is diagonal in the Fock basis achieves a $\mathcal O(n^{-1})$ rate.
\end{prop}

\begin{proof}[Proof of Proposition~\ref{prop:Fockdiag}]

By Theorem \ref{thm:QBE'} it suffices to show that $D^3\chi_{\rho}(0)=0$ under the assumptions of the Proposition. We start by recalling that any density operator $\rho$ has an expansion into the Fock basis such that 
\bb \label{eq:densityop}
\rho = \sum_{i,j=0}^\infty \braket{i|\rho|j} \ketbraa{i}{j}\, .
\ee
Hence, we find for the characteristic function that
\bb
\label{eq:charfunexp}
\chi_{\rho}(z) = \sum_{i,j =0}^\infty \braket{i|\rho|j} \chi_{\ket{i}\bra{j}}(z).
\ee
Using a finite-rank approximation of the density operator $\rho$, it suffices then by Theorem \ref{thm:all-of-us} to analyse the component-wise derivatives in \eqref{eq:charfunexp}. The functions $\chi_{\ket{i}\bra{j}}$ are explicitly given by \cite[Eq.~(4.4.46) and~(4.4.47)]{BARNETT-RADMORE} 
\bb
\label{eq:hermite}
\chi_{\ket{i}\bra{j}}(z) = \left\{ \begin{array}{ll} \sqrt{\frac{i!}{j!}}\, (-z)^{j-i} e^{-\frac{|z|^2}{2}} L_{j}^{j-i}\left( |z|^2\right) & \text{if $i\leq j$,} \\[1ex] \sqrt{\frac{j!}{i!}}\, (z^*)^{i-j} e^{-\frac{|z|^2}{2}} L_{i}^{i-j}\left( |z|^2\right) & \text{if $i>j$.} \end{array} \right.
\ee
Here, $L_n^k(x)\coloneqq \frac{d^k}{dx^k}\, L_n(x) = (-1)^k \sum_{\ell=0}^{n-k}\binom{n}{\ell+k} \frac{(-x)^\ell}{\ell!}$ are the associated Laguerre polynomials.
By assumption, it suffices to consider the case where $|i-j|$ is even or $\vert i-j \vert$ is odd and at least~$5$. We find that by writing the characteristic function in the form $\chi_{\ket{i}\bra{j}} (z)\coloneqq e^{-\frac{\vert z \vert^2}{2}}H_{ji}(z)$ for some suitable function $H_{ji}$, as in \eqref{eq:hermite}, that for the different possible third derivatives, we have
\bb
\begin{aligned}
 \partial_z^3 \chi_{\ket{i}\bra{j}}(0) &= -3 \partial_z H_{ji}(0)+ \partial_z^3 H_{ji}(0),\\
 \partial_{z^*}^3  \chi_{\ket{i}\bra{j}}(0) &= -3\partial_{z^*} H_{ji}(0)+ \partial_{z^*}^3 H_{ji}(0),\\
 \partial_{z}^2\partial_{z^*}  \chi_{\ket{i}\bra{j}}(0) &= -\partial_{z^*}H_{ji}(0)+ \partial_{z}^2\partial_{z^*} H_{ji}(0), \\
\partial_{z^*}^2\partial_{z}  \chi_{\ket{i}\bra{j}}(0) &= - \partial_{z} H_{ji}(0)+ \partial_{z^*}^2 \partial_{z} H_{ji}(0).
 \end{aligned}
\ee
Therefore, the only possible non-zero contribution to the third derivative of the quantum characteristic function $\chi_{\rho}$ at zero could be due to terms that contain either one or three derivatives of functions $H_{ji}$ evaluated at zero. 

If $|i-j| \ge 4$ then $z$ and $z^*$ appear in \eqref{eq:hermite} with a joint power of at least $4$; thus, this term's contribution necessarily has to vanish. It suffices therefore to consider the case where $\vert i-j \vert=2$. If $H_{ji}$ is only differentiated once, then it is clear that this derivative has to vanish at zero, since $z,z^*$ appear with a joint power of at least two.

If $H_{ji}$ is differentiated three times, then the term $|z|^{2}$ causes the derivative to vanish at zero unless this term is differentiated precisely two times. This, however, implies that the Laguerre polynomial is differentiated exactly once. However, by the chain rule any first order derivative of the term $L_j^{\vert j-i \vert}(\vert z\vert^2)$ vanishes at the origin. 
This concludes the proof.
\end{proof}

The following example shows that the $\mathcal O(n^{-1})$ convergence rate stated in Proposition \ref{prop:Fockdiag}, under the assumption that $D^3\chi_\rho(0)=0$, is in fact attained.

\begin{ex}[($\mathcal O(n^{-1})$-rate)]
By Proposition \ref{prop:Fockdiag} we can take $\rho=\ketbra{1}$ to obtain a convergence rate of at least $\mathcal O(n^{-1})$ in the QCLT. That the $\mathcal O(n^{-1})$ rate is actually attained is illustrated in the right figure in Figure~\ref{fig:rates}. The $\mathcal O(n^{-1})$ rate is saturated both in Hilbert--Schmidt and trace norm. 
\label{ex:Ludovico}
\end{ex}
\medskip

The following example shows that the $\mathcal O(n^{-1/2})$ of Theorem \ref{thm:QBElow} is attained.

\begin{ex}[($\mathcal O(n^{-1/2})$-rate)]
\label{ex:Simon}
Consider the state\footnote{We use states $\ket{0}$ and $\ket{3}$ rather than $\ket{0}$ and $\ket{1}$ because the latter choice does not lead to a centred state.}
\bb
\rho = \frac{\ket{0}+\ket{3}}{\sqrt2}\frac{\bra{0}+\bra{3}}{\sqrt2}\, .
\ee
Its characteristic function is explicitly given by \eqref{eq:hermite}
\bb
\chi_{\rho}(z)= \frac{1}{12}\, e^{-\frac{|z|^2}{2}} \left( 12 - 18|z|^2 + \sqrt6 \left( z^3 - (z^*)^3 \right) + 9|z|^4 - |z|^6 \right)
\ee
Now, since $\braket{0|\rho|3}\neq 0$ we see that the condition of Proposition \ref{prop:Fockdiag} does not hold. One verifies directly that $\chi_\rho(z) = 1 - 2 |z|^2 + o\left(|z|^2\right)$, so that $\rho$ is already in Williamson form (cf.~\eqref{Williamson chi}). 
Letting $\Phi(z) = e^{-2|z|^2}$, we then find that $\left\| {\chi}_{\rho} - \Phi \right\|_{\LL^2(\mathbb R^2)}$ converges with rate $n^{-1/2}$, see Figure~\ref{fig:rates}.
\end{ex}

\begin{figure}
    \begin{center}
        \includegraphics[width=0.8\textwidth]{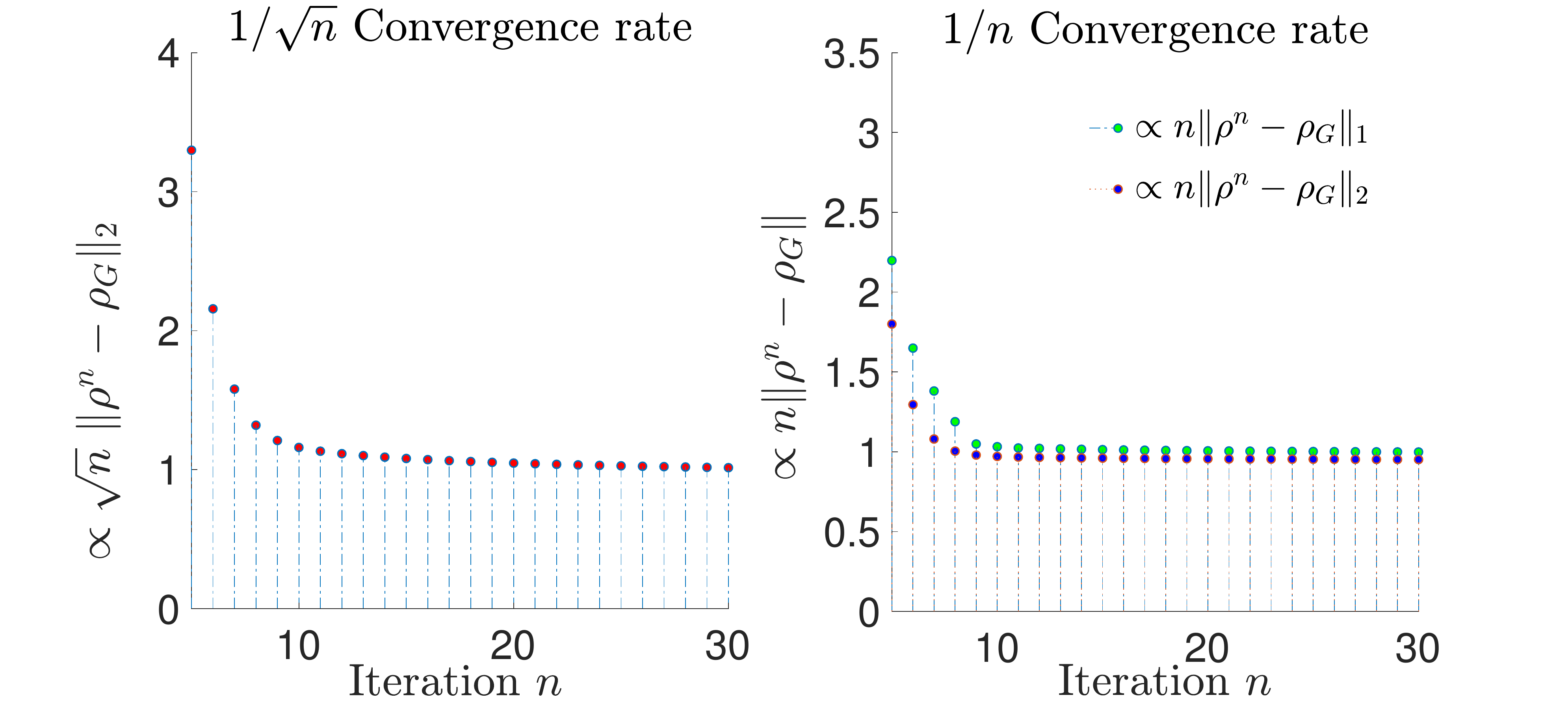}
        \end{center}
        \caption{This plot shows the expressions $cn^{\alpha} \Vert \rho^{\boxplus n}-\rho \Vert$ for a constant $c>0$ such that $\lim_{n \rightarrow \infty}cn^{\alpha} \Vert \rho^{\boxplus n}-\rho \Vert=1.$ The left figure shows that the $\mathcal O(1/\sqrt{n})$ convergence rate is sharp (Theorem~\ref{thm:QBElow}) by using the state from Ex.~\ref{ex:Simon}. 
        The right figure shows that we can obtain a rate $\mathcal O(1/n)$ if $D^3 \chi_{\rho}(0)=0$ (Theorem~\ref{thm:QBE'}) by using the state from Example~\ref{ex:Ludovico}. In both figures we write $\rho^n$ for $\rho^{\boxplus n}.$}
        \label{fig:rates}
\end{figure}

\section{Cascade of beam splitters: Proofs} \label{sec-proofs:cascade}

In this section, we prove the results claimed in Section \ref{subsec:applications}, namely convergence rates for cascades of beam splitters converging to thermal attenuator channels.

\subsection{Generalities of the cascade channels}

In order to study the convergence of the cascade channel, we start by proving the following elementary equivalence.

\begin{lemma} \label{lem:cascade-equivalence}
For an $m$-mode quantum state $\rho$, some $\lambda\in [0,1]$, and a positive integer $n$, consider the cascade channel $\oneN^n$ (cf.~\eqref{channel N}). One has that
\bb
\oneN^n = \mathcal{N}_{\rho(\lambda, n),\, \lambda}: \chi_\omega(z)\ \longmapsto\ \chi_{\oneN^n(\omega)}(z) = \chi_{\omega}\left(\sqrt{\lambda}\,z\right) \chi_{\rho(\lambda, n)}\!\left( \sqrt{1-\lambda}\, z\right) ,
\label{eq:cascade}
\ee
where the \emp{effective environment state} $\rho(\lambda,n)$ is defined via its characteristic function
\bb
\chi_{\rho(\lambda, n)}(z) \coloneqq \prod_{\ell=1}^n\,\chi_{\rho}\left(\sqrt{\frac{1-\lambda^{1/n}}{1-\lambda}}\,\lambda^{\frac{\ell-1}{2n}} z\right) .
\label{eq:charcascade}
\ee
\end{lemma}

\begin{proof}
We proceed by induction. The case $n=1$ follows from \eqref{boxplus characteristic functions}. Let us assume that the claim holds for $n-1$, so that
\bbb
\mathcal{N}_{\rho,\, \mu^{\text{\scalebox{0.8}{$1/(n-1)$}}}}^{n-1} = \mathcal{N}_{\rho(\mu,\,n-1),\, \mu}\, ,\qquad \chi_{\mathcal{N}_{\rho,\,\mu^{\text{\scalebox{0.8}{$1/(n-1)$}}}}^{n-1}(\omega)}(z) = \chi_\omega\left( \sqrt{\mu}\, z\right) \prod_{\ell=1}^{n-1}\,\chi_{\rho}\left(\sqrt{1-\mu^{1/(n-1)}}\,\mu^{\frac{\ell-1}{2(n-1)}} z\right) .
\eee
By setting $\mu=\lambda^{(n-1)/n}$ we see that
\bbb
\oneN^{n-1} = \mathcal{N}_{\rho\left(\lambda^{(n-1)/n},\,n-1\right),\, \lambda^{(n-1)/n}}\, ,\qquad \chi_{\oneN^{n-1}(\omega)}(z) = \chi_\omega\left( \lambda^{\frac{n-1}{2n}}\, z\right) \prod_{\ell=1}^{n-1}\,\chi_{\rho}\left(\sqrt{1-\lambda^{1/n}}\,\lambda^{\frac{\ell-1}{2n}} z\right) .
\eee
Since $\oneN^{n}(\omega) = \oneN\left( \oneN^{n-1}(\omega)\right) = \oneN^{n-1}(\omega) \boxplus_{\lambda^{1/n}} \rho$, composition with the $n^{\text{th}}$ copy of the channel $\oneN$ yields 
\begin{align*}
\chi_{\oneN^n(\omega)}(z) &= \chi_{\oneN^{n-1}(\omega)}\left( \lambda^{1/2n} z \right) \chi_\rho\left( \sqrt{1-\lambda^{1/n}}\, z\right) \\
&= \chi_\omega\left( \sqrt{\lambda}\, z\right)  \left(\prod_{\ell=1}^{n-1}\,\chi_{\rho}\left(\sqrt{1-\lambda^{1/n}}\,\lambda^{\frac{\ell}{2n}} z\right)\right) \chi_\rho\left( \sqrt{1-\lambda^{1/n}}\, z\right) \\
&= \chi_\omega\left( \sqrt{\lambda}\, z\right)\prod_{\ell=0}^{n-1}\,\chi_{\rho}\left(\sqrt{1-\lambda^{1/n}}\,\lambda^{\frac{\ell}{2n}} z\right) \\
&= \chi_\omega\left( \sqrt{\lambda}\, z\right)\prod_{\ell=1}^{n}\,\chi_{\rho}\left(\sqrt{1-\lambda^{1/n}}\,\lambda^{\frac{\ell-1}{2n}} z\right) ,
\end{align*}
which proves \eqref{eq:cascade} and \eqref{eq:charcascade}. Finally, one can also verify by induction that
\bb
\rho(\lambda, n) = \rho\left( \lambda^{\frac{n-1}{n}}, n-1\right) \boxplus_{\eta(\lambda, n)} \rho\, ,
\label{rho lambda n induction}
\ee
where $\eta(\lambda, n) \coloneqq \frac{\lambda^{1/n}\left( 1- \lambda^{(n-1)/n}\right)}{1-\lambda}\in [0,1]$, so that $\rho(\lambda, n)$ is a legitimate quantum state for all $\lambda\in [0,1]$ and all $n$.
\end{proof}

\subsection{On the effective environment state}

Thanks to Lemma~\ref{lem:cascade-equivalence}, the study of the cascade channel $\oneN^n$ boils down to that of the iteratively convolved state $\rho(\lambda, n)$ of \eqref{eq:charcascade}. Since such a convolution is \textit{not} symmetric (cf.~\eqref{chi CH state}), to proceed further we need to extend our quantum Berry--Esseen results to a non-i.i.d.~scenario. 
Note that the classical central limit theorem has indeed been extended to sequences of independent, non-identically distributed random variables \cite{lindeberg1922neue, Feller}, and even to sequences of correlated random variables \cite{bryc1993remark}. Rates of convergence for the former case can be found for instance in \cite{bhattacharya1986normal} (see e.g.\ Theorem~13.3 of \cite{bhattacharya1986normal}).

\begin{prop}\label{propnon-i.i.d.}
Let $\rho$ be a centred $m$-mode quantum state with finite second-order phase space moments. Then the sequence of quantum states $\rho(\lambda,n)$ defined via \eqref{eq:charcascade} converges to the Gaussification $\rho_\G$ of $\rho$ in trace norm. Moreover, if $\rho$ has finite third-order phase space moments then
\begin{align}
    \left\|\rho(\lambda,n) - \rho_{\G} \right\|_2 &= \mathcal{O}_{\lambda, M_3'}\left(n^{-1/2}\right) , \label{fibre states 2-norm} \\
    \left\|\rho(\lambda,n) - \rho_{\G} \right\|_1 &= \mathcal{O}_{\lambda, M_3'}\left(n^{-1/(2m+2)}\right) . \label{fibre states 1-norm}
\end{align}
Here, $M_3'=M'_3(\rho,\varepsilon)$ is defined by \eqref{eq:ps-moments}, and $\varepsilon>0$ is sufficiently small.
\end{prop}


\begin{proof}
The argument is a variation of that used to prove Theorem \ref{thm:QBElow} in Section \ref{subsec-proofs:convrates}. First of all, reasoning as in Section \ref{subsubsec:Williamson}, we can assume without loss of generality that $\rho$ is in its Williamson form. To simplify the notation, we introduce the re-scaled vectors $w_\ell\coloneqq \sqrt{\frac{1-\lambda^{1/n}}{1-\lambda}}\, \lambda^{\frac{\ell-1}{2n}} z\in \CC^m$, where $\ell\in \{1,\ldots, n\}$. Then clearly $\chi_{\rho(\lambda,n)}(z) = \prod_{\ell=1}^n \chi_\rho\left( w_\ell \right)$. Note that $\left|w_\ell\right|\leq \sqrt{\frac{\log \left(1/\lambda\right)}{n(1-\lambda)}}\,|z|$; substituting $z\mapsto w_\ell$ into \eqref{Taylor 1} and \eqref{Taylor 2}, we see that whenever $|z|\leq \sqrt{\frac{n(1-\lambda)}{\log\left(1/\lambda\right)}}\, \varepsilon$ it holds that
\begin{align}
\left| \chi_\rho\left( w_\ell \right) - 1 \right| &\leq \frac{m(2m+1)}{2}\, M'_2(\rho,\varepsilon)\, \frac{\log \tfrac{1}{\lambda}}{1-\lambda}\, \frac{|z|^2}{n}\, , \label{adapted Taylor 1} \\
\left| \chi_\rho\left( w_\ell \right) - 1 + \frac12 \frac{1-\lambda^{1/n}}{1-\lambda}\, \lambda^{\frac{\ell-1}{n}} \sumno_j \nu_j |z_j|^2 \right| &\leq \frac{m(m+1)(2m+1)}{9}\, M'_3(\rho,\varepsilon) \left(\frac{\log \tfrac{1}{\lambda}}{1-\lambda}\right)^{3/2} \frac{|z|^3}{n^{3/2}}\, . \label{adapted Taylor 2}
\end{align}
We start by choosing $\varepsilon>0$ small enough so that \eqref{eq:mu} holds for some $\mu\in (0,2)$. We can now mimic the calculations in \eqref{eq:estm2}, obtaining
\bb
\begin{aligned}
&\left| \log \chi_{\rho(\lambda,n)}(z) + \frac12 \sumno_j \nu_j |z_j|^2 \right| \\
&\quad = \left| \log \prod_{\ell=1}^n \chi_\rho\left( w_\ell\right) + \frac12 \sumno_j \nu_j |z_j|^2 \right| \\
&\quad \textleq{1} \sum_{\ell=1}^n \left| \log \chi_\rho\left( w_\ell\right) + \frac12 \frac{1-\lambda^{1/n}}{1-\lambda}\, \lambda^{\frac{\ell-1}{n}} \sumno_j \nu_j |z_j|^2\right| \\
&\quad \leq \sum_{\ell=1}^n \Bigg\{ \left| \log \chi_\rho\left( w_\ell\right) + 1 - \chi_\rho\left( w_\ell\right) \right| + \left| \chi_\rho\left( w_\ell\right) - 1 + \frac12 \frac{1-\lambda^{1/n}}{1-\lambda}\,\lambda^{\frac{\ell-1}{n}} \sumno_j \nu_j |z_j|^2\right| \Bigg\} \\
&\quad \textleq{2} \sum_{\ell=1}^n \Bigg\{ \frac{m^2(2m\!+\!1)^2}{4}\, \frac{\log \tfrac{1}{\lambda}}{1-\lambda}\, M'_2(\rho,\varepsilon)^2\, a(\mu)\, \frac{|z|^4}{n^2} + \frac{m(m\!+\!1)(2m\!+\!1)}{9}\, M'_3(\rho,\varepsilon) \left(\frac{\log \tfrac{1}{\lambda}}{1-\lambda}\right)^{3/2} \!\frac{|z|^3}{n^{3/2}} \Bigg\} \\
&\quad \textleq{3} \frac{C_6 |z|^3}{\sqrt{n}}\, ,
\end{aligned}
\label{eq:lambda-estimate}
\ee
Here, in~1 we observed that $\sum_{\ell=1}^n \frac{1-\lambda^{1/n}}{1-\lambda}\,\lambda^{\frac{\ell-1}{n}} = 1$ and applied the triangle inequality. To deduce~2, instead, we proceeded as for~\eqref{eq:estim0}. Namely, on the first addend we used the identity $\log \left( 1-\frac{x}{2}\right) + \frac{x}{2} = -\frac{x^2}{4}\, a(x)$ satisfied by the function $a(x)$ defined by \eqref{a}, we set $x = 2\left( 1-\chi_\rho\left( w_\ell\right)\right)$, we noted that $|x|\leq \mu$ implies that $|a(x)|\leq a(\mu)$, and lastly we employed \eqref{adapted Taylor 1}. The second addend, instead, has been estimated thanks to \eqref{adapted Taylor 2}. Finally, for fixed $m$ the constant introduced in~3 depends only on $M'_3$ and $\lambda$ (again, $M'_2\leq M'_3$ by construction).

Proceeding as usual, we continue to estimate
\begin{align*}
e^{\frac12 \sum_j \nu_j |z_j|^2} \left|\chi_{\rho(\lambda,n)}(z) - e^{-\frac12 \sum_j \nu_j |z_j|^2} \right| &= e^{\frac12 \sum_j \nu_j |z_j|^2} \left| \prod_{\ell=1}^n \chi_\rho \left( w_\ell \right) - e^{-\frac12 \sum_j \nu_j |z_j|^2} \right| \\
&= \left| \exp\left(\log \left(\prod_{\ell=1}^n \chi_\rho \left( w_\ell \right) + \frac12 \sum_j \nu_j |z_j|^2\right)\right) - 1 \right| \\
&\textleq{4} \frac{C_6 |z|^3}{\sqrt{n}}\, e^{C_6 |z|^3/\sqrt{n}} \\
&\textleq{5} \frac{C_6 |z|^3}{\sqrt{n}}\, e^{\frac14 |z|^2}\, .
\end{align*}
Note that in~4 we applied the elementary inequality $\left|e^u-1\right|\leq |u|e^{|u|}$, observed that $\RR\ni x\mapsto x e^x$ is a monotonically increasing function, and leveraged the bound in \eqref{eq:lambda-estimate}. In~5, instead, we wrote $\frac{C_6 |z|^3}{\sqrt{n}}\leq C_6\, \sqrt{\frac{1-\lambda}{\log\tfrac{1}{\lambda}}}\, \varepsilon \, |z|^2\leq \frac14 |z|^2$, where the last estimate holds provided that $\varepsilon>0$ is small enough.

Remembering that $\nu_1,\ldots, \nu_m\geq 1$, we can massage the above relation so as to get
\bb
\left| \chi_{\rho(\lambda,n)}(z) - e^{-\frac12 \sum_j \nu_j |z_j|^2} \right| \leq \frac{C_6 |z|^3}{\sqrt{n}}\, e^{-\frac14 |z|^2}\, .
\label{eq:lambda-estimate2}
\ee
Now, we can repeat the steps that led to \eqref{eq:firstid}. We obtain that
\bb
\begin{aligned}
&\pi^{m}\left\|\rho(\lambda,n) - \rho_\G\right\|_2^2 \\
&\qquad = \int d^{2m}z \left| \chi_{\rho(\lambda,n)}(z) - e^{-\frac12 \sum_j \nu_j |z_j|^2} \right|^2 \\
&\qquad \textleq{6} \int_{|z|\leq \sqrt{n}\, \varepsilon} d^{2m}z \left| \chi_{\rho(\lambda,n)}(z) - e^{-\frac12 \sum_j \nu_j |z_j|^2} \right|^2 + \int_{|z|>\sqrt{n}\, \varepsilon} d^{2m}z \left| \chi_{\rho(\lambda,n)}(z) - e^{-\frac12 \sum_j \nu_j |z_j|^2} \right|^2 \\
&\qquad \textleq{7} \int_{|z|\leq \sqrt{n}\, \varepsilon} d^{2m}z\, \frac{C_6 |z|^6}{n}\, e^{-\frac12 |z|^2} + 2 \int_{|z|>\sqrt{n}\, \varepsilon} d^{2m}z \left| \chi_{\rho(\lambda,n)}(z)\right|^2 + 2 \int_{|z|>\sqrt{n}\, \varepsilon} d^{2m}z e^{- \sum_j \nu_j |z_j|^2} \\
&\qquad \textleq{8} \frac{C_7}{n} + 2 \left( \sup_{|z|>\sqrt{n}\, \varepsilon} \prod_{\ell=2}^n \left|\chi_\rho\left(w_\ell\right)\right|^2 \right) \int_{|z|>\sqrt{n}\, \varepsilon} d^{2m}z \left| \chi_\rho(w_1)\right|^2 + 2 e^{-\frac{\varepsilon^2}{2}\, n} \int_{|z|\geq \sqrt{n}\,\varepsilon} e^{-\frac12 |z|^2} \\
&\qquad \textleq{9} \frac{C_7}{n} + 2 \left( \sup_{|v|>\frac{1}{\sqrt{2}}\sqrt{\lambda \log (1/\lambda)}\, \varepsilon} \left|\chi_\rho(v)\right| \right)^{2(n-1)} \left( \frac{2n(1-\lambda)}{\log \tfrac{1}{\lambda}} \right)^m + 2 (2\pi)^m e^{-\frac{\varepsilon^2}{2}\, n}
\end{aligned}
\label{eq:lambda-final}
\ee
The justification of the above steps is as follows. The estimate in~6 is just an application of the triangle inequality. In~7 we used \eqref{eq:lambda-estimate2} and the elementary fact that $|u+v|^2\leq 2|u|^2+2|v|^2$ on the second addend. As for~8, we: (i)~performed the integral and introduced a constant $C_7$ that depends on $m$ only on the first addend; (ii)~decomposed $\chi_{\rho(\lambda, n)}(z) = \chi_\rho(w_1) \cdot\prod_{\ell=2}^n \chi_\rho(w_\ell)$ on the second; and~(iii) used the fact that $e^{-\sum_j \nu_j |z_j|^2}\leq e^{-|z|^2}< e^{-\frac{\varepsilon^2}{2}\, n} e^{-\frac12 |z|^2}$ in the prescribed range on the third. Finally, in~9 we noted that if $|z|>\sqrt{n}\,\varepsilon$ then eventually in $n$
\bbb
|w_\ell| = \sqrt{\frac{1-\lambda^{1/n}}{1-\lambda}}\lambda^{\frac{\ell-1}{2n}} |z| > \frac{1}{\sqrt{2}}\sqrt{\lambda \log \tfrac{1}{\lambda}}\,\varepsilon
\eee
for all $\ell\in \{1,\ldots, n\}$; moreover, we used the fact that $\int d^{2m} u\, |\chi_\rho(u)|^2\leq 1$ to evaluate the integral in the second addend.

Since the second term in the rightmost side of \eqref{eq:lambda-final} decays faster than any inverse power of $n$ as $n\to\infty$ thanks to Proposition \ref{prop1}, the proof of \eqref{fibre states 2-norm} is complete.
Lastly, \eqref{fibre states 1-norm} follows similarly to Corollary \ref{cor:tracerel}.
\end{proof}

\subsection{Approximating cascade channels}

With this convergence at hand, we provide a quantitative bound on the approximation of thermal attenuator channels by cascades of beam splitters (with possibly non-Gaussian environment states). Recall that, to an environment state $\rho$ one can associate an attenuator channel $\mathcal{N}_{\rho,\lambda}(\omega)\coloneqq\omega\boxplus_\lambda\rho$ of transmissivity $0\le \lambda\le 1$. The following simple lemma is crucial to convert the above state approximation result (Proposition \ref{propnon-i.i.d.}) into a statement about approximations of attenuator channels.

\begin{lemma}\label{lem:channeltostate}
Given any two $m$-mode quantum states $\rho_1$ and $\rho_2$, and some $\lambda\in[0,1]$, the corresponding channels defined as in \eqref{channel N} satisfy
\bb
 \left\|\mathcal{N}_{\rho_1,\,\lambda} - \mathcal{N}_{\rho_2,\,\lambda}\right\|_{\diamond} \le \|\rho_1-\rho_2\|_1\, .
 \label{diamond as trace}
\ee
\end{lemma}

\begin{proof}
Let $R$ be any reference system, and let $\omega\in\cD(\cH_{AR})$ be a state on the bipartite system $AR$. Then
\begin{align*}
    \|(\mathcal{N}_{\rho_1,\,\lambda}-\mathcal{N}_{\rho_2,\,\lambda})\otimes\operatorname{id}_{R}(\omega)\|_1&=\| \tr_R[U_{\lambda}(\omega\otimes \rho_1)U_{\lambda}^\dagger] -\tr_R[U_{\lambda}(\omega\otimes\rho_2)U_{\lambda}^\dagger] \|_1\\
    &\le\| U_{\lambda}(\omega\otimes \rho_1-\omega\otimes \rho_2)U_{\lambda}^\dagger \|_1\\
    &=\|\rho_1-\rho_2\|_1\,,
    \end{align*}
where the inequality stems from the monotonicity of trace distance under quantum channels.
\end{proof}

With this lemma at hand, we are ready to prove Theorem \ref{thm:approxchann}.


\begin{proof}[Proof of Theorem \ref{thm:approxchann}]
Recall from Lemma \ref{lem:cascade-equivalence} that $\oneN^n = \mathcal{N}_{\rho(\lambda,n),\,\lambda}$, where $\rho(\lambda,n)$ is the state with characteristic function given by \eqref{eq:charcascade}. Applying Lemma \ref{lem:channeltostate} and Proposition \ref{propnon-i.i.d.}, we have that
\bbb
    \left\| \oneN^{n}-\mathcal{N}_{\rho_\G,\lambda}\right\|_{\diamond}\ \texteq{\eqref{eq:cascade}}\ \left\| \mathcal{N}_{\rho(\lambda,n),\, \lambda} - \mathcal{N}_{\rho_\G,\lambda}\right\|_{\diamond}\  \textleq{\eqref{diamond as trace}}\ \|\rho(\lambda, n) - \rho_\G\|_1\ \texteq{\eqref{fibre states 1-norm}}\ \mathcal{O}_{M_3'}\left(n^{-1/(2m+2)}\right) ,
\eee
concluding the proof.
\end{proof}

\begin{proof}[Proof of Corollary \ref{cor:capac}]
We now move on to Corollary \ref{cor:capac}. Let us start by proving the statement on quantum capacities, namely \eqref{quantum capacity cascade} and \eqref{remainder term cascade}. Our aim is to apply \cite[Theorem~9]{winter2017energy} to the two channels $\oneN^n$ and $\mathcal{N}_{\rho_\G,\lambda}=\att$, for the special case $m=1$ (cf.~\eqref{thermal attenuator}). We set 
\bbb
\varepsilon_n \coloneqq \frac12 \left\|\oneN^n - \att\right\|_\diamond \geq \frac12 \left\|\oneN^n - \att\right\|_{\diamond E}\, ,
\eee
where the energy-constrained diamond norm is defined with respect to the canonical Hamiltonian, namely the number operator $H=a^\dag a$ (see \cite[Eq.~(2)]{shirokov2018energy} and \cite[Eq.~(2)]{winter2017energy}). Note that $\varepsilon_n = \mathcal{O}_{M_3'}\left(n^{-1/4}\right)$ by Theorem \ref{thm:approxchann}.

The input--output energy relations can be easily determined for both channels thanks to \eqref{rho lambda n induction} and \eqref{energy convolution}, which together show that
$\tr\left[ \rho(\lambda, n)\, a^\dag a \right] = \tr\left[\rho\, a^\dag a\right] = N = \tr\left[ \tau_N a^\dag a\right]$.
One obtains that
\bb
\begin{aligned}
\tr\left[ \oneN^{n}(\omega)\, a^\dag a \right] &= \tr\left[ \mathcal{N}_{\rho(\lambda, n),\,\lambda}(\omega)\, a^\dag a \right] = \lambda \tr\left[ \omega\, a^\dag a \right] + (1-\lambda) N\, , \\
\tr\left[ \att (\omega)\, a^\dag a \right] &= \lambda \tr\left[ \omega\, a^\dag a \right] + (1-\lambda) N\, .
\end{aligned}
\label{input-output energy relations}
\ee
This means that we can set $\alpha=\lambda$ and $E_0=(1-\lambda)N$, and hence $\widetilde{E}=\lambda E+(1-\lambda)N$, in \cite[Theorem~9]{winter2017energy}. We obtain that
\begin{align*}
\left|\mathcal{Q}\big(\oneN^n, E\big) - \mathcal{Q}\left(\att, E\right) \right| &\textleq{1} 56 \sqrt{\varepsilon_n}\, g\left( \frac{4(\lambda E+(1-\lambda)N)}{\sqrt{\varepsilon_n}} \right) + 6\,g\left( 4\sqrt{\varepsilon_n}\right) \\
&\textleq{2} c\, \sqrt{\varepsilon_n} \log \frac{1}{\varepsilon_n} \\
&\textleq{3} C(M_3')\, \frac{\log n}{n^{1/8}}\, .
\end{align*}
Here, in step 1 we applied \cite[Theorem~9]{winter2017energy} together with the formula $S(\tau_N)=g(N)$ (see \eqref{tau maximises entropy} and \eqref{g}); the inequality in 2 holds eventually in $n$ for some universal constant $c \leq 57 + 24 \log e$, as can be seen by combining the bounds $g(x)\leq \log (x+1)+\log e$ (tight for large $x$) and $g(x)\leq -2x\log x$ (valid for sufficiently small $x$); finally, in 3 we used the fact that $\varepsilon_n\leq C(M_3')\, n^{-1/4}$ eventually in $n$ by the already proven Theorem \ref{thm:approxchann}, together with the observation that $x\mapsto -x\log x$ is an increasing function for sufficiently small $x>0$. 

To complete the first part of the proof we need to estimate the classical capacity of $\oneN^n$ in terms of that of the thermal attenuator $\att$ of \eqref{thermal attenuator}, in turn given by \eqref{classical capacity thermal attenuator}. Although we could use the estimates in \cite{winter2017energy}, we prefer to resort to the tighter ones provided in \cite{shirokov2018energy}. We obtain that
\begin{align*}
\left|\mathcal{C}\big(\oneN^n, E\big) - \mathcal{C}\left(\att, E\right) \right| &\textleq{4} 7\varepsilon_n \left( \log \left(\lambda E\!+\! (1\!-\!\lambda)N\!+\! 1\right) - \log \frac{\varepsilon_n}{2} + 1\right) + 2 g\left( \frac{5\varepsilon_n}{2} \right) + 4 h_2\left( \frac{\varepsilon_n}{2} \right) \\
&\textleq{5} c' \varepsilon_n \log \frac{1}{\varepsilon_n} \textleq{6} C'(M_3')\, \frac{\log n}{n^{1/4}}\, .
\end{align*}
The inequality in 4 is an application of \cite[Proposition~6]{shirokov2018energy}. To see why, let us re-write the result of Shirokov \cite[Proposition~6]{shirokov2018energy} for one-mode channels and with respect to the canonical Hamiltonian as
\bbb
\left|\mathcal{C}\big(\mathcal{N}_1, E\big) - \mathcal{C}\left(\mathcal{N}_2, E\right) \right| \leq 2\epsilon \left(2t+r_\epsilon(t)\right) (\log (E'+1) + 1 - \log (\epsilon t)) + 2 g\left(\epsilon r_\epsilon(t)\right) + 4h_2(\epsilon t)\, .
\eee
Here, $\mathcal{N}_i$ ($i=1,2$) are two quantum channels with $\frac12 \left\|\mathcal{N}_1-\mathcal{N}_2\right\|_{\diamond E}\leq \epsilon$, we picked $E'$ such that $\sup_{\rho:\, \tr[\rho\, a^\dag a]\leq E} \tr\left[\mathcal{N}_i(\rho)\, a^\dag a\right]\leq E'$, the function $r_\epsilon$ is defined by $r_\epsilon(t)\coloneqq\frac{1+t/2}{1-\epsilon t}$, and $h_2(x)\coloneqq -x \log x - (1-x)\log (1-x)$ is the binary entropy. Setting $\mathcal{N}_1=\oneN^n$, $\mathcal{N}_2=\att$, we see that $E'=\lambda E+(1-\lambda) N$ (cf.~\eqref{input-output energy relations} and~\cite[Eq.~(21)]{shirokov2018energy}); choosing $t=1/2$ and hence $r_\varepsilon(t) \leq r_1(t) = r_1(1/2) = 5/2$ yields the above relation 4, as claimed. The inequality in 5 holds for all sufficiently large $n$ and for some absolute constant $c'\leq 15$. Finally, 6 is analogous to 3 above.
\end{proof}

\begin{rem*}
Let us stress that the threshold in $n$ above which the inequalities in the above proof hold true depends on both $\lambda E+(1-\lambda)N$ and $M_3'$ (which dictates the rate of convergence of $\varepsilon_n\to 0$). Although this is a minor point from the point of view of the mathematical derivation, it may be important for applications.
\end{rem*}

\begin{rem*}
An analytical formula for the quantum capacity of the thermal attenuator that appears in Corollary \ref{cor:capac} is currently not known. The best lower bound to date reads~\cite[Eq.~(9)]{Noh2020}
\bb
\begin{aligned}
\mathcal{Q}\left(\att, E\right) \geq \max_{0\leq x\leq 1} x \bigg\{ g\left( \lambda\, \frac{E}{x} + (1-\lambda) N \right) &- g\left(\frac12 \left( D_{\lambda,N}(x) + (1-\lambda)\left(\frac{E}{x} - N \right) - 1 \right) \right) \\
&- g\left(\frac12 \left( D_{\lambda,N}(x) - (1-\lambda)\left(\frac{E}{x} - N \right) - 1 \right) \right) \bigg\}\, ,
\end{aligned}
\label{lower bound Q th att 1}
\ee
where
\bb
D_{\lambda, N}(x) \coloneqq \sqrt{\left( (1+\lambda)\frac{E}{x} + (1-\lambda)N + 1 \right)^2 - 4\lambda\, \frac{E}{x}\left(\frac{E}{x} +1\right)}\, .
\label{lower bound Q th att 2}
\ee
The best upper bound to date, instead, can be obtained by combining the results of~\cite[Eq.~(23)--(25)]{PLOB} (see also~\cite[Section~8]{MMMM}) with those of \cite[Theorem~9]{Noh2019} and~\cite[Theorem~46]{Sharma2018}, in turn derived by refining a technique introduced in~\cite{Rosati2018}. We look at the case where $\lambda \geq \frac{N+1/2}{N+1}$, because below that value of $\lambda$ the channel $\att$ becomes 2-extendable \cite{extendibility} (that is, anti-degradable~\cite{Devetak-Shor, Caruso2006, Wolf2007}) and therefore $\mathcal{Q}\left(\att, E\right)=0$.
\begin{align}
\mathcal{Q}\left(\att, E\right) \leq&\ \max\left\{ F_1(N,\lambda; E),\, F_2(N,\lambda; E) \right\} , \label{upper bound Q th att 1} \\[1ex]
F_1(N,\lambda, E) \coloneqq&\ g\left( \lambda E + (1-\lambda)N\right) - g\left( \frac{(1-\lambda)(N+1)}{\lambda - (1-\lambda)N} \left( \lambda E + (1-\lambda)N \right) \right) , \label{upper bound Q th att 2} \\[1ex]
F_2(N,\lambda; E) \coloneqq&\ - \log \left( (1-\lambda) \lambda^N\right) - g(N)\, . \label{upper bound Q th att 3}
\end{align}
\end{rem*}

\section*{Acknowledgements}
ND would like to thank M.~Jabbour for helpful discussions. LL acknowledges financial support from the European Research Council under the Starting Grant GQCOP (grant no.~637352) and from Universit\"{a}t Ulm; he is also grateful to V.~Giovannetti, A.~Holevo and K.~Sabapathy for discussions on Lemma \ref{positive W lemma}, and to M.B.~Plenio and M.~Wilde for sharing their thoughts on our model of optical fibre. SB thanks G.~Baverez for interesting discussions on stable laws and gratefully acknowledges support by the EPSRC grant EP/L016516/1 for the University of  Cambridge  CDT, the CCA. CR acknowledges financial support from the TUM university Foundation Fellowship and by the DFG cluster of excellence 2111 (Munich Center for Quantum Science and Technology).

\begin{appendix}

\section{Standard moments vs phase space moments: the integer case} \label{app:moments}

In this appendix we prove that a state with finite standard moments of order up to $k$ also has finite phase space moments of order up to $k$, i.e.~Theorem \ref{theo:moments}. More precisely, we show that its characteristic function is differentiable $k$ times, and that there are constants $c_{k,m}(\varepsilon)<\infty$ such that the standard moments and phase space moments, defined by \eqref{eq:moments} and \eqref{eq:ps-moments}, respectively, satisfy $M'_k(\rho,\varepsilon)\leq c_{k,m}(\varepsilon) M_k(\rho)$ for all $m$-mode quantum states $\rho$. We start with the following lemma.

\begin{lemma} \label{messy inequality lemma}
For all positive integers $m$ and real numbers $k\in [0,\infty)$, there is a universal constant $d_{k,m}>0$ such that
\bb
\left( \sumno_j a_j a_j^\dag \right)^{k/2} \geq d_{k,m} \sum_j\left(|x_j|^k + |p_j|^k\right) ,
\label{number vs xp}
\ee
where $x_j\coloneqq (a_j+a^\dag_j)/\sqrt2$ and $p_j\coloneqq -i (a_j - a^\dag_j)/\sqrt2$ are the position and momentum quadratures of the $j^{\text{th}}$ mode. 
\end{lemma}

\begin{proof}
First of all, it suffices to consider the one-mode case. Indeed, assume that $(a a^\dag)^{k/2}\geq d_k (|x|^k + |p|^k)$ for some $d_k>0$. Then, leveraging the fact that the operators $a_j a_j^\dag$ commute with each other, and employing standard inequalities between $p$-norms, we deduce that
\bbb
\left( \sumno_j a_j a_j^\dag \right)^{k/2} \geq m^{{\min}\{k/2-1,\,0\}} \sumno_j (a_j a_j^\dag)^{k/2} \geq m^{{\min}\{k/2-1,0\}} d_k \sumno_j \left(|x_j|^k + |p_j|^k\right).
\eee
Therefore, from now on we look at the one-mode case only. The vector space $V_m\coloneqq \operatorname{span}\{ \ket{n}: n \in \mathbb{N}_0 \}$ of states with a finite expansion in the Fock basis is a core for both $\left( a a^\dag \right)^{k/2}$, as well as $|x|^k$ and $|p|^k$. Thus, it suffices for us to prove the inequality \eqref{number vs xp} on states in $V_m$.

It is enough to show that $(a a^\dag)^{k/2} \geq d_k |x|^k$ for some constants $d_k>0$, as the other inequality $(a a^\dag)^{k/2} \geq d_k |p|^k$ is obtained by performing a phase space rotation of an angle $\pi/2$, i.e.\ by conjugating both sides by the unitary operator $e^{i\frac{\pi}{2}\, a a^\dag}$.

We now prove that the inequality $(a a^\dag)^{k/2}\geq d_k |x|^k$ holds for some $d_k>0$ on all vectors in $V_m$. Write $k = 2r h$, where $r\in (0,1]$ and $h \coloneqq \ceil{k/2} \in \NN_0$. Since the function $A\mapsto A^r$ is well known to be operator monotone \cite[Proposition~10.14]{schmuedgen}, it suffices to show that $(a a^\dag)^{h}\geq d_h x^{2h}$ for all non-negative integers $h\in \NN_0$. To this end, let us take advantage of our restriction to states with a finite expansion in the Fock basis. Defining $\Pi_N$ as the projector onto the span of the first $N+1$ Fock states (from $0$ to $N$), we have to show that
\bbb
A_N\coloneqq \sum_{n=0}^N (1+n)^{h} \ketbra{n} - d_h \Pi_N x^{2h} \Pi_N^\dag = \Pi_N \left( (1+a^\dag a)^{h} - d_h x^{2h}\right) \Pi_N^\dag \geq 0\, ,
\eee
where the inequality now involves only matrices. Thanks to Gershgorin's circle theorem~\cite{Gershgorin, VARGA}, in order to show that $A_N$ is positive semi-definite, it suffices to prove that $A_N$ is diagonally dominant, i.e. that~for all $N$ and $0\leq n\leq N$ the inequality
\bb
\braket{n|A_N|n} - \sum_{\substack{n'=0,\ldots,N,\\ n'\neq n}} |\braket{n|A_N|n'}|\geq 0
\label{Gershgorin condition}
\ee
holds true. Writing down the left-hand side yields
\begin{align*}
&\braket{n|A_N|n} - \sum_{\substack{n'=0,\ldots,N,\\ n'\neq n}} |\braket{n|A_N|n'}| \\
&\qquad \ge(1+n)^h - d_h \sum_{n'=0}^N \left|\braket{n|x^{2h}|n'}\right| \\
&\qquad \textgeq{1} (1+n)^h - d_h \sum_{n':\, |n-n'|\leq 2h} \left|\braket{n|x^{2h}|n'}\right| \\
&\qquad = (1+n)^h - d_h \sum_{n':\, n\leq n'\leq n+2h} \left|\braket{n|x^{2h}|n'}\right|  - d_h \sum_{n':\, \max\{n-2h,0\}\leq n'< n} \left|\braket{n|x^{2h}|n'}\right| \\
&\qquad \textgeq{2} (1+n)^h - d_h \sum_{n':\, n\leq n'\leq n+2h} \sum_{\ell=0}^{h-(n'-n)/2} |f_{h,\ell,n'-n}| \left|\braket{n|(a^\dag a)^{\ell} a^{n'-n} |n'}\right| \\
&\qquad \qquad - d_h \sum_{n':\, \max\{n-2h,0\}\leq n'< n} \sum_{\ell=0}^{h-(n-n')/2} |f'_{h,\ell,n-n'}| \left|\braket{n|(a^\dag a)^{\ell} (a^\dag)^{n-n'} |n'}\right| \\
&\qquad = (1+n)^h - d_h \sum_{n':\, n\leq n'\leq n+2h} \sum_{\ell=0}^{h-(n'-n)/2} |f_{h,\ell,n'-n}|\, n^{\ell} \sqrt{\frac{n'!}{n!}} \\
&\qquad \qquad - d_h \sum_{n':\, \max\{n-2h,0\}\leq n'< n} \sum_{\ell=0}^{h-(n-n')/2} |f'_{h,\ell,n-n'}|\, n^\ell \sqrt{\frac{n!}{n'!}} \\
&\qquad \textgeq{3} (1+n)^h - d_h F_h (n+2h)^h \\
&\qquad \textgeq{4} 0\, .
\end{align*}
Here, in 1 we extended the sum over $n'$ to all values that yield a non-vanishing result, i.e.\ those that satisfy $|n-n'|\leq 2h$. In 2 we used the canonical commutation relations \eqref{CCR} to expand 
\bbb
x^{2h} = 2^{-h} (a+a^\dag)^{2h}=\sum_{\ell,\ell'\geq 0,\, \ell + \ell'\leq 2h} \left( f_{h,\ell,\ell'} (a^\dag a)^{\ell} a^{\ell'} + f'_{h,\ell,\ell'} (a^\dag a)^{\ell} (a^\dag)^{\ell'} \right) .
\eee
In 3 we applied standard estimates for factorials: for example, when $n\leq n'\leq n+2h$ we used the fact that $n^\ell \sqrt{\frac{n'!}{n!}}\leq n^\ell (n')^{(n'-n)/2}\leq (n')^{\ell+(n'-n)/2}\leq (n')^{2h}\leq (n+2h)^{2h}$; moreover, we defined $F_h\coloneqq \sum_{\ell,\ell'\geq 0,\, \ell + \ell'\leq 2h} \left(|f_{h,\ell,\ell'}|+ |f'_{h,\ell,\ell'}|\right)$. Finally, 4 follows by choosing e.g.\ $d_h^{-1}=(2h+1)^{h} F_h$. Since \eqref{Gershgorin condition} holds for all $n$, we conclude that $A_N\geq 0$ for all $N$, which completes the proof.
\end{proof}

\begin{rem*}
The inequality in Lemma \ref{messy inequality lemma} depends critically on the special properties of the canonical operators. In fact, there is no universal constant $d_k>0$ that makes the general relation $(A+B)^k\geq d_k\left(A^k+B^k\right)$ true for all positive matrices $A,B\geq 0$. To see why this is the case, it suffices to consider two pure states $A=\ketbra{\psi}$ and $B=\ketbra{\phi}$. Setting $\lambda\coloneqq 1 - |\!\braket{\psi|\phi}\!|\in [0,1]$, it can be shown that the minimal eigenvalue of $(A+B)^k$ is $\lambda^k$, while that of $A^k+B^k=A+B$ is clearly $\lambda$. By Weyl's principle, the conjectured matrix inequality would imply that $\lambda^k\geq d_k \lambda$ for all $\lambda\in [0,1]$, absurd.
\end{rem*}

\begin{prop}
Let $k\geq 0$ and $m\geq 1$ be integers; also, let $\varepsilon>0$ be given. Then, there is a constant $c_{k,m}(\varepsilon)<\infty$ such that every $m$-mode quantum state $\rho$ with finite $k$-moments $M_k(\rho)$, as defined by \eqref{eq:moments}, also satisfies 
\bb
M'_{k}(\rho,\varepsilon) = \left\|\chi_\rho \right\|_{C^{k}\left(B(0,\varepsilon)\right)} \leq c_{k,m}(\varepsilon) M_k(\rho)\, .
\ee
In particular, according to \eqref{eq:ps-moments}, $\rho$ has finite phase space moments of order up to $k$.
\label{prop:Ludovico}
\end{prop}

\begin{proof}
Let $\rho$ be an $m$-mode quantum state. We start by considering the modified state $\sigma\coloneqq \rho \boxplus \ketbra{0}$ that is obtained by convolving it with the (multi-mode) vacuum state according to the rule \eqref{boxplus} (for $\lambda=1/2$). A first important observation is that the moments of $\sigma$ and $\rho$ are related. Namely, 
\bb
M_k(\sigma)\leq M_k(\rho)\qquad \forall\ k\in [0,\infty)\, .
\label{Mk sigma vs rho}
\ee
To see why, we pick a multi-index $n\in \NN_0^m$ and evaluate the $n^{\text{th}}$ diagonal entries of $\sigma$ with respect to the Fock basis. We obtain that
\begin{align*}
    \braket{n|\sigma|n} &\texteq{1} \braket{n|\Delta(\rho\boxplus \ketbra{0})|n} \\
    &\texteq{2} \braket{n|\Delta(\rho)\boxplus \ketbra{0}|n} \\
    &= \sum_{\ell\in \NN_0^m} \braket{\ell|\rho|\ell} \braket{n|\left( \ketbra{\ell}\boxplus \ketbra{0}\right) |n} \\
    &\texteq{3} \sum_{\ell\geq n} \braket{\ell|\rho|\ell} 2^{-|\ell|} \binom{\ell}{n} \, .
\end{align*}
Here, in~1 we introduced the dephasing operator in the Fock basis, whose action is given by $\Delta(X)\coloneqq \sum_{k\in \NN_0^m} \ketbra{k} X \ketbra{k}$. In~2 we observed that $\Delta(\omega \boxplus \delta) = \Delta(\omega)\boxplus \delta$ for all $m$-mode quantum states $\omega$ whenever $\delta=\Delta(\delta)$ is already diagonal in the Fock basis. To show this, first exploit linearity and factorisation of $\Delta$ to reduce to the one-mode case. Then, use the representation $\Delta(X)=\int_0^{2\pi} \frac{d\varphi}{2\pi}\, e^{i \varphi\, a^\dag a} X\, e^{-i \varphi a^\dag a}$, valid for bounded $X$ and where the integrals are as usual weakly converging, and remember that $e^{i\varphi \left(a^\dagger a+b^\dag b\right)}=e^{i\varphi a^\dagger a} \otimes e^{i\varphi b^\dag b}$ is a function of the total Hamiltonian and thus commutes with the action of the beam splitter. The identity in~3 follows from the formula 
\bb
\ketbra{\ell}\boxplus \ketbra{0} = 2^{-|\ell|} \sum_{\ell'\leq \ell} \binom{\ell}{\ell'} \ketbra{\ell'}
\ee
for the convolution of a Fock state with the vacuum. Here, $\ell,\ell'\in \NN_0^m$ are multi-indices, ordered entry-wise, and $\binom{\ell}{\ell'}\coloneqq \prod_j \binom{\ell_j}{\ell'_j}$. The above expression can be obtained easily e.g.\ by first reducing to the one-mode case, and then by induction on $\ell$, employing the relations \eqref{BS action creation}. Computing the $k^{\text{th}}$ moment of $\sigma$ then yields
\begin{align*}
    M_k(\sigma) &\texteq{4} \sum_{n\in \NN_0^m} (m+|n|)^{k/2} \braket{n|\sigma|n} \\
    &\texteq{5} \sum_{\ell\in \NN_0^m} \braket{\ell|\rho|\ell} \sum_{n\leq \ell} (m+|n|)^{k/2} 2^{-|\ell|} \binom{\ell}{n} \\
    &\textleq{6} \sum_{\ell\in \NN_0^m} \braket{\ell|\rho|\ell} (m+|\ell|)^{k/2} \\
    &\texteq{7} M_k(\rho)\, .
\end{align*}
Here, 4~and~7 follow from the representation in \eqref{eq:moments-Fock}; in~5 we rearranged a double series of non-negative terms, and in~6 we observed that for a given $\ell\in \NN_0^m$ the coefficients $P_\ell(n) \coloneqq 2^{-|\ell|} \binom{\ell}{n}$ form a probability distribution over the set of multi-indices $n\in \NN_0^m$ with $n\leq \ell$. This proves that the $k^{\text{th}}$ moments of $\sigma$ are upper bounded by those of $\rho$.

The state $\sigma$ is also useful because its characteristic function is a close relative of that of $\rho$. Namely, according to \eqref{boxplus characteristic functions} we have that $\chi_\sigma (z) = \chi_\rho(z/\sqrt2)\, e^{-\|z\|^2/4}$, and hence 
\bb
\left\|\chi_\rho \right\|_{C^{k}\left(B(0,\varepsilon)\right)}\leq g_{k,m}(\varepsilon) \left\|\chi_\sigma \right\|_{C^{k}\left(B(0,\varepsilon)\right)}
\label{Ck norm chi sigma vs chi rho}
\ee
for some constants $g_{k,m}(\varepsilon)$. Thus, it suffices to find a suitable upper estimate for the norm $\left\|\chi_\sigma \right\|_{C^{k}\left(B(0,\varepsilon)\right)}$. By Lemma \ref{positive W lemma}, the Fourier transform of $\chi_\sigma$, i.e.\ the Wigner function $W_\sigma$ of $\sigma$, is everywhere non-negative. Hence, $\chi_\sigma$ can be seen as the characteristic function of a classical random variable $Z$ over $\CC^m$, with probability density function $W_\sigma$. If we show that $Z$ has finite absolute moments of order $k$, then thanks to \cite[Theorem~1.8.15]{Ushakov} we deduce that $\chi_\sigma$ is $k$-fold differentiable everywhere, and since
\bb
\left|\partial_z^\alpha \partial_{z^*}^\beta \chi_\sigma(z)\right| = \left| \int d^{2m}u\, \Big(\prod\nolimits_j (-u_j)^{\beta_j} (u_j^*)^{\alpha_j} \Big) W_\sigma(u)\, e^{z^\intercal u^* - z^\dag u} \right| \leq \int d^{2m}u\, \|u\|^{|\alpha|+|\beta|}\, W_\sigma(u)
\label{eq:interpolation}
\ee
for all multi-indices $\alpha,\beta\in \NN_0^m$, we in fact have that
\bb
\left\|\chi_\sigma \right\|_{C^{k}\left(B(0,\varepsilon)\right)}\leq 1+\int d^{2m}u\, \|u\|^{k}\, W_\sigma(u) \eqqcolon 1+ L_k(\sigma)\, .
\label{upper bound Ck norm chi sigma}
\ee

Therefore, we now look at the quantity $L_k(\sigma)$. For a vector $u=u_R+ i u_I\in \CC^m$, with $u_R,u_I\in \RR^m$, we observe that 
\bbb
\|u\|^k = \left(\sumno_j (u_{Rj}^2 + u_{Ij}^2)\right)^{k/2} \leq (2m)^{\max\left\{\frac{k}{2}-1,\, 0 \right\}} \left(\sumno_j \left(|u_{Rj}|^k + |u_{Ij}|^k\right)\right) .
\eee
Thus,
\begin{align*}
L_k(\sigma) &\leq (2m)^{\max\left\{\frac{k}{2}-1,\, 0 \right\}} \int d^{m}u_R\, d^m u_I\, \sum_j \left(|u_{Rj}|^k + |u_{Ij}|^k\right) W_\sigma(u) \\
&= (2m)^{\max\left\{\frac{k}{2}-1,\, 0 \right\}} \bigg( \int d^{m}u_R\, \left(\sumno_j |u_{Rj}|^k\right) \int d^m u_I\, W_\sigma(u) \\
&\qquad + \int d^{m}u_I\, \left(\sumno_j |u_{Ij}|^k\right) \int d^m u_R\, W_\sigma(u) \bigg) \\
&\texteq{8} (2m)^{\max\left\{\frac{k}{2}-1,\, 0 \right\}}\, 2^{-\frac12 (m+k-1)} \left( \tr\left[\sigma\, \left( \sumno_j |x_j|^k\right)\right] + \tr\left[\sigma\, \left( \sumno_j |p_j|^k\right)\right] \right) \\
&\textleq{9} (2m)^{\max\left\{\frac{k}{2}-1,\, 0 \right\}}\, 2^{-\frac12 (m+k-1)} d_{k,m}^{-1} \tr\left[\sigma \left( \sumno_j a_j a_j^\dag \right)^{k/2}\right] \\
&\textleq{10} c'_{k,m}\, M_k(\sigma)\, .
\end{align*}
In the above derivation, the identity in 8 can be verified by first reducing to the case of a pure $\sigma$, which can be done by linearity and by multiple applications of Tonelli's theorem, and by subsequently remembering that for a pure state $\ket{\psi_f}$ with wave function $f\in L^2(\RR^m)$ it holds e.g.\ that $\int d^m u_I\, W_{\ket{\psi_f}\bra{\psi_f}}(u) = \sqrt2 \left| f(\sqrt2\, u_R)\right|^2$. The inequality in 9 is just an application of Lemma \ref{messy inequality lemma}. Finally, in 10 we introduced a suitable constant $c'_{k,m}\geq 1$.

Combining the above estimate with \eqref{Ck norm chi sigma vs chi rho}, \eqref{upper bound Ck norm chi sigma}, and \eqref{Mk sigma vs rho}, we deduce that
\begin{align*}
M_k(\rho,\varepsilon) &= \left\|\chi_\rho \right\|_{C^{k}\left(B(0,\varepsilon)\right)}\leq g_{k,m}(\varepsilon) \left\|\chi_\sigma \right\|_{C^{k}\left(B(0,\varepsilon)\right)} \leq g_{k,m}(\varepsilon) (1+L_k(\sigma))\\
&\leq 2\,g_{k,m}(\varepsilon) c'_{k,m}\, M_k(\sigma) \leq 2\,g_{k,m}(\varepsilon) c'_{k,m} M_k(\rho) \eqqcolon c_{k,m}(\varepsilon) M_k(\rho)\, ,
\end{align*}
which concludes the proof.
\end{proof}

\section{ Standard moments vs phase space moments: the fractional case}
\label{sec:SMF}
In the last section, we showed that the $k^{\text{th}}$ phase space moment was controlled by the $k^{\text{th}}$ standard moment in the case of an integer constant $k$.

Here, we show that this fact still holds when $k$ is a positive real number by an interpolation argument. In principle, we could conclude this fact from the setting of Proposition \ref{prop:Ludovico}, using that for $L^p(w_0)$ spaces with weight function $w_0$ and $L^p(w_1)$ spaces with weight function $w_1$, the real interpolation spaces \cite[Theorem~5.4.1]{bergh2012interpolation} satisfy
\[ \left(L^p(w_0),L^p(w_1)\right)_{\theta} = L^p(w_{\theta}) \]
where $w_{\theta}:=w_0^{1-\theta}w_1^{\theta}$. This would allow us to extend the estimate in \eqref{eq:interpolation} to fractional powers as well. However, we want to establish the stronger result that shows that the moments themselves naturally induce an interpolating family of normed spaces. That is, we show the following:

\begin{prop}\label{prop:fractional}
Let $\rho$ be an $m$-mode quantum state and $k\ge 0$. If $\tr\left[\rho \left( \sumno_{j} a_ja_j^\dagger  \right)^{k/2}\right]<\infty$, then $\|\chi_\rho\|_{C^k(B(0,\eps))}<\infty$ for some $\eps>0$. Moreover,
\begin{align*}
     \|\chi_\rho\|_{C^k(B(0,\eps))}\le C_\eps\,\tr\left[\rho \left(\sumno_{j} a_ja_j^\dagger \right)^{k/2}\right]\,,
\end{align*}    
for some constant $C_\eps>0$.
\end{prop}

We have seen in Appendix \ref{app:moments} that the map $\rho\mapsto \chi_\rho$ is bounded from $\mathcal{W}^{k,1}(\cH_m)$ to $C^k(B(0,\eps))$ for any $k$ integer. Since the spaces $C^k(B(0,\eps))$ form an interpolation family, meaning that for any $k_0,k_1\in\mathbb{N}_0$ with $k_1:=k_0+1$, $C^{(1-\theta) k_0+\theta k_1}(B(0,\eps))=(C^{k_0}(B(0,\eps)),C^{k_1}(B(0,\eps)))_\theta$, we have from the previously mentioned interpolation method that
\begin{align}\label{eq:interpol}
    \|\chi_\rho\|_{C^{(1-\theta) k_0+\theta k_1}(B(0,\eps))}\le C_\eps\,\|\rho\|_{(\mathcal{W}^{k_0,1}(\cH_m),\mathcal{W}^{k_1,1}(\cH_m))_\theta}\,,
\end{align}
for some positive constant $C_\eps$ that comes from the bounds derived in Section \ref{app:moments} for $k_0$ and $k_1$. It only remains to prove that the interpolated norms $\|\rho\|_{(\mathcal{W}^{k_0,1}(\cH_m),\mathcal{W}^{k_1,1}(\cH_m))_\theta}$ can further be bounded above by $\|\rho\|_{\mathcal{W}^{(1-\theta)k_0+\theta k_1,1}}$. First, we recall a useful technical lemma \cite[Lemma~3.4]{bourin2012unitary}.

\begin{lemma} \label{lemm:nice}
Let $T= \begin{pmatrix} T_{11} & T_{21} \\ T_{21}^\dagger & T_{22} \end{pmatrix}$ be a positive semi-definite trace class operator such that $T_{11}:\CC^d \rightarrow \CC^d$, then
\[  \Vert T_{21} \Vert_1 \le \frac{1}{2} \left(\Vert T_{11} \Vert_1 +    \Vert T_{22} \Vert_1 \right).\]
\end{lemma}

\begin{proof}[Proof of Proposition \ref{prop:fractional}] We provide the proof only for $m=1$, since the general case follows similarly. First, observe that
\begin{align}
\tr\left[\rho (aa^\dagger )^{k/2}\right] = \sum_{n=0}^\infty \braket{n|\rho|n} (n+1)^{k/2}.
\end{align}
First, we restrict attention to states $\rho$ that are orthogonal in the Fock basis. We then write $\Pi_E$ for the spectral projection onto the Fock states of energy at most $E$, that is $\Pi_{\lfloor E \rfloor}:=\sum_{j\le \lfloor E \rfloor}\ketbra{j}$. Next, we fix two parameters $0\le k_0\le k_1$ and introduce the quantity $\gamma_n:= (n+1)^{(k_1-k_0)/2}$, fix a parameter $t>0$, define $N_0(t) \in \mathbb N$ such that 
\[ \forall n \le N_0(t) : \gamma_n \le t^{-1}\text{ and }\gamma_{N_0(t)+1} \ge t^{-1}\,,\]
the two operators 
\[X_0(t):= (I-\Pi_{N_0(t)}) \rho \,(I-\Pi_{N_0(t)}) \ge 0 \text{ and }X_1(t):= \Pi_{N_0(t)}\rho \Pi_{N_0(t)}\ge 0\]
and $\rho\equiv \rho_{\operatorname{diag}}(t):=X_0(t)+X_1(t).$ Using these two operators we start estimating
\begin{equation*}
\begin{split}
K(t,\rho_{\operatorname{diag}}(t)) :=&\  \inf_{\rho_{\operatorname{diag}}(t)=X_0+X_1} \left\lVert X_0 \right\rVert_{\mathcal{W}^{k_0,1}(\cH_1)}+t\left\lVert X_1 \right\rVert_{\mathcal{W}^{k_1,1}(\cH_1)} \\
\le&\  \left\lVert X_0(t) \right\rVert_{\mathcal{W}^{k_0,1}(\cH_1)}+t\left\lVert X_1(t) \right\rVert_{\mathcal{W}^{k_1,1}(\cH_1)}\\
=&\  \left(\sum_{n > N_0(t)}  \langle n \vert  \rho_{\operatorname{diag}}(t) \vert n \rangle (n+1)^{k_0/2}  +  t \sum_{n \le N_0(t)}  \langle n \vert  \rho_{\operatorname{diag}}(t) \vert n \rangle (n+1)^{k_0/2}\gamma_n     \right) \\
=&\ \sum_{n=0}^\infty \langle n \vert \rho_{\operatorname{diag}}(t) \vert n \rangle (n+1)^{k_0/2} \left(\alpha_n + t \gamma_n \beta_n\right)
\end{split}
\end{equation*}
where $\alpha_n=\delta_{n > N_0(t)}$ and $\beta_n =\delta_{n \le N_0(t)}$ with Kronecker delta $\delta$.
Thus, we obtain for the norm $\|\rho\|_{(\mathcal{W}^{k_0,1}(\cH_1),\mathcal{W}^{k_1,1}(\cH_1))_\theta}$ the upper bound
\begin{equation*}
\begin{split}
\|\rho\|_{(\mathcal{W}^{k_0,1}(\cH_1),\mathcal{W}^{k_1,1}(\cH_1))_\theta}
&\le \sup_{t>0}   \sum_{n=0}^\infty \langle n \vert \rho_{\operatorname{diag}}(t) \vert n \rangle (n+1)^{k_0/2} t^{-\theta}\left(\alpha_n + t \gamma_n \beta_n\right).
\end{split}
\end{equation*}
We now recall that for $\gamma_n \le t^{-1}$ we have $\alpha_n=0$ and $\beta_n=1$ such that
\begin{equation*}
\begin{split}
t^{-\theta}\left(\alpha_n + t \gamma_n \beta_n\right)=t^{-\theta} t \gamma_n  =t^{1-\theta} \gamma_n^{1-\theta} \gamma_n^{\theta} \le \gamma_n^{\theta}. 
\end{split}
\end{equation*}
For $\gamma_n > t^{-1}$ we have $\alpha_n=1$ and $\beta_n=0$ such that 
\[t^{-\theta}\left(\alpha_n + t \gamma_n \beta_n\right) =t^{-\theta}\le  \gamma_n^{\theta}.\]
Thus, in either case, we have the estimate
\[ \|\rho\|_{(\mathcal{W}^{k_0,1}(\cH_1),\mathcal{W}^{k_1,1}(\cH_1))_\theta} \le \sum_{n=0}^\infty \langle n \vert \rho \vert n \rangle (n+1)^{((1-\theta)k_0 + \theta k_1)/2}. \]
This shows that for arbitrary density operators 
\[ \|\rho\|_{(\mathcal{W}^{k_0,1}(\cH_1),\mathcal{W}^{k_1,1}(\cH_1))_\theta} \le \Vert \rho \Vert_{\mathcal{W}^{(1-\theta)k_0+\theta k_1,1}(\cH_1)}\,.\]

To extend the bound to a density operator $\rho$ that is not diagonal in the Fock basis, and not only for the diagonal $\rho_{\operatorname{diag}}(t) $, we partition $\rho$ as
\[ \rho= \begin{pmatrix} X_0(t) & \rho_{21}(t)^\dagger \\ \rho_{21}(t) & X_1(t) \end{pmatrix} \]  and a self-adjoint diagonal operator $S^{(k)}(t):=\operatorname{diag}\left(S_1^{(k)}(t),S_2^{(k)}(t) \right)$ where 
$S_1^{(k)}(t):= \Pi_{N_0(t)}(aa^{\dagger})^{k/4}\Pi_{N_0(t)}$ and $S_2^{(k)}(t):=(I-\Pi_{N_0(t)})(aa^{\dagger})^{k/4}(I-\Pi_{N_0(t)}).$ This implies that 
 \[T^{(k)}:=S^{(k)}\rho S^{(k)} = \begin{pmatrix} S_1^{(k)}(t)X_0(t) S_1^{(k)}(t) & \left(S_1^{(k)}(t) \rho_{21}(t) S_2^{(k)}(t)\right)^\dagger \\ S_1^{(k)}(t) \rho_{21}(t) S_2^{(k)}(t) & S_2^{(k)}(t) \rho_{22}(t) S_2^{(k)}(t) \end{pmatrix}.\]
Let then $S_1^{(k)}(t):= \Pi_{N_0(t)}(aa^{\dagger})^{k/4}\Pi_{N_0(t)}$ and $S_2^{(k)}(t):=(I-\Pi_{N_0(t)})(aa^{\dagger})^{k/4}(I-\Pi_{N_0(t)}).$ The previous Lemma \ref{lemm:nice} then shows that 
\[ \Vert S_1^{(k)}(t) \rho_{21} S_2^{(k)}(t)\Vert_1 \le \frac{1}{2} \left( \Vert S_1^{(k)}(t) \rho_{11} S_1^{(k)}(t) \Vert_1 + \Vert S_2^{(k)}(t) \rho_{22} S_2^{(k)}(t) \Vert_1 \right) \,.\]
From here, we examine three cases separately:
\begin{itemize}
\item Case 1: $\Vert T_{11}^{(k_1)} \Vert_1 \ge \Vert T_{22}^{(k_1)} \Vert_1.$
In this case, we find from choosing $X_0:=\rho_{21}$ and $X_1:=0$ in \eqref{eq:K}
\[ K(t,\rho_{21}) \le \left\lVert \rho_{21}(t) \right\rVert_{\mathcal{W}^{k_0,1}(\cH_1)} \le \left\lVert X_0(t) \right\rVert_{\mathcal{W}^{k_0,1}(\cH_1)} . \] 

\item Case 2: $\Vert T_{22}^{(k_0)} \Vert_1 \ge \Vert T_{11}^{(k_0)} \Vert_1$
In this case, we find from choosing $X_0:=0$ and $X_1:=\rho_{21}$ in \eqref{eq:K}
\[ K(t,\rho_{21}) \le t \left\lVert \rho_{21} \right\rVert_{\mathcal{W}^{k_1,1}(\cH_1)} \le t \left\lVert X_1(t) \right\rVert_{\mathcal{W}^{k_1,1}(\cH_1)}. \] 

\item Case 3: $\Vert T_{22}^{(k_0)} \Vert_1 \le \Vert T_{11}^{(k_0)} \Vert_1$ and $\Vert T_{22}^{(k_1)} \Vert_1 \ge \Vert T_{11}^{(k_1)} \Vert_1$
In this case, we find from choosing $X_0=\rho_{21}/2$ and $X_1= \rho_{21}/2$ in \eqref{eq:K} that
\begin{equation}
\begin{split}
 K(t,\rho_{21}) 
 &\le  \frac{ \left\lVert \rho_{21} \right\rVert_{\mathcal{W}^{k_0,1}(\cH_1)} +t  \left\lVert \rho_{21} \right\rVert_{\mathcal{W}^{k_1,1}(\cH_1)}}{2} \\
 &\le \frac{1}{2} \left( \left\lVert X_0(t)  \right\rVert_{\mathcal{W}^{k_0,1}(\cH_1)}+t \left\lVert X_1(t)  \right\rVert_{\mathcal{W}^{k_1,1}(\cH_1)} \right). 
 \end{split}
\end{equation}

\end{itemize}

Hence, we have altogether that
\begin{equation*}
\begin{split}
K(t,\rho(t)) 
&\le K(t,\rho_{\operatorname{diag}}(t)) + 2  K(t,\rho_{21}(t)) \\
&\le 3 \left( \left\lVert X_0(t) \right\rVert_{\mathcal{W}^{k_0,1}(\cH_1)}+t\left\lVert X_1(t) \right\rVert_{\mathcal{W}^{k_1,1}(\cH_1)} \right)
\end{split}
\end{equation*}
which implies that 
\[ \|\rho\|_{(\mathcal{W}^{k_0,1}(\cH_1),\mathcal{W}^{k_1,1}(\cH_1))_\theta}\le 3 \Vert \rho \Vert_{\mathcal{W}^{(1-\theta)k_0+\theta k_1,1}(\cH_1)}\,.\]

The result follows from the interpolation bound (\ref{eq:interpol}).

\end{proof}

\section{Standard moments vs phase space moments: a partial converse}
\label{sec:partconv}

We now show that at least for even integers $k$, the existence of $k^{\text{th}}$ order phase space moments implies the existence of standard moments of the same order.
\begin{thm}
Let $\rho$ be an $m$-mode quantum state such that its characteristic function $\chi_{\rho}$ is $2k$ times totally differentiable at $z=0$ for some integer $k$, then the $2k^{\text{th}}$ standard moment is finite as well.
\label{thm:converse}
\end{thm}
\begin{proof}
For simplicity, we restrict attention to $m=1$. Let $H_{\operatorname{lin}}^{+}\coloneqq a + a^{\dagger}$ and $H_{\operatorname{lin}}^{+}\coloneqq (-i)( a - a^{\dagger})$ be two Hamiltonians, and consider the spectral decomposition of the density operator $\rho = \sum_{i=1}^{\infty} \lambda_i \ketbra{e_i}$. Then, there exist unique probability measures $\mu_{e_i}$ such that 
\bbb
\braket{e_i| f(H_{\operatorname{lin}}^{\pm}) |e_i} =\int_{\sigma(H_{\operatorname{lin}}^{\pm})} f(\lambda) \  d \mu_{e_i}(\lambda) \text{ for all } f \text{ bounded measurable.}
\eee

We then define the new probability measure $\mu_{\rho}\coloneqq \sum_{i=1}^{\infty} \lambda_i \mu_{e_i}$ such that 
\bbb
\tr\left[ \rho\, f(H_{\operatorname{lin}}^{\pm}) \right] = \int_{\sigma(H_{\operatorname{lin}}^{\pm} )} f(\lambda) \  d\mu_{\rho}(\lambda)\text{ for all } f \text{ bounded measurable.}
\eee

We now proceed with an induction argument. Start by noting that for $k=0$ the result holds. For $k \ge 1$, define the auxiliary function $\varphi:\mathbb{R} \rightarrow \mathbb{C}$ as
\bbb
\varphi(t)\coloneqq \tr\left[\rho\,e^{it H_{\operatorname{lin}}^{\pm}}\right] , 
\eee
which is by assumption $2k$ times differentiable at zero and let $u(t)=\Re \varphi(t)$. Then, $u$ is also $2k$ times differentiable at zero. Since $\varphi^{2k}(0)$ exists, for $t \in (-\varepsilon, \varepsilon)$, with sufficiently small $\varepsilon>0$, the function $t \mapsto \varphi^{(2k-1)}(t)$ exists and is continuous.

We record that Taylor's formula implies that for $t \in (-\varepsilon,\varepsilon)$ 
\bbb
\left\lvert u(t)-\sum_{i=0}^{k-1}u^{(2i)}(0) \frac{t^{2i}}{(2i)!} \right\rvert \le \frac{\vert t \vert^{2k-1}}{(2k-1)!} \sup_{\theta \in (0,1]} \left\lvert u^{(2k-1)}(\theta t) \right\rvert,
\eee
where odd derivatives vanish at zero, since $u$ is even.

We then define a positive continuous function $f_k: \RR \rightarrow [0,\infty)$ with $f_k(0)=1$ and for $t\neq 0$ as
\bbb
f_k(t)\coloneqq (-1)^k (2k)!\, t^{-2k} \left(\cos(t)-\sum_{i=0}^{k-1}(-1)^i \frac{t^{2i}}{(2i)!} \right).
\eee
From Taylor's formula above we obtain the following estimate for $t$ sufficiently small
\begin{equation}
\begin{split} \tr\left[\rho\,f_k(tH_{\operatorname{lin}}^{\pm}){H_{\operatorname{lin}}^{\pm}}^{2k}\right] &= \int_{\sigma(H_{\operatorname{lin}}^{\pm})} f_k(t\lambda)\lambda^{2k} d\mu_{\rho}(\lambda)  \\
&= \frac{(-1)^k(2k)!}{t^2k} \left(u(t)-\sum_{i=0}^{k-1}(-1)^i \frac{\tr\left[\rho\,(tH_{\operatorname{lin}}^{\pm})^{2i}\right]}{(2i)!} \right)  \\
&\le 2k \sup_{\theta \in (0,1]} \frac{\vert u^{2k-1}(\theta t) \vert}{ \theta \vert t \vert} \eqqcolon g_k(t)
\end{split}
\end{equation}

Then, we have from Fatou's lemma
\begin{equation}
\begin{split}
\tr\left[\rho\, \big(H_{\operatorname{lin}}^{\pm}\big)^{2k}\right] 
&= \tr\left[\rho\,f_k(0) \big(H_{\operatorname{lin}}^{\pm}\big)^{2k}\right] \\
&\le \liminf_{t \downarrow 0}\int_{\sigma(H_{\operatorname{lin}}^{\pm})} f_k(t\lambda)\lambda^{2k} d\mu_{\rho}(\lambda) \\
&= \liminf_{t \downarrow 0} g_k(t) = 2k \vert u^{2k}(0) \vert < \infty.
\end{split}
\end{equation}
Using integration by parts and standard estimates only, it is straightfroward to verify that the finiteness of both $\tr\left[\rho\,H_{{\operatorname{lin}}^{\pm}}^{2k}\right]$ implies the finiteness of $\tr\left[\rho (aa^{\dagger})^{k}\right]$.
\end{proof}

\end{appendix}

\bibliography{biblio}

\end{document}

%% file: def.tex
\usepackage[latin1]{inputenc}
\usepackage{amsthm}
\usepackage{amssymb}
\usepackage{amsmath}
\usepackage{bbold}
\usepackage{bbm}
\usepackage{nonfloat}
\usepackage{hyperref}
\usepackage{braket}
\usepackage{dsfont}
\usepackage{mathdots}
\usepackage{mathtools}
\usepackage{enumerate}
\usepackage{csquotes}
\usepackage{stmaryrd}
\usepackage[cal=boondox]{mathalfa}
\usepackage{graphicx}
\usepackage{stackengine}
\usepackage{scalerel}
\usepackage[table]{xcolor}
\usepackage{array}
\usepackage{makecell}
\newcolumntype{x}[1]{>{\centering\arraybackslash}p{#1}}
\usepackage{tikz}
\usetikzlibrary{shapes.geometric, arrows, decorations.pathreplacing, angles, quotes}
\usepackage{booktabs}
\usepackage{xfrac}
\usepackage{subcaption}
\usepackage{comment}

\let\savecorresponds\corresponds
\let\corresponds\relax
\let\saveboxplus\boxplus
\let\boxplus\relax
\usepackage{mathabx}
\let\corresponds\savecorresponds
\let\boxplus\saveboxplus

\newtheorem{thm}{Theorem}
\newtheorem*{thm*}{Theorem}
\makeatletter
\newcommand{\setthmtag}[1]{
  \let\oldthethm\thethm
  \newcommand{\thethm}{#1}
  \g@addto@macro\endthm{
    \addtocounter{thm}{-1}
    \global\let\thethm\oldthethm}
  }
\makeatother

\newtheorem{prop}[thm]{Proposition}
\newtheorem*{prop*}{Proposition}
\newtheorem{lemma}[thm]{Lemma}
\newtheorem*{lemma*}{Lemma}
\newtheorem{cor}[thm]{Corollary}
\newtheorem*{cor*}{Corollary}

\newtheorem*{cj*}{Conjecture}
\newtheorem{Def}[thm]{Definition}
\newtheorem*{Def*}{Definition}

\theoremstyle{definition}
\newtheorem{rem}{Remark}
\newtheorem*{rem*}{Remark}

\newtheorem{ex}{Example}

\newcommand{\bb}{\begin{equation}}
\newcommand{\bbb}{\begin{equation*}}
\newcommand{\ee}{\end{equation}}
\newcommand{\eee}{\end{equation*}}
\newcommand*{\coloneqq}{\mathrel{\vcenter{\baselineskip0.5ex \lineskiplimit0pt \hbox{\scriptsize.}\hbox{\scriptsize.}}} =}
\newcommand*{\eqqcolon}{= \mathrel{\vcenter{\baselineskip0.5ex \lineskiplimit0pt \hbox{\scriptsize.}\hbox{\scriptsize.}}}}
\newcommand\floor[1]{\left\lfloor#1\right\rfloor}
\newcommand\ceil[1]{\left\lceil#1\right\rceil}
\newcommand{\texteq}[1]{\stackrel{\mathclap{\scriptsize \mbox{#1}}}{=}}
\newcommand{\textleq}[1]{\stackrel{\mathclap{\scriptsize \mbox{#1}}}{\leq}}

\newcommand{\textgeq}[1]{\stackrel{\mathclap{\scriptsize \mbox{#1}}}{\geq}}
\newcommand{\ketbra}[1]{\ket{#1}\!\bra{#1}}
\newcommand{\ketbraa}[2]{\ket{#1}\!\bra{#2}}
\newcommand{\sumno}{\sum\nolimits}

\newcommand{\G}{\mathrm{\scriptscriptstyle G}}
\newcommand{\smallonen}{\text{\scalebox{0.8}{$1\!/\!n$}}}
\newcommand{\oneN}{\mathcal{N}_{\rho,\, \lambda^{\text{\scalebox{0.8}{$1/n$}}}}}
\newcommand{\att}{\mathcal{E}_{N,\lambda}}
\newcommand{\emp}[1]{{\em{#1}}}

\newcommand{\RR}{{\mathbb{R}}}
\newcommand{\NN}{\mathbb{N}}
\newcommand{\CC}{\mathbb{C}}
\newcommand{\LL}{{L}}
\newcommand{\E}{\mathds{E}}

\DeclareMathAlphabet{\pazocal}{OMS}{zplm}{m}{n}

\DeclareMathOperator{\dom}{Dom}

\newcommand{\lsmatrix}{\left(\begin{smallmatrix}}
\newcommand{\rsmatrix}{\end{smallmatrix}\right)}

\stackMath
\newcommand\xxrightarrow[2][]{\mathrel{%
  \setbox2=\hbox{\stackon{\scriptstyle#1}{\scriptstyle#2}}%
  \stackunder[5pt]{%
    \xrightarrow{\makebox[\dimexpr\wd2\relax]{$\scriptstyle#2$}}%
  }{%
   \scriptstyle#1\,%
  }%
}}

\newcommand{\tends}[2]{\xxrightarrow[\! #2 \!]{\mathrm{#1}}}
\newcommand{\tendsn}[1]{\xxrightarrow[\! n\rightarrow \infty\!]{\mathrm{#1}}}

\def\indic{\operatorname{1\hskip-2.75pt\relax l}}

\def\smallsection#1{\smallskip\noindent\textbf{#1}.}

\stackMath

\makeatletter

\newcommand{\cH}{{\pazocal{H}}}

\newcommand{\cT}{{\pazocal{T}}}
\newcommand{\cB}{{\pazocal{B}}}
\newcommand{\cD}{{\pazocal{D}}}
\newcommand{\cX}{{\pazocal{X}}}
\newcommand{\cY}{{\pazocal{Y}}}

\newcommand{\eps}{{\varepsilon}}

\newcommand{\tr}{\operatorname{Tr}}
\newcommand{\D}{\mathcal{D}}

\usepackage[normalem]{ulem}